\documentclass{article}

\usepackage[dvipdfmx]{graphicx}
\usepackage[margin=3cm]{geometry}

\usepackage{mymacro}

\usepackage{capt-of}
\usepackage{autonum}
\usepackage{authblk}

\usepackage[colorinlistoftodos]{todonotes}
\setuptodonotes{inline, inlinewidth=5cm, noinlinepar}

\title{Persistence-based topological optimization: a survey}
\author[1]{Mathieu Carri\`ere}
\author[2]{Yuichi Ike}
\author[3]{Th\'eo Lacombe}
\author[4]{Naoki Nishikawa}
\affil[1]{\texttt{mathieu.carriere@inria.fr}, DataShape, Centre Inria d'Université Côte d'Azur, Sophia-Antipolis, France}
\affil[2]{\texttt{ike@ms.u-tokyo.ac.jp}, Graduate School of Mathematical Sciences, The University of Tokyo, 3-8-1 Komaba Meguro-ku Tokyo 153-8914, Japan}
\affil[3]{\texttt{theo.lacombe@univ-eiffel.fr}, Laboratoire d'Informatique Gaspard Monge, Université Gustave Eiffel, CNRS, F-77454 Marne-la-Vallée, France}
\affil[4]{\texttt{nishikawa-naoki259@g.ecc.u-tokyo.ac.jp},  Graduate School of Information Science and Technology, The University of Tokyo, 7-3-1 Hongo Bunkyo-ku Tokyo 113-0033, Japan; RIKEN AIP, Nihonbashi 1-chome Mitsui Building, 15th floor,1-4-1 Nihonbashi, Chuo-ku, Tokyo 103-0027, Japan}
\date{}

\graphicspath{{images/}{images/topoae/}}

\begin{document}

\maketitle

\begin{abstract}
Computational topology provides a tool, \emph{persistent homology}, to extract quantitative descriptors from structured objects (images, graphs, point clouds, etc). These descriptors can then be involved in optimization problems, typically as a way to incorporate topological priors or to regularize machine learning models. This is usually achieved by minimizing adequate, topologically-informed losses based on these descriptors, which, in turn,
naturally raises theoretical and practical questions about the possibility of optimizing such loss functions using gradient-based algorithms. This has been an active research field in the topological data analysis community over the last decade, and various techniques have been developed to enable optimization of persistence-based loss functions with gradient descent schemes. 

This survey presents the current state of this field, covering its theoretical foundations, the algorithmic aspects, and showcasing practical uses in several applications. 
It includes a detailed introduction to persistence theory and, as such, aims at being accessible to mathematicians and data scientists newcomers to the field. It is accompanied by an open-source library\footnote{\url{https://github.com/git-westriver/benchmark_ph_optimization}} which implements the different approaches covered in this survey, providing a convenient playground for researchers to get familiar with the field.
\end{abstract}


\paragraph{Acknowledgments.} 
YI and NN are supported by JSPS Grant-in-Aid for Transformative Research Areas (A) Grand Number JP22H05107 and JST, CREST Grant Number JPMJCR24Q1, Japan. 
TL is supported by l’Agence Nationale de la Recherche (ANR) under grant TheATRE ANR-24-CE23-7711.
MC is supported by l’Agence Nationale de la Recherche (ANR) under grants TopModel ANR-23-CE23-0014 and 3IA ANR-23-IACL-0001.

\clearpage 
\tableofcontents

\section{Introduction}

Topological Data Analysis (TDA) is a field of data science that focuses on the characterization, inference, and encoding of topological features (such as connected components, branches, loops, voids, etc.) in structured objects (such as graphs, point clouds, time series, etc.). 
Such topological features have often been shown to carry complementary information to traditional data descriptors in downstream machine learning tasks \cite{liu2016applying,rieck2019persistent,zhao2020persistence,akai2021experimental,verma2024topological}, and to be able to significantly improve models in terms of predicting performance
in a wide range of applications including computer graphics \cite{pascucci2010topological,carriere2015stable,tierny2017topology}, analysis of machine learning models \cite{rieck2019neural,lacombe2021topological,birdal2021intrinsic,andreeva2024topological}, material and molecular science \cite{hiraoka2016hierarchical,saadatfar2017pore,olejniczak2023topological}, or computational biology \cite{bukkuri2021applications,aukerman2022persistent,chung2024morphological} for example.

However, the advent of deep learning has shifted the machine learning paradigm from manually crafting data features to the automatic learning of such features through the optimization of neural network architectures with \emph{gradient descent}. 
Hence, among the different directions driving current TDA research, the ability to \emph{differentiate functions depending on topological features} has recently appeared as an important question for smooth integration of TDA into modern deep learning pipelines. 
This field of research, called {\em persistence-based topological optimization}, has recently seen several breakthroughs, both in terms of theoretical developments and numerical
implementations and experiments. 
In essence, persistence-based topological optimization involves the theoretical derivation of {\em differentiability properties} of topological features built using an algebraic machinery called \emph{persistent homology}, and practical 
implementations of corresponding, well-defined {\em gradients} associated to their computation. 
Note that this concept should not be confused with ``topological optimization'' as used in material science \cite{sigmund2013topology}, which consider the optimization of topological features not based on persistent homology (though persistent homology in increasingly used in that field as well, see \cite{sugai2024data,kii2025data} and \cref{sec:applications} for instance). 
For the sake of concision, we will nonetheless use the terminology \emph{topological optimization} to refer to persistence-based topological optimization in this survey. 
\Cref{fig:exdiff} provides a general overview on the methodology we will present in this survey.

\begin{figure}
	\centering
	\includegraphics[width=.9\textwidth]{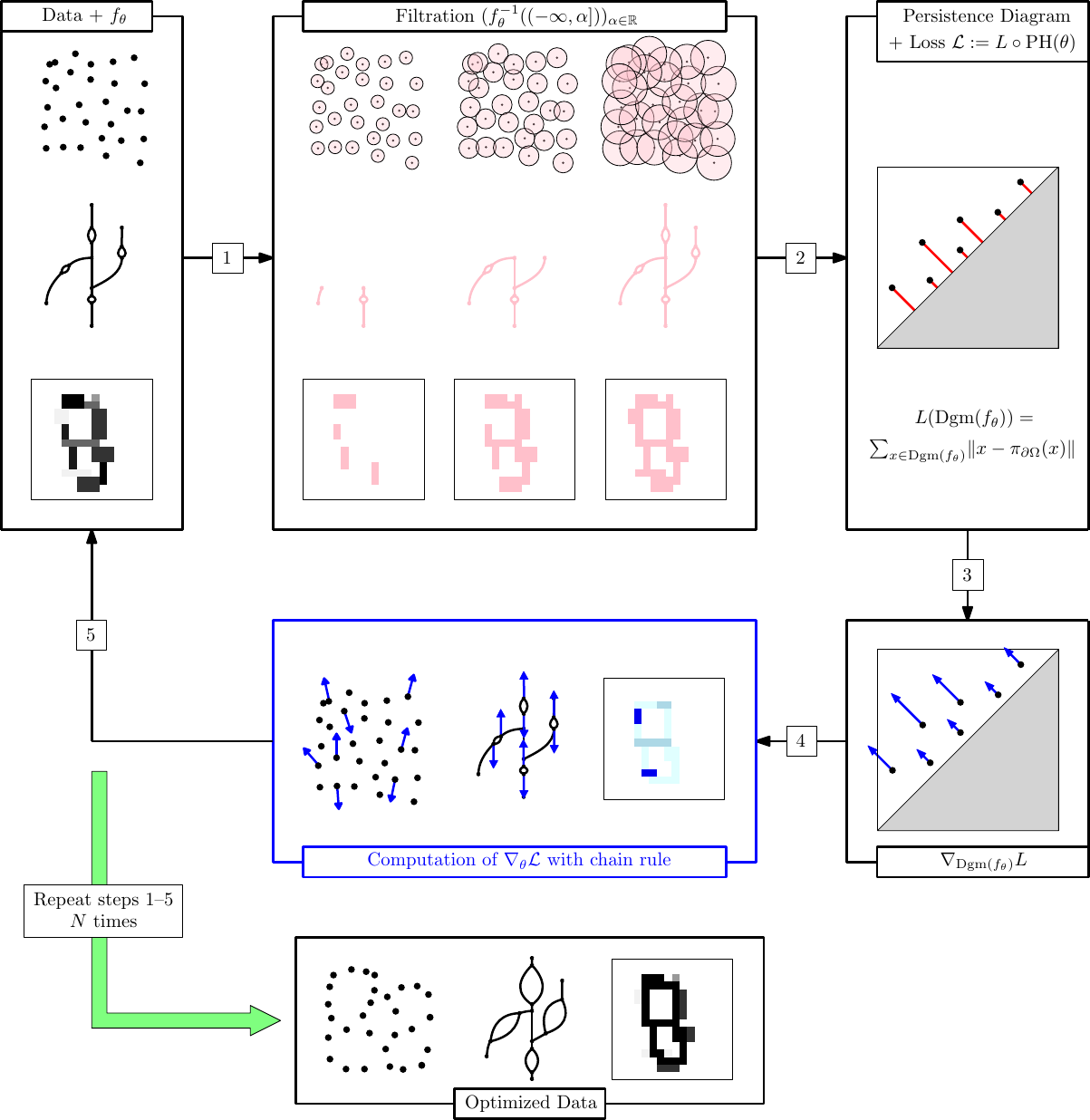}
	\caption{\label{fig:exdiff} Standard topological optimization scheme typically used in deep learning pipelines, with examples on point cloud, graph and image data. Note that step $1$ needs not be at the beginning of the pipeline, and can be incorporated at any stage (this happens when input data are computed as the outputs of some other models, such as, e.g., image filters automatically computed by CNNs). The blue square is the main theoretical question that we discuss in this survey.}
\end{figure}

Such gradients can then be used subsequently in several different tasks based on gradient descent, such as 
filtration learning~\cite{Carriere2021a, pmlr-v119-hofer20b, horn2021topological} or model regularization~\cite{chen2019topological, LCLO2023gradient}.
While vanilla methods for computing these topological gradients have been identified since 2016 thanks to the seminal work of Gameiro, Hiraoka and Obayashi~\cite{gameiro2016continuation},
it quickly appeared that such natural first approaches were limited both in terms of lack of theoretical guarantees and erratic behaviors in numerical experiments.
The difficulty associated to addressing these limitations left these questions open in the TDA community for a few years until several answers were proposed recently, with new developments still appearing as of today.

In this survey, our aim is to present in a unified and consistent framework the different methods, with their theoretical results and algorithms, that are currently available in the literature for performing
topological optimization, as well as to provide a all-in-one-place public library\footnote{\url{https://github.com/git-westriver/benchmark_ph_optimization}} where these different approaches are implemented. 

\paragraph{Outline.} \cref{sec:background} provides a general introduction to Topological Data Analysis, with a focus on the tools required later in this survey. We recall the construction of persistence diagrams (PDs) from filtrations on simplicial complexes, including their algorithmic computation as it will be core to define and compute gradients involving PDs.

\cref{sec:differentiability} is dedicated to the construction of a differential structure for maps valued in and from the space of persistence diagrams, following the work of Leygonie, Oudot, and Tillman \cite{LOT2022}. 
In essence, this construction allows the user to treat the byproduct of persistence homology, a topological descriptor called \emph{persistence diagram} (PD)---when it comes to compute differentials---as vectors in an Euclidean space by identifying them with their \emph{lift}. The main result is that even if several lifts are possible, the choice has no impact when computing the (usual) gradient of composite maps going through the non-linear space of persistence diagrams.  
This formal construction theoretically justifies practical implementations in which persistence diagrams are indeed manipulated as arrays (i.e.,~vectors). 

\cref{sec:gradient_descent} makes use of the now defined gradients to perform (gradient-based) topological optimization. 
We start by reviewing the so-called \emph{vanilla} gradient descent (with fixed or decreasing step size) and its practical limitations, and then present several variations of it that have been recently introduced in the literature and that either strengthen the theoretical guarantees or the numerical efficiency of the vanilla approach. 

\cref{sec:applications} proposes an overview of practical applications involving topological optimization in a wide range of settings: filtration learning, dimensionality reduction, computer vision and regularization of objective functions used in machine learning. 

Eventually, \cref{sec:expes} showcases the different methods presented in \cref{sec:gradient_descent} in practice, in the context of topological optimization for point clouds. 
All methods have been (re)implemented in an all-in-one Python library that may be of its own interest for future research in the field.

\paragraph{Notations.} We summarize below the core notation that will be used throughout the paper:
\begin{itemize}
    \item $\bN = \{0,1,2,\dots \}$ denotes the set of natural numbers including $0$ and $\bR$ denotes the set of real numbers. 
    \item For a finite set $A$, the cardinality of $A$ is denoted by $|A|$.
    \item $\pdspace$ denotes the space of persistence diagrams (PDs), that are finite multisets\footnote{i.e.,~sets where repetition of points is allowed.} supported on the closed extended half-plane $\groundspace \coloneqq \{(b,d) \in \overline{\bR}^2 \,|\, b \leq d\}$ (\Cref{def:PD}). The notation $\pdspace^o$ denotes ordinary persistence diagrams, i.e.,~PDs with points supported on the open half-plane $\{(b,d) \in \bR^2,\  b < d \}$. 
    It is equipped with partial matching metrics denoted by $\FG_q,\ q \in [1,+\infty]$ (\Cref{def:dist_wasserstein}). The empty persistence diagram is denoted by $\varnothing$.
    \item If $x \in \groundspace$ belongs to the support of a persistence diagram $\alpha \in \pdspace$, $m(x)$ denotes its multiplicity.
    \item In equations, $\Pers$ typically denotes a map from a manifold $M$ (typically $M = \bR^d$) valued in $\pdspace$, $L$ typically denotes a map from $\pdspace$ to some manifold $N$ (typically $N = \bR^{d'}$ with $d'=1$), and $\cL$ denotes a composite map $\cL \coloneqq L \circ \Pers \colon M \to N$. 
    \item Given a simplicial complex $K$ (\cref{def:simpcomp}), the set of filtrations on $K$ is denoted by $\Filt_K$ (\cref{def:filtration}). Given $f \in \Filt_K$, $\Dgm(f) \in \pdspace$ denotes the persistence diagram induced by $f$ on $K$ (\Cref{def:filtration}), and $\Dgm_p(f)$ denotes its persistence diagram restricted to homology dimension $p$. 
    \item $\bF_2$ denotes the field with two elements, i.e.,~$\bZ/2 \bZ$.
\end{itemize}

\section{Background on Topological Data Analysis}
\label{sec:background}

The main backbone of TDA is the so-called {\em persistent homology} (PH) theory, which can be used to define the main
TDA descriptors called {\em persistence diagrams} (PDs).
Hence, in this section, we briefly recall the basics associated to the construction of PDs from PH. We
will then discuss their differentiability properties in the following sections. 

\subsection{Simplicial complexes, homology groups and filtrations}\label{subsec:background}

Generally speaking, the main difficulty for defining quantitative topological descriptors comes from the fact that topological
information is usually only well-defined for continuous spaces, as opposed to discrete data (such as point clouds, graphs, meshes, etc.) that one 
has to deal with in data science. 
As a solution to this issue, the so-called {\em simplicial complexes} are the most
basic bricks of TDA, as they are formal objects that can be handled by computers thanks to their combinatorial nature,
yet for which topological information, in the form of {\em homology groups}, can still be defined. 

\begin{definition}\label{def:simpcomp}
    Let $V$ be a finite set. A {\em simplicial complex} whose vertex set is $V$ is a collection $K$ of subsets of $V$ satisfying the following conditions:
    \begin{enumerate}
        \item[(i)] $\varnothing \not\in K$;
        \item[(ii)] for any $v \in V$, $\{v\} \in K$;
        \item[(iii)] if $\sigma \in K$ and $\varnothing \neq \tau \subseteq \sigma$, then $\tau \in K$. 
    \end{enumerate}
    The set $V$ is called the {\em vertex set} of $K$, and is denoted by $V(K)$. 
    Moreover, any $\sigma\in K$ with cardinality $|\sigma|=p+1$ is called a {\em simplex of dimension $p$}, or $p$-simplex for short, and the 
    set of all $p$-simplices, denoted by $\Sk_p(K)$, is called the {\em $p$-skeleton} of $K$.
    The dimension of $K$ is defined as the maximum of the dimension of $\sigma \in K$. 
\end{definition}

\begin{example}
    Graphs (represented as a set of vertices $V$ and a set of edges $E \subset V \times V$) are particular cases of simplicial complexes, namely simplicial complexes of dimension $1$.
\end{example}

We can now define the {\em homology groups} of simplicial complexes.
In essence, these groups aim at encoding the numbers of {\em holes} of the simplicial complexes using linear algebra. Indeed, these holes, or {\em cycles}, can be entirely
characterized as the elements of the kernel of a linear map called the {\em boundary operator}. In words, cycles are linear combinations of simplices, or {\em chains},
whose boundaries are null.

\begin{definition}
Let $K$ be a simplicial complex. 
Let $\bF_2$ be the field consisting of two elements.
The vector space of \emph{$p$-chains} of $K$, denoted by $C_p(K)$, is the $\bF_2$-vector space with basis $\Sk_p(K)$.
The \emph{boundary operator} $\partial_p \colon C_p(K)\to C_{p-1}(K)$ is the linear operator defined on any $p$-simplex $\sigma=\{v_0,v_1,\dots,v_p\}$ as:
\begin{equation}
    \partial_p(\sigma)=\sum_{i=0}^p \{v_0,\dots,v_{i-1},v_{i+1},\dots, v_p\}.
\end{equation}
Finally, we call $\Ker \partial_p$ the vector space of \emph{$p$-cycles} of $K$, and $H_p(K) \coloneqq \Ker \partial_p/\Image \partial_{p+1}$ the \emph{$p$-th homology group}\footnote{Note that $H_1$ is actually a vector space, due to the fact that chains are built on top of the \emph{field} $\bF_2$. See \cref{rem:field}.} of $K$. 
\end{definition}

Considering the quotient of $\Ker \partial_p$ by $\Image \partial_{p+1}$ is motivated by the relation $\partial_p\circ \partial_{p+1}=0$. 
Doing so, non-zero elements of $H_p(K)$ correspond to (equivalence class of) $p$-cycles that are not obtained as boundaries of $(p+1)$-chains.  

\begin{figure}[ht]
    \centering
    \begin{tikzpicture}
            \draw (0,0) node[label=below: {$v_0$}] {} -- (2,0) node[label=below: {$v_1$}] {}; 
            
            \draw (2.5,0) node {$\mapsto$};
            \draw (2.5,0.5) node {$\partial_1$};

            \fill (3.,0) circle [radius=2pt] node[label=below: {$v_0$}] {};
            \fill (5.,0) circle [radius=2pt] node[label=below: {$v_1$}] {};

            \fill[lightgray] (7,-1) -- (8,1) -- (9,-1) --cycle;
            \draw (7,-1) node[label=below: {$v_0$}] {};
            \draw (9,-1) node[label=below: {$v_1$}] {};
            \draw (8,1) node[label=above: {$v_2$}] {};
            
            \draw (9.5,0) node {$\mapsto$};
            \draw (9.5,0.5) node {$\partial_2$};
            
            \draw (10,-1) node[label=below: {$v_0$}] {} -- (12,-1) node[label=below: {$v_1$}] {} -- (11,1) node[label=above: {$v_2$}] {} -- cycle;
    \end{tikzpicture}
    \caption{Boundary operator}
    \label{fig:boundary}
\end{figure}
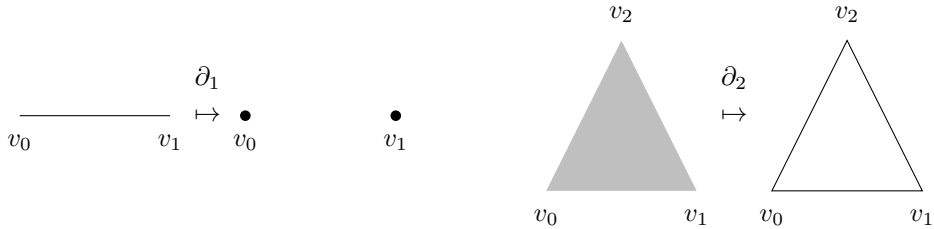

\begin{example} Let $\{v_0,v_1\}$ and $\{v_0,v_1,v_2\}$ be the $1$-simplex and the $2$-simplex displayed in \cref{fig:boundary}. We have 
\begin{equation}
    \partial_1 \{v_0,v_1 \} = \{v_1\} + \{v_0\}
\end{equation}
and 
\begin{equation}
    \partial_2 \{v_0,v_1,v_2\} = \{v_1,v_2\} + \{v_0,v_2\} + \{v_0,v_1\}.
\end{equation}
These imply
\begin{equation}
    \partial_1 \circ \partial_2 (\{v_0,v_1,v_2\}) 
    = 
    (\{v_2\} + \{v_1\}) + (\{v_2\}+\{v_0\}) + (\{v_1\}+\{v_0\})
    = 0,
\end{equation}
since the coefficient field is $\bF_2$.
\end{example}

\begin{example}\label{example:torus_CS}
\begin{figure}
    \centering
    \includegraphics[width=0.3\linewidth]{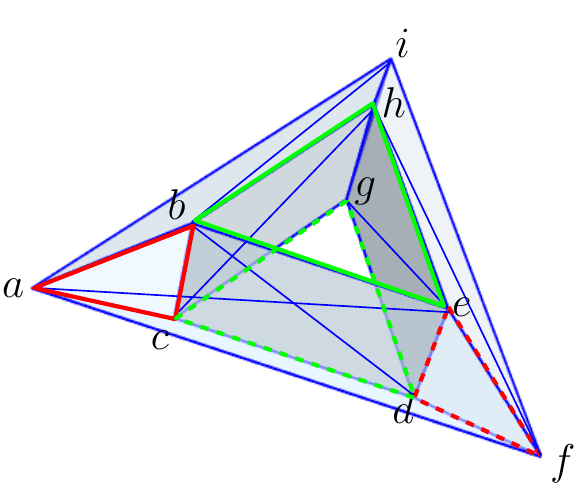}
    \caption{A simplicial complex (whose geometric realization is that of a torus embedded in $\bR^3$).}
    \label{fig:torus_SC}
\end{figure}
    \cref{fig:torus_SC} showcases a simplicial complex $K$ embedded in $\bR^3$, topologically equivalent to a torus. 
    Here, the set of vertices $V$ is $\{ \{a\}, \{b\},\dots , \{i\}\}$. 
    Its $1$-skeleton (edges) are depicted by blue lines, and its $2$-skeleton (faces, e.g.,~$\{ a,b,e \}$ and $\{a,e,f\}$) are depicted in shadow-gray. 
    There is no higher dimensional simplex, i.e.,~$K$ is of dimension $2$. 
    The $1$-chains $\{a,b\}+\{b,c\} +\{c,a\}$ and $\{d,e\}+\{e,f\} + \{f,d\}$, displayed in red are both $1$-cycles (i.e.,~belong to $\ker \partial_1$), but happen to be equivalent in the homology group $H_1$ as they are both in the boundary of a same $2$-chain (the ``cylinder'' obtained by restricting $K$ to $\{a,b,c,d,e,f\}$). 
    A similar comment holds for the $1$-cycles $\{b,h,e\}$ and $\{c,d,g\}$ (displayed in green). 
    Furthermore, red and green cycles are nonetheless independents from each others, and thus yield two different elements in $H_1$, which are actually the two generators of that group, which is thus a vector space of dimension $2$. 
\end{example}

\begin{remark}\label{rem:field}
    One can also define homology groups over a ring other than $\bF_2$, e.g.,~$\bZ$, for which we need to equip each simplex with an orientation and modify the boundary operator with $\pm$-signs. 
    This allows to identify \emph{torsion subgroups} in the homology groups. 
    This allows to distinguish between, e.g., a torus and a Klein bottle. 
    As most applications of topological data analysis use $\bF_2$, we specify our exposition to this choice in this survey and refer to \cite{obayashi2023field} for a thorough discussion on that topic. 
    Note that using a field (instead of a more general ring) as coefficients to build our chains implies that the homology groups are actually (finite dimensional) \emph{vector spaces}. 
    The list of dimensions of $H_p$ for $p \in \bN$ are called \emph{Betti numbers} of $K$, denoted by $(\beta_p)_p$. 
    Betti numbers are topological invariants, in that two objects that are homotopy equivalent (i.e.,~one can be obtain from the other by a smooth deformation, see \cite[Ch.~III.2]{edelsbrunner2010computational} for a formal definition) have the same Betti numbers. 
    In \cref{example:torus_CS}, one obtains $\beta_0 = 1$, $\beta_1 = 2$, $\beta_2 = 1$, and $\beta_p = 0$ for $p \geq 3$, accounting the for that that a torus has one connected component, two (independent) generating loops, and one cavity.
\end{remark}

In general, datasets are not made of objects represented as simplicial complexes, preventing one from straightforwardly use this construction in practice. 
Furthermore, building a specific simplicial complex $K$ on top of objects (say, point clouds) often relies on user-based arbitrary choices (e.g.,~connecting points that are closer than an arbitrary threshold $\epsilon$). 
In order to avoid depending on such priors, the main idea of
\emph{persistent homology} (PH) is to look at the homology groups of a \emph{growing sequence} of subcomplexes of a fixed (usually large) complex $K$, called a \emph{filtration} of $K$.

\begin{definition}\label{def:filtration}
    A \emph{filtration} of a simplicial complex $K$ is an indexed family of complexes $(K_t)_{t\in\bR}$ satisfying $K_t\subseteq K$ for all $t$, and $K_s\subseteq K_t$ whenever $s\leq t$. 
    Whenever the intervals $I_\sigma \coloneqq \{t\in\bR\,|\,\sigma\in K_t\}$ attain their infimums, 
    a filtration can be described by the function $f \colon K\to\bR$ defined with $f(\sigma) \coloneqq \min I_\sigma$. 
    Reciprocally, any function $f \colon K \to \bR$ satisfying $f(\sigma) \le f(\sigma')$ for any pair $\sigma \subseteq \sigma'$
    induces a filtration with $K_t\coloneqq\{\sigma\in K \mid f(\sigma)\leq t\}$.
    The set of filtrations on $K$ is denoted by $\Filt_K$. 
\end{definition}

\begin{remark}
In practice, most of standard filtrations in TDA are provided through functions. 
Hence, we will use the terms filtration and function satisfying the condition in \Cref{def:filtration} interchangeably hereafter, even though the former is more general than the latter.
From now on, we will implicitly assume that all filtrations considered are actually induced by functions, i.e.,~the intervals $I_\sigma$ defined above do attain their infimums. 
\end{remark}

Practically speaking, filtrations can be interpreted in different ways. For instance, in the case of Vietoris--Rips filtrations (see~\cref{def:vr} below),
these growing complexes can be seen as approximations of a topological space at different scales; computing the homology groups for all these complexes 
is thus a way to avoid computing the homology groups at a single arbitrary scale.
We provide below several examples of filtration that are routinely used in the TDA literature for various types of data (point clouds, images, etc.). 

\begin{example}[\v Cech filtration on point clouds] \label{example:Cech}
\begin{figure}
    \centering
    \includegraphics[width=\linewidth]{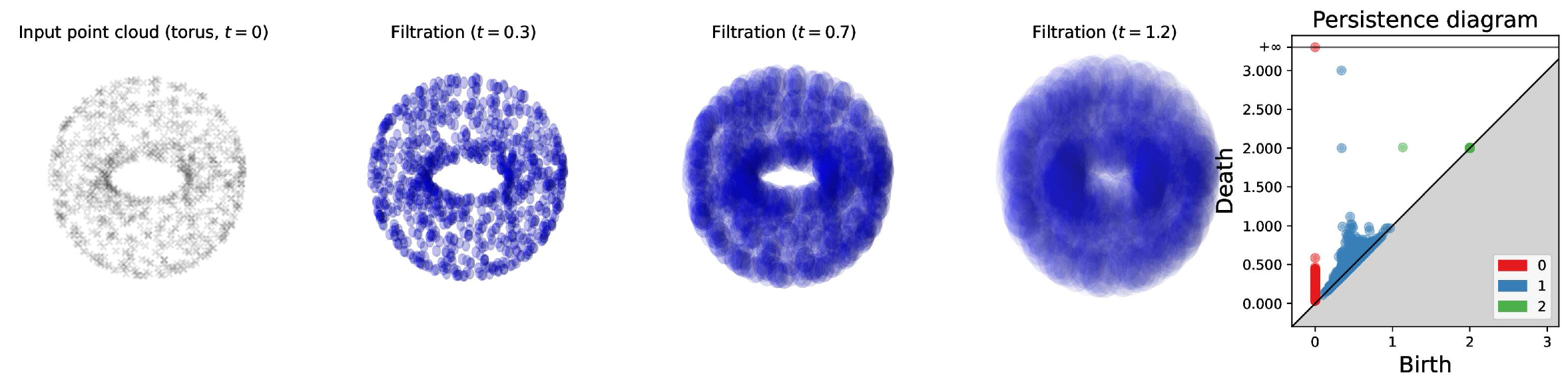}
    \caption{Illustration of the \v Cech filtration and resulting persistence diagram. From left to right, an input point cloud (sample of $n=1000$ points on a torus), increasing sublevel sets of the distance function to that point cloud, i.e.,~$\bigcup_{x \in X} B(x,t)$, and eventually on the right the corresponding persistence diagram (see \cref{def:PD}) Colors in the persistence diagram correspond to different homology dimension; the prominent point in red (death = $+\infty$) accounts for the unique connected component of the (underlying) manifold. The two prominent points in blue correspond to the two generating loops in $H_1$ (recall \cref{example:torus_CS})---their death coordinate indicating a ``small'' radius $r_1 = 2$ and a large radius of $r_2 = 5$, and the green one accounts for the cavity ($H_2$, also filled in at $t=2$). Points closer to the diagonal $\{b=d\}$ (in particular, red and blue ones) correspond to less persistent features accounting for the sampling (they would vanish, i.e.,~collapse on the diagonal, when $n \to +\infty$).}
    \label{fig:illu_cech_PD}
\end{figure}

    Let $X = \{x_1,\dots,x_n\} \subset \bR^d$ be a finite point cloud and $t \in \bR$. The \v Cech simplicial complex $C[X]_t$ is defined by
    \begin{equation}\label{eq:cech_characterization}
        \{x_{i_1},\dots,x_{i_p}\} \in C[X]_t \Longleftrightarrow \bigcap_{j=1}^p B(x_{i_{j}}, t) \neq \varnothing.
    \end{equation}
    The family $(C[X]_t)_{t \in \bR}$ is a filtration, which is called the \emph{\v Cech filtration} of $X$. 
    An important result referred to as the \emph{nerve theorem} states that $C[X]_t$ has the same homology as the set $\bigcup_{x \in X} B(x,t) \eqqcolon \cC[X]_t \subset \bR^d$ (see \cite[\S III.2]{edelsbrunner2010computational} and references therein). 
    When $X$ is sampled almost-uniformly on a sufficiently regular submanifold $M \subset \bR^d$ with $n$ large enough, $\cC[X]_t$ has itself the same homology as $M$ for a wide range of values of $t$  \cite{niyogi2008finding}, making this filtration a natural candidate to build estimators to infer topological properties of $M$. 
    See \cite{chazal2009sampling,boissonnat2018geometric,divol2021minimax,cohen2022lexicographic} and \cref{fig:illu_cech_PD}.
\end{example}

\begin{remark}
    An important drawback of the \v Cech filtration that limits its use in practical applications lies in its computational efficiency: asserting whether the intersection of balls in \eqref{eq:cech_characterization} is non-empty becomes quickly intractable when the ambient dimension $d$ increases.
\end{remark}

\begin{example}[Vietoris--Rips (VR) filtration on point clouds]
\label{def:vr}
    Let $X=\{x_0,\dots,x_n\} \subset \bR^d$ be a finite point cloud and $t \in \bR$. 
    The Vietoris--Rips simplicial complex $\Rips{X}_t$ at scale $t$ is the simplicial complex with vertex set $V(\Rips{X}_t) = X$ and whose $p$-skeleton is prescribed by: 
    \begin{equation}
        \{x_{i_0},\dots,x_{i_p} \} \in \Rips{X}_t 
        \Longleftrightarrow \text{$\|x_{i_j}-x_{i_k}\| \le 2t$ for any $j,k \in \{0,\dots,p\}$}.
    \end{equation}
    The family $\{\Rips{X}_t\}_{t \in \bR}$ is a filtration, which is called the \emph{Vietoris--Rips (VR) filtration} of $X$. It can be equivalently described with the function $f_{\rm VR}(X)$ defined as 
    $f_{\rm VR}(X)(\{x_{i_0},\dots,x_{i_p}\}) = \max \{\|x_{i_j}-x_{i_k}\|/2\,|\,j,k \in \{0,\dots,p\}\}$.
\end{example}

\begin{remark}
    An important property of the VR filtration is that---in contrast to the \v Cech filtration---it only depends on the pairwise distance matrix $(\|x_i - x_j\|)_{1 \leq i,j \leq n}$. 
    This has two interesting consequences: $(i)$ computing $R[X]_t$ becomes linear in the ambient dimension\footnote{Note however that the \emph{statistical efficiency} of the Vietoris--Rips filtration---when it comes to infer the topology of an underlying submanifold $M \subset \bR^d$ on which $X$ may be sampled---decreases when the \emph{intrinsic} dimension of $M$ increases \cite{balakrishnan2012minimax,divol2021minimax}.} $d$ (the cost of computing a distance), a sharp advantage over the \v Cech filtration, $(ii)$ the Rips filtration can be faithfully adapted to general metric spaces (e.g.,~a Riemannian manifold equipped with its geodesic distance). 
    This flexibility of the VR filtration enables the design of many variants: considering the geodesic distance on a triangulated mesh \cite{oudot2010geodesic}, e.g.,~weighting the radius by some local information based on density \cite{anai2020dtm,nishikawa2024adaptive}, etc. 
\end{remark}

\begin{example}[Height filtration] \label{example:height_filtration}
    Let $K$ be a simplicial complex embedded in $\bR^d$ and let $x_1,\dots,x_n$ denote its vertices. Pick a parameter $t \in \bR$ and a direction $\theta \in \bS^{d-1}$ (the unit sphere in $\bR^d$). 
    Then, define
    \begin{equation}
        H[K,\theta]_t \coloneqq \left\{ \{x_{i_1},\dots,x_{i_p}\} \in K \relmid \max_{1\leq j\leq p} \braket{x_{i_{j}}, \theta} \leq t \right\}.
    \end{equation}
    The family $(H[K,\theta]_t)_{t \in \bR}$ defines the \emph{height filtration} of $K$ in the direction $\theta$. 
    Looking at all possible directions $\theta \in \bS^{d-1}$ is analogous to considering a Radon transform of $K$, and it can be proved that this transform is injective when $d=2,3$, i.e.,~$ (H[K,\theta]_t)_t = (H[K',\theta]_t)_t,\ \forall \theta \in \bS^{d-1}  \Leftrightarrow K = K'$ \cite{turner2014}.
\end{example}

\begin{example}[Filtrations on graphs]\label{ex:graph_filt}
    As graphs are particular instances of simplicial complexes, many specific filtrations have been designed to deal with them. 
    Values on vertices and/or edges can typically be obtained as output of diffusion processes (heat equation, wave equation, etc.) or geodesic distances; see for instance \cite{archambault2007topolayout,li2012effective,ferrara2012topological,carriere2020perslay}.
    These filtrations also extend naturally to triangulations (as they can be seen as ``higher-dimensional'' graphs with triangles in addition to vertices and edges) in order to deal with, e.g., 3D meshes~\cite{carriere2015stable, poulenard2018topological, skrabaPersistenceBasedSegmentationDeformable2010}.
\end{example}

\begin{example}[Filtrations on images]\label{ex:img_filt}
	Images are structured objects and can be easily turned into simplicial complexes by simply 
	turning every pixel into two right triangles glued along their diagonals, and attaching these pairs of triangles together using pixel connectivity. Additionally, an even simpler way
	is to use specific complexes, called \emph{cubical complexes}, that are specifically tailored for 
	array data. The pixel values (such as, e.g., gray scale, one of the RGB channels, or segmentation masks) are then natural functions that can be used to filter images, and that are particularly efficient in generative models~\cite{wangTopoGANTopologyAwareGenerative2020, guptaTopoDiffusionNetTopologyAwareDiffusion2025, Barbarani2024}. This  approach extends straightforwardly to higher-dimensional images, such as scans with voxels~\cite{singhTopologicalDataAnalysis2023}.
\end{example}

\subsection{Persistent homology and persistence diagrams}

Let $\scK \coloneqq (K_t)_{t\in\bR}$ be a filtration of a finite simplicial complex $K$. The main idea of persistent homology is to track the appearance and disappearance of topological features in the filtrations, seen as elements in the corresponding homology groups. 
For $s \leq t$, we denote the inclusion between $K_s$ and $K_t$ by $\iota_{t,s}\colon K_s \hookrightarrow K_t$.
The map $\iota_{t,s}$                    induces a linear map $(\iota_{t,s})_\ast\colon H_p(K_s) \to H_p(K_t)$ between the corresponding homology groups (a property referred to as functoriality).
Persistent homology (PH) is then defined as the family of these homology groups and their connecting linear maps \cite{edelsbrunner2010computational,oudot2015persistence}.\footnote{Note that the term \emph{persistence module} is often encountered in the TDA literature as well to denote a family of vector spaces connected by linear maps (but not necessarily indexed over $\bR$).}

\begin{definition}[Persistent homology]
    Let $\scK=(K_t)_{t\in\bR}$ be a filtration of a simplicial complex $K$.
    The \emph{$p$-th persistent homology} of $\scK$ is defined as the pair $((H_p(K_t))_{t\in\bR}, ({\iota_{s, t}}_\ast)_{t\leq s})$ of the family of homology groups $(H_p(K_t))_{t\in\bR}$ together with the family of linear maps $((\iota_{s, t})_\ast)_{t\leq s}$ induced by inclusions.
\end{definition}

Intuitively, the linear map $(\iota_{s, t})_\ast$ encodes the correspondence between the bases of $H_p(K_s)$ and $H_p(K_t)$: if a non-trivial cycle in these bases is present in both $H_p(K_s)$ and $H_p(K_t)$ (at the $i$-th and $j$-th positions), then the entry at location $(j,i)$ in the matrix representation of $(\iota_{s, t})_\ast$ is 1. 
Similarly, if the $i$-th non-trivial cycle in the basis of $H_p(K_s)$ has become a boundary in $K_t$, then the $i$-th column of $(\iota_{s, t})_\ast$ is zero, and if the $j$-th non trivial cycle in the basis of $H_p(K_t)$ has not yet appeared in $K_s$, then the $j$-th row of $(\iota_{s, t})_\ast$ is zero.

The information captured by PH can be summarized by finding appropriate bases whose elements correspond exactly to the appearing or disappearing cycles in the filtration. 
The following theorem is a consequence of the explicit construction we present in \cref{subsec:computation} and can be found for instance in \cite[Ch.~2]{oudot2015persistence}. 

\begin{theorem}
\label{th:pd}
    Let $\scK=(K_t)_{t\in\bR}$ be a filtration of a finite simplicial complex $K$.
    For any $t\in\bR$, we let $\partial_{p,t} \coloneqq \partial_p|_{C_p(K_t)}$, i.e.,~the restriction of
    $\partial_p$ to the vector subspace $C_p(K_t)\subseteq C_p(K)$. 
    Moreover, for each $z\in \Ker \partial_{p,t}$, let $[z]_t$ denote the corresponding homology class in $H_p(K_t)$.
    Then, there exists $m\in\bN$, a sequence $\{z_i\}_{i=1}^m$ on $C_p(K)$ and a multiset $\{(b_i, d_i)\}_{i=1}^m$ of pairs on $\bR\cup\{+\infty\}$ that satisfy the following three conditions,
    for any $t\in\bR$:
    \begin{enumerate}
        \item[(i)] $b_i \leq t  \Longleftrightarrow z_i \in \Ker \partial_{p,t}$,
        \item[(ii)] $t\geq d_i \Longleftrightarrow z_i \in \Image \partial_{p+1,t}$,
        \item[(iii)] let $I_t \coloneqq \{i\in \{1,\dots, m\} \mid t\in [b_i, d_i)\}$, then the set $\{[z_i]_t\mid i\in I_t\}$ is a basis of $H_p(K_t)$.
    \end{enumerate}
    Additionally, the multiset $\{(b_i, d_i)\}_{i=1}^m$ that satisfies the conditions above is unique (up to permutations).
\end{theorem}

The multiset $\{(b_i, d_i)\}_{i=1}^m$ defined in \cref{th:pd} is called \emph{persistence diagram} (PD),
and is usually represented as a multiset of points in the extended plane by using the $b_i$'s and $d_i$'s as coordinates. These coordinates are usually called the \emph{birth times} and \emph{death times} of the corresponding cycles, or topological features.
For any filtration $f \colon K\to\bR$, we let $\Dgm(f)$ denote the corresponding PD.

\begin{remark}\label{rem:discard_essential_part}
    A point $(b,d)$ in a persistence diagram is said to be \emph{essential} if its second coordinate is $d = +\infty$. Intuitively, the essential points in a persistence diagram 
    correspond to the topological features of the final simplicial complex $K$ at the end of the filtration.
    For instance, when using the Vietoris--Rips filtration (\cref{def:vr}, see also \cref{fig:illu_cech_PD}), the final simplicial complex
    is the complete simplicial complex on $n$ points (i.e., the simplicial complex such that any subset of the vertices is a simplex), which has one connected component and no other higher-dimensional cycle. Hence, there is always exactly one essential point whose coordinates are $(0,+\infty)$ in homology dimension $p=0$: this point accounts for the single connected component created at birth time $b=0$ and persisting ``forever''. 

    In the context of topological optimization, these essential points are barely used (i.e.,~the loss functions that are typically used in that context do not depend on such points), and one rather focuses on \emph{ordinary} points in PDs, i.e.,~points $(b,d)$ with $b < d < +\infty$. 
    On the other hand, taking the essential part into account is important in some statements regarding the construction of persistence diagrams, such as \cref{prop:permpd}. 
    We thus decide to make an explicit distinction between these two types of objects to stress whenever one should consider PDs in their greatest generality (i.e.,~including essential parts) or may simplify by restricting to the ordinary points in PDs.
\end{remark}

This yields the following definition.
\begin{definition}\label{def:PD}
    A persistence diagram (PD) is a finite multiset of points supported on the extended closed half-plane $\overline{\groundspace} \coloneqq \{ (b,d) \in \overline{\bR}^2\,|\,\ b \leq d\}$.
    The set of persistence diagram is denoted by $\pdspace$. 
    We let $\pdspace$ denote the set of all persistence diagrams. 

    An \emph{ordinary} persistence diagram is a persistence diagrams with no essential part, i.e.,~a finite multiset of points supported on the open half-plane $\groundspace \coloneqq \{ (b,d) \in \bR^2\,|\,\ b < d\}$. 
    We let $\pdspace^o$ denote the set of all ordinary persistence diagrams.
    The \emph{multiplicity} of a point $x$ belonging to the support of a persistence diagram $\spt(\alpha),\ \alpha \in \pdspace$, is denoted by $m(x)$. 

    Eventually, we let $\pdspace_m$ (resp.~$\pdspace_m^o$) denote the subset of $\pdspace$ (resp.~$\pdspace^o$) made of persistence diagram with at most $m \in \bN$ points. 
\end{definition}

Moreover, ordinary persistence diagrams can be compared using \emph{partial matching} distances. 

\begin{definition}
\label{def:dist_wasserstein}
 Let $\alpha,\beta \in \pdspace^o$ be two ordinary persistence diagrams and $q \in [1,+\infty)$. 
The \emph{$q$-th diagram distance} between $\alpha$ and $\beta$ is defined as
\begin{equation}\label{eq:distance_between_PD}
    \FG_q(\alpha,\beta) = \inf_{\pi \in \Gamma(\alpha,\beta)} \left( \sum_{x \in \alpha \cup \thediag} \|x - \pi(x)\|^q \right)^{\frac{1}{q}},
\end{equation}
where $\| \cdot \|$ denote the $q'$-norm in $\bR^2$ (typically $q' = 2$ in applications), $\Gamma(\alpha,\beta)$ is the set of partial matching between $\alpha$ and $\beta$, i.e.,~bijections between $\alpha \cup \thediag$ and $\beta \cup \thediag$, and $\thediag \coloneqq \{(t,t),\ t \in \bR\}$ is the boundary of $\groundspace$, typically referred to as ``the diagonal''. 
When $q = +\infty$, the sum becomes a supremum and the distance $\FG_\infty$ is referred to as the \emph{bottleneck distance} between persistence diagrams.
See also \cref{fig:matching_distance}.
\end{definition}

\begin{remark}
    When $q < \infty$, the distance $\FG_q$ is often referred to as the {\em Wasserstein distance} between persistence diagrams in the TDA literature \cite[Ch.~VIII.2]{edelsbrunner2010computational}, due to its similarity with the Wasserstein distance between probability measures used in optimal transport (see for instance \cite[Ch.~5]{santambrogio2015optimal}). 
    Interestingly, the distance $\FG_q$ was actually introduced initially by Figalli and Gigli \cite{figalli2010new} in a context unrelated to TDA.  
    The connection with the metrics used in topological data analysis has been made later in \cite{divol2019understanding}. 
    We use the somewhat unusual notation $\FG$ to stress the distinction between these distances and their counterpart in optimal transport literature.
\end{remark}

\begin{figure}[ht]
    \centering
    \begin{tikzpicture}[baseline=(current bounding box.center)]
        \fill[lightgray] (0,0) -- (4.5,4.5) -- (4.5,0) --cycle;
        \draw[->] (0,0)--(4.5,0) node [label=right: {$b$}] {};
        \draw[->] (0,0) -- (0,4.5) node [label=above: {$d$}] {};
        \draw (0,0)--(4.5,4.5);
    
        \draw[gray] (1.7,3.2)--(2,3);
        \draw[gray] (0.3,1.5)--(0.5,2);
        \draw[gray] (0.5,3.5)--(1,4);
        \draw[gray] (3,4.4)--(3.5,4.4);
        \draw (2,2.3)--(2.15,2.15);
        
        \draw plot [only marks, mark=*] coordinates {(1.7,3.2) (0.3,1.5) (0.5,3.5) (3.5,4.4)};
        \draw[blue] plot [only marks, mark=square*] coordinates {(2,3) (0.5,2) (1,4) (2,2.3) (3,4.4)};
    \end{tikzpicture}
    \caption{Partial matching distance between PDs.}
    \label{fig:matching_distance}
\end{figure}
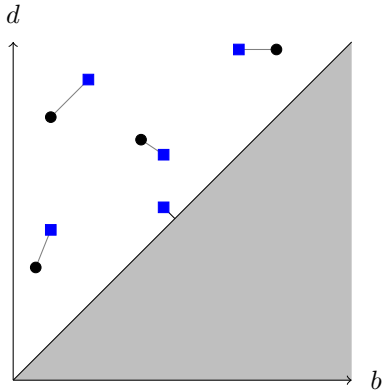

\begin{remark}\label{rem:PD_space_structure} 
    The metric space $(\pdspace^o_m, \FG_2)$ enjoys several properties \cite{mileyko2011probability,turner2014frechet,chowdhury2019geodesics,divol2019understanding}: it is a Polish space, non-negatively curved, with geodesics known in closed-form (provided that one can compute the optimal partial matching $\pi$ between two diagrams in \eqref{eq:distance_between_PD}). 
    It implies that several statistical tools, such as probability distributions supported on $\pdspace^o_m$ and their (Fr\'echet) means are well-defined \cite{turner2014frechet,turner2013means,munch2015probabilistic,cao2022geometric}, enabling the use of some standard algorithms that only requires a metric structure, such as $k$-means clustering \cite{marchese2017k,davies2020fuzzy,lacombe2018large,cao2024k}. 
    
    However, this metric space does not exhibit a linear (Hilbert) structure and a fruitful line of works suggests that it cannot be embedded in such spaces (in the greatest generality) without distorting arbitrarily the metric if the number of points in the persistence diagram goes to infinity or to $0$ \cite{bubenik2020embeddings,mitra2024geometric,carriere2019metric,mitra2021space,pritchard2024coarse}. 
    This means that most machine learning methods and models, which are designed for input data that belongs to a (often finite-dimensional) vector space, 
    cannot be used faithfully when working with persistence diagrams. 
    
    Two natural workarounds consist of $(i)$ either explicitly pushing persistence diagrams into a linear space nonetheless (hence necessarily loosing some geometric information as diagram distances are arbitrarily distorted) by manually designing a \emph{vectorization} or \emph{representation} of persistence diagrams, that is, a map $\Phi \colon \pdspace\rightarrow\mathcal H$, where $\mathcal H$ is a Hilbert space (see, e.g., \cite{bubenik2015statistical,adams2017persistence,chazal2014stochastic,umeda2017time,hofer2019learning,carriere2020perslay}), or $(ii)$ designing such a map implicitly using \emph{kernel methods} (see, e.g.,~\cite{Carriere2017b,le2018persistence,kusano2018kernel,vishwanath2020robust}).
    
    Moreover, the absence of a simple linear structure in $\pdspace$ makes challenging to even properly define the derivative of maps taking value from (or valued in) $\pdspace$, which is the problem at the core of topological optimization, that we focus on in this survey.
\end{remark}

\begin{remark}
For the sake of simplicity, and because it matches standard practice, we restrict in this work to persistence diagrams with finitely many points. 
Most theoretical properties regarding $\pdspace$ and $\pdspace^o$ presented in this survey hold when defining persistence diagrams more generally as \emph{locally} finite multisets $\alpha$ supported on $\groundspace$ satisfying the integrability constraint $\perstot_q(\alpha) \coloneqq \FG_q(\alpha,\emptydgm) < \infty$. 
Here, $\emptydgm$ denotes the empty persistence diagram. 
This extension can be understood as the completion of $\pdspace^o$ for the metric $\FG_q$ when $q < +\infty$ (the case $q=+\infty$ is slightly more intricate and is discussed for instance in \cite{blumberg2014robust,bubenik2018topological,pereaApproximatingContinuousFunctions2023} and \cite[\S 3.3]{divol2019understanding}).
\end{remark}

The main motivation to use the partial distance $\FG_q,\ 1 \leq q \leq +\infty$ lies in the following ``stability theorems'', guaranteeing that close filtrations must induce close PDs. 

\begin{theorem}[\cite{Cohen-Steiner2007, chazal2012structure}]
\label{th:stability}
Let $K$ be a finite simplicial complex. 
Then the map $\Filt_K \to \pdspace; f \mapsto \Dgm(f)$ is $1$-Lipschitz continuous in the following sense:
\begin{equation}
    \FG_\infty(\Dgm(f), \Dgm(g))\leq \|f-g\|_\infty.
\end{equation}

Similar (yet more intricate) stability results exist for $q < \infty$ \cite{cohen2010lipschitz,skraba2020wasserstein}. 
In particular, if $X$ and $X'$ are two point clouds with the same cardinality $n$, letting $\DgmRips(\cdot)$ denote the persistence diagram built on top of a point cloud using the Vietoris--Rips filtration, one can prove that 
\begin{equation}
    \FG_q(\DgmRips(X)),\DgmRips(X')) \leq C(q,n) \inf_{\phi} \left( \sum_{x \in X} \|x - \phi(x)\|^q \right)^\frac{1}{q},
\end{equation}
where $\phi$ is a bijection between $X$ and $X'$, and $C(q,n)$ is a constant depending on $q$ and $n$.
\end{theorem}

From an optimization perspective, these stability results are of importance: they often allow one to prove that composite maps of the form $\cL = L \circ \Pers \colon \bR^{d_1} \to \bR^{d_2}$ (with $\Pers \colon \bR^{d_1} \to \pdspace$ and $L \colon \pdspace \to \bR^{d_2}$) are Lipschitz, and thus must be differentiable almost everywhere thanks to Rademacher's theorem, without the need to define a global differential structure in $\pdspace$. 
This observation will be at the core of \Cref{sec:differentiability}.
Before elaborating on this idea, we present how one can compute persistent homology in practice, as the corresponding algorithm is also crucial to understand how composite maps going through $\pdspace$ can be differentiated. 

\subsection{Persistence computation}\label{subsec:computation}

Now that we have seen how PDs are defined, we explain in this section how they can be computed practically. 
Given a filtration $f \colon K \to \bR$ on a simplicial complex $K$, the computation of PDs splits into two parts: 
the first is purely combinatorial, and the second adapts the results to the provided filtration values.
Moreover, they both rely on the (filtration-dependent) pre-order\footnote{i.e.,~one can have $\sigma \neq \tau$ while $\sigma \preceq_f \tau \preceq_f \sigma $.} on the simplices of $K$ defined by $\tau \preceq_f \sigma$ if and only if $f(\tau) \leq f(\sigma)$. 
This pre-order can be refined into a total order by breaking ties in some fairly arbitrary way, as long as it is consistent with face relations, i.e., as long as $\tau\subseteq\sigma \Rightarrow\tau \preceq_f \sigma$.

\paragraph{First part: combinatorial part (persistence pairing).} 
The first part of PD computations, called the \emph{persistence pairing algorithm}, is about finding the simplices $\sigma_i$'s that
give rise to the appearance of the cycles $z_i$'s 
(which, in turn, create new topological features in the filtration) in the total order induced by $f$,
as well as the simplices $\sigma'_i$'s that give rise to the boundaries $z'_i$'s killing these cycles (i.e.,~such that $z_i=\partial_p(z'_i)$) according to \cref{th:pd}.
The pairs $(\sigma_i,\sigma'_i)$ are called \emph{persistence pairs}.
 
Note that while different total orders (associated to the same pre-order) may yield different persistence pairs, they all translate to the same persistence diagram in the second part (see below).
The basic algorithm to compute persistence pairs iterates over the ordered set of simplices $\sigma_1 \preceq_f \dots \preceq_f \sigma_{|K|}$ according to \cref{alg:persistent-pairs} below---see \cite[Section~11.5.2]{boissonnat2018geometric} for a detailed description of the algorithm. 
In a nutshell, this algorithm relies on the following observation: when inserting a simplex $\sigma$ of dimension $p$ in the filtration, this either creates a new generator in dimension $p$ (increase the dimension of $H_p$), or kills one of dimension $p-1$ (decrease the dimension of $H_{p-1}$). 

\begin{algorithm}[ht]
  \caption{$\mathtt{PersistencePairs}(f)$}
  \label{alg:persistent-pairs}
  \begin{algorithmic}
  	\STATE{\bf Input:} Filtration $f\in\Filt_K$
    \STATE Order the simplices of $K$ so that $\sigma_1 \preceq_f \dots \preceq_f \sigma_{|K|}$;
  	\STATE $K_0 \leftarrow \varnothing$;
    \STATE $\mathrm{Pairs}_0 = \mathrm{Pairs}_1 = \dots = \mathrm{Pairs}_{d-1} = \varnothing$;
    \FOR{$j=1$ to $|K|$}
    	\STATE $p \leftarrow \dim \sigma_j$;
    	\STATE $K_j \leftarrow K_{j-1} \cup \{\sigma_j\}$;
    	\IF{$\sigma_j$ does not create a new homology class in $H_p(K_j)$}
    	    \STATE A homology class in $H_{p-1}(K_{l(j)})$ which was created by $\sigma_{l(j)}$ (for some $l(j) < j$) becomes homologous to $0$ as the boundary of a chain created by $\sigma_j$; 
    		\STATE $\mathrm{Pairs}_{p-1} \gets \mathrm{Pairs}_{p-1} \cup \{ (\sigma_{l(j)}, \sigma_j) \}$;
    	\ENDIF
    \ENDFOR
    \STATE{\bf Output:} Persistence pairs in each dimension $\mathrm{Pairs}_0, \mathrm{Pairs}_1, \dots, \mathrm{Pairs}_{d-1}$
  \end{algorithmic}
\end{algorithm}

Note that for each dimension $p$, some $p$-dimensional simplices may remain unpaired at the end of the algorithm; they correspond to the 
pairs $(b_i,+\infty)$, i.e.,~correspond to \emph{essential points} in the PD (as opposed to the \emph{ordinary points} with finite coordinates, see \cref{def:PD}).

\paragraph{Second part: associated filtration values.}
The persistence diagram $\Dgm(f)$ of the filtration $f$ is then obtained by associating to each persistence pair $(\sigma_{l(j)}, \sigma_j)$ the point $a_j$ in the extended plane
whose coordinates are $a_j=(f(\sigma_{l(j)}), f(\sigma_j))$.
Moreover, each unpaired simplex $\sigma_j$ induces an essential point as per $a_j=(f(\sigma_j),+\infty)$.
We let $P^f$ and $U^f$ denote the sets of persistence pairs and unpaired simplices respectively, so that the final output of the PD computation can be written as:
\begin{equation}
    \Dgm(f) = \{ (f(\sigma),f(\sigma')) \}_{(\sigma,\sigma') \in P^f} \cup \{ (f(\tau),\infty) \}_{\tau \in U^f}. 
\end{equation}

In practice, the above construction is usually done dimension by dimension (in order to get a single PD for every homological dimension) by restricting the algorithm
to the simplices of $K$ of dimension $p$ and $p+1$. 
The resulting $p$-th dimensional PD is denoted by $\Dgm_p(f)$. The pairs $\{(P_p^f,U_p^f)\}_p$ are called barcode templates in \cite[Def.~4.3]{LOT2022}.

Finally, it is also useful to characterize those filtrations that induce the same persistence pairs and unpaired simplices.

\begin{definition}[{\cite[Def.~4.2]{LOT2022}}]\label{def:order_eq}
    Two filtrations $f, g \colon K \to \bR$ are said to be {\em ordering equivalent}, written as $f \sim g$, if they induce the same pre-order on the simplices of $K$.
    Note that this relation is an equivalence relation on $\Filt_K$. 
\end{definition}

Since the persistence pairs and the unpaired simplices only depend on the pre-order induced by a filtration (after arbitrarily breaking the ties), we have the following result.

\begin{proposition}\label{proposition:order_persistence_pairs}
    If $f$ and $g$ are ordering equivalent filtrations, then $P_p^f=P_p^g$ and $U_p^f=U_p^g$.
\end{proposition}

\section{Differential framework for persistence diagrams}\label{sec:differentiability}

In applications, filtrations over a (finite) simplicial complex $K$ are often \emph{parametrized} by some $\theta$ belonging to a submanifold $M \subseteq \bR^d$ through a map $M \ni \theta \mapsto f_\theta \in \Filt_K$. 
For instance, the height filtration (\cref{example:height_filtration}) over a point cloud $X = \{ x_1,\dots,x_n \} \subset \bR^d$ is typically parametrized by the direction $\theta \in M = S^{d-1}$, the \v Cech and Vietoris--Rips filtrations (\cref{example:Cech,def:vr}) can be considered as being parametrized by the point cloud $X \in M = \bR^{n \times d}$ itself, etc. 
It is natural to wonder whether some parameters are preferable for a given purpose, which is quantified by a task-dependent \emph{loss function} $L \colon \pdspace \to \bR$ that evaluates the quality of a diagram $\Dgm(f_\theta) \in \pdspace$ for a given parameter $\theta \in M$. 
Therefore, one seeks to optimize (say, minimize) the composite map $\cL \colon M \ni \theta \mapsto L(\Dgm(f_\theta)) \in \bR$, a task that we refer to as \emph{(persistence-based) topological optimization}. 
All the methods we present in this survey rely on minimizing $\cL$ using gradient-based methods. 
See also~\Cref{fig:exdiff}.

\begin{remark}\label{rem:rademacher}
Provided that $\theta \mapsto f_\theta$ and $\pdspace \ni \alpha \mapsto L(\alpha)$ are both locally Lipschitz---which is the case for all filtrations and loss functions considered in this survey---, \cref{th:stability} along with Rademacher's theorem ensure that $\cL$ is locally Lipschitz hence differentiable almost everywhere, meaning that the gradient $\nabla_\theta \cL$ exists generically. 
However, the (lack of linear) structure of $\pdspace$ (see \cref{rem:PD_space_structure}) a priori prevents from a straightforward computation of $\nabla_\theta \cL$ using the chain rule as differentials of maps from $M$ to $\pdspace$ and from $\pdspace$ to $\bR$ have yet to be defined. 
\end{remark}

\subsection{Formulation of differentiability}\label{subsec:pd_diff}

This subsection summarizes the main results of \cite{LOT2022}, providing a formal framework to defined differentials (in particular, gradients) of composite maps $\cL \colon M \to \pdspace \to N$ where $M$ and $N$ are two manifolds of class $C^\infty$ and without boundary (in practice, one often considers $M = \bR^{d_1}$ and $N = \bR^{d_2}$, typically with $d_2 = 1$).

\begin{definition}[{\cite[Def.~3.1]{LOT2022}}]
    Let $m, n \in \bN$. 
    The space of \emph{ordered persistence diagrams} with $m$ ordinary points and $n$ essential points is $\bR^{2m} \times \bR^n$ equipped with the Euclidean norm. 
    The map $Q_{m,n} \colon \bR^{2m} \times \bR^n \to \pdspace$ quotients the space by the action of the product of the symmetric groups $\perm_m \times \perm_n$ through permutations of points. 
    That is, for any ordered persistence diagram $\tilde{\alpha}= ((b_1,d_1, \dots, b_m, d_m), (v_1,\dots,v_n)) \in \bR^{2m} \times \bR^n$, 
    \begin{equation}
        Q_{m,n}(\tilde{\alpha}) \coloneqq \{ (b_i, d_i),\ 1 \leq i \leq m \} \cup \{ (v_j, +\infty),\ 1 \leq j \leq n\} \in \pdspace,
    \end{equation}
    where points are counted with multiplicity.
\end{definition}

\begin{definition}[{\cite[Def.~3.3]{LOT2022}}]\label{def:r-diff}
    Let $M$ be a manifold and $r \in \bN \cup{\{\infty\}}$. 
    A persistence diagram-valued map $\Pers \colon M \to \pdspace$ 
    is said to be \emph{$r$-differentiable} at $\theta \in M$ if there exist an open neighborhood $U$ of $\theta$, $m, n \in \bN$, 
    and a map $\tilde{B} \colon U \to \bR^{2m} \times \bR^n$ of class $C^r$ such that $\Pers = Q_{m,n} \circ \tilde{B}$. 
    One calls $\tilde{B}$ a \emph{local lift} of $\Pers$. 
\end{definition}

\begin{remark}\label{rem:diffeology}
    This definition of $r$-differentiability, based on \emph{diffeology} \cite{iglesias2013diffeology}, implies in particular that $\Pers(\theta')$ must have exactly $m$ ordinary points and $n$ essential points for any $\theta'$ in the neighborhood $U$ of $\theta$. 
    In particular, since the quotient map $Q_{m,n}$ is (Lipschitz) continuous when $\pdspace$ is equipped with the distance $\FG_q$ for $q \in [1,+\infty]$ (\cite[Prop.~3.2]{LOT2022}), $0$-differentiability of a map $\Pers$ implies its continuity with respect to $\FG_q$. 
    The converse is however false: a map $\Pers$ could create or destroy a point locally (which would make it not $0$-differentiable) while still being continuous for the distance $\FG_q$. 
    For instance, let $M = \bR$ and $\Pers(\theta) = \{(1-\theta^2,1+\theta^2)\} \in \pdspace$ if $\theta \neq 0$, and $\emptydgm$ (the empty diagram) if $\theta = 0$. 
    This map is not $0$-differentiable at $\theta = 0$, but it is continuous as $\FG_q(\Pers(\theta), \emptydgm) = \sqrt{2} \theta^2 \to 0$ as $\theta \to 0$. 

    Enforcing the conservation of the number of points $m$ in \cref{def:r-diff} might seem to be a strong restriction. 
    However, when $\Pers$ is given by $\Pers(\theta) \coloneqq \Dgm(f_\theta)$ for some (parametrized) filtration $f_\theta \in \Filt_K$ (on some finite simplicial complex $K$)---a very standard case in practice---, generic configurations (i.e., configurations where $f_\theta$ takes distinct values on different simplices) always induce, locally, the same pre-order on $K$ and thus the same number of points in the resulting persistence diagram. 
    It means that, in applications, considering local perturbations of the parameter $\theta$ (which is the case when computing gradients) must, generically, let $m$ and $n$ unchanged. 
\end{remark}

\begin{definition}[{\cite[Def.~3.7]{LOT2022}}]\label{def:diffeology_M_to_D}
    Let $M$ be a smooth manifold and $r \in \bN \cup \{\infty\}$. 
    Moreover, let $\Pers \colon M \to \pdspace$ be a PD-valued map and $\theta \in M$. 
    A \emph{$C^r$ local coordinate system} for $\Pers$ at $\theta$ is a collection of maps $(b_i,d_i \colon U \to \bR)_{i \in I}$ and $v_j \colon U \to \bR)_{j \in J}$ for finite sets $I,J$ defined on an open neighborhood $U$ of $\theta$ such that:
    \begin{enumerate}
        \item[(1)] The maps $b_i,d_i, v_j$ are all of class $C^r$;
        \item[(2)] For any $\theta' \in U$, one has the multi-set equality:
        \begin{equation}
            \Pers(\theta') = \{ (b_i(\theta'), d_i(\theta')) \}_{i \in I} \cup \{(v_j(\theta), +\infty) \}_{j \in J}.
        \end{equation}
    \end{enumerate}
    One simply writes $(U, (b_i,d_i)_{i \in I})$ for a local coordinate system.
\end{definition}

\begin{lemma}[{\cite[Prop.~3.8]{LOT2022}}]
    Let $M$ be a manifold and $\Pers \colon M \to \pdspace$ be a PD-valued map.
    Then $\Pers$ is $r$-differentiable at $\theta \in M$ if and only if it admits a $C^r$ local coordinate system at $\theta$. 
\end{lemma}

We now discuss the differentiability of maps {\em defined} on PDs and valued in a manifold $N$. 
In most applications, $N$ is typically a Banach, Hilbert or Euclidean space (quite often, simply $\bR$ or $\bR^d$), and such maps are then referred to as \emph{topological losses}, or \emph{vectorizations} of persistence diagrams, depending on the applications. 

\begin{definition}[{\cite[Def.~3.10]{LOT2022}}]\label{def:differentiable_fromD}
    Let $N$ be a smooth manifold and $r \in \bN \cup \{\infty\}$. 
    A map $L \colon \pdspace \to N$ is said to be \emph{$r$-differentiable} at $\alpha \in \pdspace$ if for any $m, n \in \bN$ and any vector $\tilde{\alpha} \in \bR^{2m} \times \bR^n$ 
    satisfying $Q_{m, n}(\tilde{\alpha})=\alpha$, the map $L \circ Q_{m, n} \colon \bR^{2m} \times \bR^n \to N$ is $C^r$ on an open neighborhood of $\tilde{\alpha}$.
\end{definition}

The differentials of PD-valued maps and maps defined on PDs can then be defined using lifts.

\begin{definition}[{\cite[Def.~3.6 and 3.13]{LOT2022}}]\label{def:diffeology_D_to_N}
    One has the following differentials:
    \begin{enumerate}
        \item Let $\Pers \colon M \to \pdspace$ be $1$-differentiable at $\theta \in M$ and let $\tilde{B} \colon U \to \bR^{2m} \times \bR^n$ be a $C^1$ lift of $\Pers$ defined on an open neighborhood $U$ of $\theta$. 
        The differential $\dd_{\theta,\tilde{B}} \Pers$ of $\Pers$ at $\theta$ with respect to the lift $\tilde{B}$ is defined as the differential of $\tilde{B}$ at $\theta$: 
        \begin{equation}
            \dd_{\theta,\tilde{B}} \Pers \colon T_\theta M \xrightarrow{\dd_\theta \tilde{B}} \bR^{2m}.
        \end{equation}
        \item Let $L \colon \pdspace \to N$ be $1$-differentiable at $\alpha \in \pdspace$ and $\tilde{\alpha} \in \bR^{2m} \times \bR^n$ be a pre-image of $\alpha$ under $Q_{m, n}$. 
        The differential of $L$ at $\alpha$ with respect to $\tilde{\alpha}$ is defined as the differential of $L\circ Q_{m, n}$ at $\tilde \alpha$: 
        \begin{equation}
            \dd_{\alpha,\tilde{\alpha}}L \colon \bR^{2m} \times \bR^n \xrightarrow{\dd_{\tilde{\alpha}}(L \circ Q_{m, n})} T_{L(\alpha)}N.
        \end{equation}
    \end{enumerate}
\end{definition}

While the differentials of $\Pers$ and $L$ defined above depend respectively on the choice of the lift $\tilde{B}$ and the ordered diagram $\tilde{\alpha}$, 
the following result established in \cite{LOT2022}~states that the chain rule used for composing these maps is oblivious to $\tilde{B}$ and $\tilde{\alpha}$. 

\begin{proposition}[{Chain rule \cite[Prop.~3.14]{LOT2022}}]\label{prop:chain}
    Let $\Pers \colon M \to \pdspace$ be $r$-differentiable at $\theta \in M$ and $L \colon \pdspace \to N$ be $r$-differentiable at $\Pers(\theta)$.
    Then, one has:
    \begin{enumerate}
        \item $L \circ \Pers \colon M \to N$ is $C^r$ at $\theta$ as a map between smooth manifolds;
        \item If $r \ge 1$, for any local $C^1$ lift $\tilde{B} \colon U \to \bR^{2m} \times \bR^n$ of $\Pers$ at $\theta$, one has 
        \begin{equation}\label{eq:chain_rule}
            \dd_\theta(L \circ \Pers) = \dd_{\Pers(\theta),\tilde{B}(\theta)} L \circ \dd_{\theta,\tilde{B}}\Pers.
        \end{equation}
    \end{enumerate}
\end{proposition}

From the practitioner viewpoint, this result is crucial: it practically allows to compute the gradients of real-valued composite maps going through $\cD$ \emph{in a canonical way}. 
Formally, it means that, in applications, when dealing with a composite map $\cL \colon M \xrightarrow{\Pers} \pdspace \xrightarrow{L} N = \bR$ where both $\Pers$ and $L$ are $1$-differentiable, it is legitimate to treat the intermediate diagram as an element of the Euclidean space $\bR^{2m} \times \bR^n$, deriving a standard gradient in $\bR^{2m} \times \bR^n$ as well, and eventually backpropagate it to obtain a gradient with respect to $\theta$. 

Just as in \cref{rem:diffeology}, we emphasize that $r$-differentiability (as per \cref{def:diffeology_M_to_D,def:diffeology_D_to_N}) of the maps $\Pers$ and $L$ also implies that of $\cL = L \circ \Pers$ in the usual sense. However, it is possible to find composite maps $\cL$ that would be of class $C^r$, $r \in \bN \cup \{\infty\}$, without $\Pers$ and/or $L$ being $r$-differentiable. 
Nonetheless, it turns out that typical maps $\Pers$ (induced by filtrations) and $L$ involved in topological optimization are $r$-differentiable, allowing one to use Equation \eqref{eq:chain_rule} faithfully. 
The next two sections provide such examples. 

\subsection{Examples of topological losses}\label{subsec:example_losses}

In most applications of topological optimization based on PDs (see \cref{sec:applications} for examples),
the functions to minimize, or \emph{(topological) losses}, are instances of $r$-differentiable maps $L \colon \pdspace\to N$ with $N=\bR$.
These losses are typically used to model some qualitative objective (such as, e.g., increasing or decreasing the sizes of topological features in order to regularize a model with topological priors). 
In the following, we provide few examples of standard topological losses; following \cref{rem:discard_essential_part}, these losses only depends on the ordinary part of the persistence diagram. 
All the results of \cref{subsec:pd_diff} can be adapted seamlessly simply by dropping the terms depending on $n$; we will write $Q_m$ (instead of $Q_{m,0}$) the corresponding quotient map.

\begin{example}[Total persistence and some variations]\label{ex:perstot}
    The {\em total persistence} $\perstot(\alpha)$ of a persistence diagram $\alpha \in \pdspace^o$ is the sum of the distances to the diagonal of the (ordinary) points $(b,d)$ of $\alpha$:
    \begin{equation}
        \perstot(\alpha) = \frac{1}{2} \sum_{(b,d)\in \alpha} (d-b)^2.
    \end{equation}
    For $\alpha \in \pdspace^o$ and an ordered persistence diagram $\tilde{\alpha} \in \bR^{2m}$ satisfying $Q_{m}(\tilde{\alpha}) = \alpha$, we have
    \begin{equation}
        \perstot \circ Q_{m} \coloneqq \bR^{2m} \ni (b_1,d_1, \dots, b_m, d_m) \mapsto \frac{1}{2} \sum_{i=1}^m (d_i-b_i)^2 \in \bR.
    \end{equation}
    Hence, the function $\perstot \colon \pdspace \to \bR$ is $\infty$-differentiable everywhere on $\pdspace$, and its gradient an be identified with that of its lift $(b_1,d_1,\dots,b_m, d_m) \mapsto (b_i - d_i, d_i - b_i)_{i=1}^m$. 
    
    The total persistence bears some similarities with the total variation: intuitively, it quantifies the \emph{topological complexity} of input data, by measuring and adding up the sizes of all the topological features.
    Minimizing this loss can be understood at reducing this complexity, resulting in smoother, less complex datasets or models. 
    As such, it has been used, e.g., for reducing overfitting in complex predictive models, as it can capture geometric patterns that traditional $L^1$ or $L^2$ penalties are oblivious to. 
    See for instance~\cite{chen2019topological, Hacquard2022} and \cref{sec:app_reg}. 
    Practically speaking, minimizing the total persistence attempts at pushing all the persistence diagram points toward the diagonal $\thediag$; its global minimum is $0$ and reached if and only if the diagram $\alpha$ is empty. 

    Naturally, this loss admits several variation, such as (non-exhaustive list): only minimizing the second coordinates, that is, the death times of the topological features (this only modifies the filtration values of the death simplices), using another exponent $p > 1$ on the terms: $|d_i - b_i|^p$, etc. 

    Instead of minimizing the total persistence, one may consider maximizing it (or minimizing its opposite), in which case one is trying to enhance the topological features appearing in the filtration by pushing persistence diagram points away from the diagonal. 
    This is what happens in \cref{fig:exdiff}.

    A slight yet useful variation of the total persistence loss is the \emph{topological simplification loss}, where one only penalizes points in the diagram that are already close to the diagonal, i.e.,~$\sum_{i=1}^m (d_i - b_i) 1_{|d_i - b_i| < \eta}$ for some fixed threshold $\eta > 0$. 
    At the level of the persistence diagram, this loss will push points that are already close to the diagonal (often considered as ``topological noise'') closer to it, while preserving points away from the diagonal (corresponding to ``topological signal'' or ``macroscopic topological features''). 
\end{example}

\begin{example}[Distance to a target diagram]\label{ex:dist_to_dgm}
    Given a diagram $\alpha_0$ in $\pdspace^o$, the (squared) \emph{distance to $\alpha_0$} is a function on $\pdspace^o$ defined as: 
    \begin{equation}\label{eq:dist_to_target}
        {\rm Dist}_{\beta} \colon \pdspace^o \to \bR;\ \alpha \mapsto \frac{1}{2} \FG_2(\alpha,\beta)^2,
    \end{equation}
    where $\FG_2$ denotes the $2$-th diagram distance as in \cref{def:dist_wasserstein}. 
    By minimizing this loss, one can enforce $\alpha$ to be equal, or at least close, to a fixed, target persistence diagram $\beta$, by moving the persistence diagram points toward the ones of the target $\beta$. 
    This objective typically occurs in generative models when the dataset to model can be efficiently characterized by its geometric patterns, as is the case for the distribution of dark matter in large scale simulations of the universe~\cite{biagettiFisherForecastsPrimordial2022, ouelletteTopologicalDataAnalysis2023, yipCosmologyPersistentHomology2024}, or for the spatial arrangements of single cells in biopsies~\cite{atienzaStableTopologicalSummaries2021, aukerman2022persistent, vipondMultiparameterPersistentHomology2021}. 
    Another standard use case is dimension reduction, where one forces the latent spaces learned by neural network architectures to have persistence diagrams as close as possible to the one of the input, high-dimensional data, in order to preserve its geometric patterns and topological features. See~\cite{moor2019topological, Wagner2021, vandaele2021topologically} and \cref{subsec:app_dr}.

    If we consider $\pi^\star \in \Gamma(\alpha,\beta)$ the (generically unique\footnote{as the solution of a linear programming problem on a convex polytope.}) optimal partial matching between two diagrams $\alpha,\beta \in \pdspace^o$, and let $\tilde\alpha = (x_1, \dots, x_m) \in \bR^{2m}$ denote a lift of $\alpha$ (where $x_i = (b_i,d_i) \in \groundspace$), one has $\FG_2(\alpha,\beta)^2 = \sum_{x \in \alpha \cup \thediag} \|x - \pi^\star(x)\|^2$. 
    It follows from the envelope theorem that the gradient of ${\rm Dist}_{\beta}$ can be identified with
    \begin{equation}\label{eq:grad_dist}
         (x_i - \pi^\star(x_i))_{i=1}^m, 
    \end{equation}
    Unsurprisingly, following the opposite of this gradient will push $x_i \in \alpha$ toward its target $\pi^\star(x_i) \in \beta \cup \thediag$.
    
    \paragraph{Stability of the gradient.} Even though the optimal partial matching $\pi^\star$ between two diagrams $\alpha,\beta \in \pdspace^o$ is generically unique, it is an unstable quantity: slightly perturbing the target $\beta$ (or the current diagram $\alpha$) may yield a significantly different optimal matching $\pi^\star$. 
    It means that while ${\rm Dist}$ is differentiable almost everywhere on $\pdspace^o$, the gradient \eqref{eq:grad_dist} is not a stable quantity, which can be an issue in practical applications (e.g.,~instability in gradient descents). 
    One way of mitigating this is to add an \emph{entropic regularization term} in \eqref{eq:dist_to_target} following the computational optimal transport literature \cite{cuturi2013sinkhorn,peyre2019computational}. 
    This approach has been adapted to PDs in \cite{lacombe2018large,lacombe2023homogeneous} and used in the context of topological optimization in \cite{hiraoka2024topological}.
\end{example}

\begin{example}[Singleton loss]\label{ex:singleton}
	Given a target point $q_0\in \bR^2$, a persistence diagram $\alpha\in \pdspace$, and a point $p_0\in\alpha$, the \emph{singleton loss} is defined as:
	\begin{equation}
	{\rm Singleton}_{q_0}(\alpha, p_0) = \|p_0-q_0\|.
	\end{equation}
Intuitively, minimizing this loss amounts to pushing a specific persistence diagram point $p_0$ in the direction of a specific target point $q_0$. This loss can be understood as the simplest version of the distance to a target persistence diagram $\alpha_0$, where the target persistence diagram is comprised of only one point $\alpha_0=\{q_0\}$, and where the matching between $\alpha$ and $\alpha_0$ is imposed (as one forces $p_0\in\alpha$ to move towards $q_0$). This loss is particularly relevant in the context of \emph{big-step} gradient descent, as it allows to drastically speed up topological optimization, see Section~\ref{subsec:BSGS}.
\end{example}

Aside from these topological losses directly mapping a diagram in $\cD$ to a value in $\bR$, another family of topological losses can be obtained by composing two maps $L=\ell\circ\Phi$: a first \emph{vectorization} map $\Phi:\cD \to N$ where $N$ is a Hilbert space or simply the Euclidean space $\bR^d$, and a second map $\ell \colon N\to\bR$ from that linear space to $\bR$ which comes from the problem at hand: mean squared error for regression tasks, cross-entropy in classification tasks, etc. 
The gradient of $\ell$ is defined in the usual sense, and the gradient of $L$ can then be obtained by composing it with a lift of the differential of the vectorization $\Phi$, defined as in \cref{subsec:pd_diff}.

The vectorization step is a convenient step as it turns the PDs (which live in a non-flat space \cite{turner2014frechet}) into vectors on which standard machine learning pipelines can be used.
The most popular vectorizations of PDs are the \emph{linear} ones: they include, e.g., persistence surfaces and their variations \cite{adams2017persistence, chen2015statistical,kusano2016persistence,reininghaus2015stable} and persistence landscapes \cite{bubenik2015statistical,chazal2014stochastic}.

\begin{example}[Linear vectorizations of persistence diagrams]\label{ex:linear_reps}
    Given a map $\phi \colon \groundspace \to N$ ($N$ being a Hilbert or Euclidean space), one can define a vectorization $\Phi \colon \pdspace^o \to N$ as $\Phi \colon \alpha \mapsto \sum_{x \in \alpha} \phi(x)$. 
    Such vectorizations are called \emph{linear} as, if one encodes a diagram $\alpha$ as a measure $\mu \coloneqq \sum_{x} \delta_x$ where $\delta_x$ denote the Dirac mass at $x$ in the open half-plane $\groundspace$, this map is the linear map $\mu \mapsto \int \phi \dd \mu$---see \cite[\S 5.1]{divol2019understanding} for details. 

    For such vectorizations, assuming that the map $\phi$ is $C^1$, it follows that the differential of $\Phi$ at $\alpha \in \pdspace^o$ can be identified with 
    \begin{equation}
        \sum_{i=1}^m \dd \phi((b_i,d_i)),
    \end{equation}
    for some lift $(b_i,d_i)_{i=1}^m \in \bR^{2m}$ of $\alpha$. 
\end{example}

Note that linear vectorizations were initially introduced as fixed maps that depend on user-defined parameters, that is, each $\phi \colon \Omega\to N$ is a parametrized map $\phi=\phi_\eta$, with parameters $\eta$ controlling the weights to be assigned to each PD point (depending on, e.g., its distance to the diagonal $\partial\Omega$). 
As tuning $\eta$ is usually done with cross-validation, which can be expensive in running time, a subsequent body of work has recently emerged about \emph{learning} $\eta$: either with deep learning using specialized neural network architectures~\cite{hofer2017deep, hoferLearningRepresentationsPersistence2019, carriere2020perslay, zhao2020persistence, kim2020efficient}, with codebooks obtained by, e.g., running $k$-means on the training PDs~\cite{royer2021atol}, or with pre-defined template functions fitted on the training set~\cite{pereaApproximatingContinuousFunctions2023}.

\begin{remark}[Kernel methods to vectorize persistence diagrams.]
Another standard way of deriving vectorizations is with \emph{kernel methods}: in that case, the map $\Phi$ is defined implicitly through a \emph{kernel} $k$ representing the scalar product of some (implicit) Hilbert space $\mathcal H$: $k(\alpha,\beta) \coloneqq \langle \Phi(\alpha), \Phi(\beta)\rangle_{\mathcal H}$ for any $\alpha,\beta\in\cD$. The most popular kernel are the Gaussian-like ones: $k(\alpha,\beta) \coloneqq \exp(-d(\alpha,\beta)/\sigma)$, $\sigma>0$, where $d$ is a distance between PDs. Note that not all distances induce valid Gaussian kernels for PDs (i.e., such that there exists an implicit corresponding Hilbert space), only the so-called \emph{conditionally negative semi-definite}\footnote{I.e., those distances that satisfy $\sum_i c_i = 0 \Rightarrow \sum_i\sum_j c_i c_j d(\alpha_i,\alpha_j)\leq 0$, for any PDs $\{\alpha_i\}_i$ and coefficients $\{c_i\}_i$.} (CNSD) do. While the diagram distances $\FG_q$ are unfortunately not CNSD, interpreting PDs as discrete measures as in \cref{ex:linear_reps} allows to borrow tools from other domains (such as optimal transport and information geometry) in order to design CNSD metrics: examples include the sliced Wasserstein distance~\cite{Carriere2017b} and the Fisher information metric~\cite{le2018persistence}. Moreover, similar to the explicit linear representations, parametrized kernels $k=k_\eta$ can also be learned at training time, using, e.g., metric learning~\cite{Zhao2019}.
\end{remark}

\begin{remark}[Automatic differentiation]
    Since persistence diagrams can be identified with their lifts when it comes to compute gradients, one can in practice resorts on automatic differentiation to compute gradients of maps defined from $\pdspace$ to $\bR$, simply by actually defining them as maps from $\bR^{2m}$ to $\bR$ (note that $m$ should be allows to vary in general) in a framework compatible with automatic differentiation (e.g.,~\texttt{PyTorch} \cite{paszke2017automatic}, \texttt{TensorFlow} \cite{abadi2016tensorflow}, \texttt{JAX} \cite{jax2018github}, etc.).
\end{remark}

\subsection{Persistence diagrams of parametrized families of filtrations}

In addition to topological losses, most topological optimization problems involve some $r$-differentiable PD-valued maps $\Pers$ (for a given homology dimension) defined on a parameter manifold $M$ (see beginning of Section~\ref{sec:differentiability}). 
As $\theta$ often parametrizes filtration values assigned to simplicial complexes, the goal of this section is to present in more details the smoothness associated to such filtration maps.

\begin{definition} \label{def:parametrized_filtrations}
    Let $K$ be a finite simplicial complex and $M$ be a smooth manifold. 
    A map $F \colon M \to \bR^{|K|}$ is said to be a \emph{parametrized family of filtrations} if for any $\theta \in M$ and $\sigma, \sigma' \in K$ 
    with $\sigma \subseteq \sigma'$, one has $[F(\theta)]_\sigma \le [F(\theta)]_{\sigma'}$, where $[v]_\sigma$ denotes the entry of a vector $v\in\bR^{|K|}$ at the position of $\sigma$ (upon using a fixed, arbitrary ordering of the simplices of $K$). 
    In other words, a parametrized family of filtrations is a map $F \colon M \to \Filt_K$ (after identifying $\bR^K$ and $\bR^{|K|}$).
\end{definition}

Note that a parametrized family of filtrations is called just a parametrization in \cite{LOT2022}.
Now, let $M$ be a smooth manifold, and $F \colon M \to \Filt_K$ be a parametrized family of filtrations of class $C^r$. 
In what follows, we consider a PD-valued map of the form:
\begin{equation}\label{eq:B_p}
    \PersMap_p \colon M \xrightarrow{F} \Filt_K \xrightarrow{\Pers_p} \pdspace.
\end{equation}
Note that from now on, the domain of $\Pers$ is $\Filt_K$ (and not $M$ directly, as in the previous sections).
We shall show that $\PersMap_p$ is $r$-differentiable on a generic (open dense) subset of $M$.
First, we state a local differentiability result. 

\begin{theorem}[{\cite[Thm.~4.7]{LOT2022}}]
    Let $\theta \in M$. 
    Suppose $F \colon M \to \Filt_K$ is of class $C^r$ on some open neighborhood of $\theta$, and that $F(\theta) \sim F(\theta')$ (as per \cref{def:order_eq}) for all $\theta' \in U$.
    Then $\PersMap_p$ is $r$-differentiable at $\theta$.
\end{theorem}

\begin{proof}
    By \cref{proposition:order_persistence_pairs}, the persistence pairs $P_p$ and the unpaired simplices $U_p$ for $F(\theta)$ remain unchanged for all (ordering equivalent) $F(\theta'), \theta' \in U$.
    Hence, within $U$, one can write:
    \begin{equation}
        \PersMap_p(\theta') = \{ ([F(\theta')]_\sigma, [F(\theta')]_{\sigma'}) \}_{(\sigma,\sigma') \in P_p} \cup \{ ([F(\theta')]_\tau, \infty) \}_{\tau \in U_p} 
    \end{equation}
    for any $\theta' \in U$. 
    This gives a $C^r$ local coordinate system for $\PersMap_p$ at $x$. 
\end{proof}

Now, we state a global generic differentiability result. 
Set 
\begin{equation}
    \tilde{M} \coloneqq \lc \theta \in M \relmid 
    \begin{aligned}
        & \exists \text{ open neighborhood $U_\theta$ of $\theta$ such that} \\
        & \text{$F(\theta') \sim F(\theta)$ for all $\theta' \in U_\theta$} 
    \end{aligned}
    \rc.
\end{equation}
Then one can prove that $\tilde{M}$ is generic (open dense) in $M$ (see \cite[Lem.~4.10]{LOT2022}).

\begin{theorem}[{\cite[Thm.~4.9]{LOT2022}}]
    Suppose $F \colon M \to \Filt_K$ is of class $C^r$ on some open subset $U$ of $M$. 
    Then $\PersMap_p$ is $r$-differentiable on $U \cap \tilde{M}$.
\end{theorem}

\begin{proposition}[{\cite[Prop.~4.14]{LOT2022}}]\label{prop:gradpd}
    Let $\theta \in U \cap \tilde{M}$. 
    Let $(\sigma_{1},\sigma'_{1}), \dots , (\sigma_{m},\sigma'_{m})$ be the persistence pairs and $\tau_{1},\dots, \tau_{n}$ be the unpaired simplices of $F(\theta)$. 
    Then, the map: 
    \begin{equation}
        \tilde{\PersMap}_p \colon \theta' \mapsto 
		\lb ([ F(\theta') ]_{\sigma_i}, [ F(\theta') ]_{\sigma'_i})_{i=1}^m, ( [F(\theta')]_{\tau_j})_{j=1}^n \rb
    \end{equation}
    is a local $C^r$ lift of $\PersMap_p$ at $\theta$, and the corresponding differential is:
    \begin{equation}
        \dd_{\theta, \tilde{\PersMap}_p} \PersMap_p(\cdot) = 
        \lb ( [\dd_\theta F(\cdot)]_{\sigma_i}, [\dd_\theta F(\cdot)]_{\sigma'_i} )_{i=1}^m, ( [\dd_\theta F(\cdot)]_{\tau_j} )_{j=1}^n \rb.
    \end{equation}
\end{proposition}

Intuitively, the differential of $\PersMap_p$ is simply obtained by computing the differential of the filtration map first, and then picking the entries of the paired and unpaired simplices. 
An important property of the space of filtrations $\Filt_K$ is that it admits a \emph{Whitney stratification}.

\begin{definition}\label{def:stratification_of_manifold}
	Let $M$ be a subset of $\bR^d$. 
	A \emph{Whitney stratification} $\cS=\{ M_i \}_{i \in I}$ of $M$ is a locally finite partition by connected (not necessarily closed) smooth submanifolds $M_i$, called {\em strata}, satisfying the following conditions:
	\begin{itemize}
		\item (Frontier) For each stratum $M_i \in \cS$, the set $\overline{M_i} \setminus M_i$ is a union of strata. 
		\item (Condition B) Let $M',M''\in \cS$ be two strata, and $\theta \in M'$. 
		If two sequences $\{\theta'_k\}_k \subset M'$ and $\{\theta''_k\}_k \subset M''$ both converge to $\theta$, if the line connecting $\theta'_k$ and $\theta''_k$ converges to some line $l$, and if $T_{\theta''_k}M''$ converges to some plane $T$, then $T$ contains $l$.
	\end{itemize}
    A stratum with the maximal dimension is called a \emph{top-dimensional stratum}.
    Eventually, given $\theta \in M$, a stratum $M_i \in \cS$ is said to be \emph{incident} to $\theta$ if $\theta$ belongs to its closure: $\theta \in \overline{M_i}$.
\end{definition}

Let $\Omega(\Filt_K)$ be the set of equivalence classes  with respect to the ordering equivalence $\sim$.
Then, $\Omega(\Filt_K)$ is a Whitney stratification of $\Filt_K$ by semi-algebraic subsets, whose stratum are separated by a family $\mathcal F$ of hyperplanes, defined with $\mathcal F \coloneqq \{ v\in\bR^{|K|}\,|\,[v]_{\sigma} = [v]_{\sigma'}\}_{\sigma \neq \sigma' \in K}$.
Thanks to this fact, one can extend the local lift of $\PersMap_p$ in \cref{prop:gradpd} to a global lift.

\begin{proposition}[{\cite[Prop.~4.23]{LOT2022}}]\label{prop:permpd}
    Let $K$ be a simplicial complex of dimension $d$. 
    For $0 \le p \le d$, there exist integers $m_p,n_p$ such that $\sum_{p=0}^d (2m_p+n_p)=|K|$ and a map $\Perm \colon \Filt_K \to \prod_{p=0}^d \bR^{2m_p} \times \bR^{n_p} \cong \bR^{|K|}$ satisfying the following:
    \begin{enumerate}
        \item the restriction $\Perm|_S$ to each stratum $S$ of $\Omega(\Filt_K)$ is a permutation map, 
        \item the following diagram commutes
        \begin{equation}
        \xymatrix{
            \Filt_K \ar[r]^-{\Perm} \ar[rd]_-{\Pers} & \prod_{p=0}^d \bR^{2m_p} \times \bR^{n_p} \ar[d]^{\prod_{p=0}^d Q_{m_p,n_p}} \\
            & \pdspace^{d+1},
        }
        \end{equation}
        where $\pdspace^{d+1}$ denotes the $(d+1)$-th Cartesian power of $\pdspace$. 
    \end{enumerate}
\end{proposition}

Consider the parametrized PD-valued map:
\begin{equation}
    \PersMap \colon M \xrightarrow{F} \Filt_K \xrightarrow{\Pers} \pdspace^{d+1}; \theta \mapsto F(\theta) \mapsto (\Dgm_p(F(\theta)))_{p=0}^d.
\end{equation}

\begin{corollary}[{\cite[Cor.~4.24]{LOT2022}}]\label{cor:global_lift}
    The map:
    \begin{equation}
        \tilde{\PersMap} \colon M \to \prod_{p=0}^d \bR^{2m_p} \times \bR^{n_p}; \quad \theta \mapsto \Perm(F(\theta))
    \end{equation}
    is a global lift of $\PersMap$, i.e., $Q \circ \tilde{\PersMap}=\PersMap$ on $M$, where $Q=\prod_{p=0}^d Q_{m_p,n_p}$.
\end{corollary}

In words, \cref{cor:global_lift} ensures that global lifts for PD-valued maps associated to parametrized families of filtrations can be obtained by computing all the permutations (i.e., the indices of the paired and unpaired simplices) of the different stratum of $\Omega(\Filt_K)$. This lift is global, as, for a given parameter $\theta$, it suffices to identify the strata to which $F(\theta)$ belongs to, and then to apply the corresponding permutation of the entries of $\bR^{|K|}$ to obtain the persistence diagram. This observation is key in the implementation of several gradient descent schemes for topological optimization (see for instance ``Sampling strata'' paragraph in Section~\ref{subsec:SGS}).

\subsection{Stratified filtrations and directional differentiability}\label{subsec:strat_filt}

In the previous section, we have seen that global lifts of PD-valued maps can be computed through a stratification of the filtration space $\Filt_K$. In this section, we
now study the case where the domain $M$ of the filtration map $F \colon M\to\Filt_K$ itself can be stratified, as this is the case for most common filtrations in TDA.

\begin{definition}
    Let $M$ and $N$ be manifolds endowed with stratifications $\cS_M$ and $\cS_N$, respectively (see \cref{def:stratification_of_manifold}).
    A map $f \colon M \to N$ is said to be \emph{weakly stratified} if for any $N_j \in \cS_N$ the preimage $f^{-1}(N_j)$ is a union of strata in $\cS_M$.

    A function $f \colon M \to \bR$ is said to be \emph{Whitney stratifiable} if its graph admits a Whitney stratification.
\end{definition}

\begin{proposition}[{\cite[Prop.~4.16]{LOT2022}}]\label{prop:semi_sub}
    Let $F \colon M \to \Filt_K$ be a continuous parameterized family of filtrations. 
    Suppose that $M$ is a semi-algebraic (resp.~compact subanalytic) set in $\bR^d$ and $F$ is a semi-algebraic (resp.~subanalytic) map.
    Then, there is a stratification of $M$ by semi-algebraic (resp.~subanalytic) sets such that the restriction of $F$ to each stratum is $C^\infty$.
\end{proposition}

\begin{theorem}[{\cite[Thm.~4.19]{LOT2022}}]\label{thm:stratified_differentiable}
    Let $M$ be a manifold endowed with a Whitney stratification $\cS$.
    Let $F \colon M \to \Filt_K$ be a continuous parameterized family of filtrations such that:
    \begin{enumerate}
        \item $F$ is a weakly stratified map with respect to $\cS$ and $\Omega(\Filt_K)$,
        \item the restriction of $F$ to each stratum in $\cS$ is $C^r$, and 
        \item for any $\theta \in M$ and any stratum $M_i$ incident to $\theta$, there is an open neighborhood $U$ of $\theta$ such that $F|_{M_i \cap U}$ extends to a $C^r$ map $U \to \bR^K$.
    \end{enumerate}
    Then, at any $\theta \in M$, the PD-valued map $\PersMap_p \colon M \to \pdspace$ (defined in~\cref{eq:B_p}) is $r$-differentiable along each stratum that is incident to $\theta$.
\end{theorem}

\begin{corollary}[{\cite[Cor.~4.20]{LOT2022}}]
    Under the assumptions of \cref{prop:semi_sub}, there is a Whitney stratification of $M$ by semi-algebraic (resp.\ subanalytic) subsets, such that $\PersMap_p$ is $\infty$-differentiable on the top-dimensional strata (see \cref{def:stratification_of_manifold}). 
    If furthermore $F$ is $C^r$, then $\PersMap_p$ is everywhere $r$-differentiable along incident strata.
\end{corollary}

Combining everything, one can finally obtain the differentiability properties of the composite map $\mathcal L = L\circ\PersMap_p$ in the case of parametrized families of filtrations.

We eventually end this section with an additional notion of regularity on the filtration with respect to the stratification, namely stating the existence of a $C^2$ extension. 
It will be used to define the gradient sampling approach of \cite{LCLO2023gradient} presented in \cref{sec:gradient_descent}.

\begin{definition}\label{def:stratifiably_smooth}
	A map $\cL \colon M \to \bR$ is said to be \emph{stratifiably smooth} if there exists a Whitney stratification $\cS$ of $M$ such that for each top-dimensional stratum $M_i \in \cS$, 
	the restriction $\cL|_{M_i}$ admits an extension $\cL_i$ of class $C^2$ in a neighborhood of $M_i$.
\end{definition}

\begin{proposition}\label{prop:2diff}
    Let $L \colon \pdspace \to \bR$ be a $2$-differentiable map, and $F \colon M \to \Filt_K$ be a parameterized family of filtrations 
    satisfying the conditions of \cref{thm:stratified_differentiable} with $r=2$.
    Then, the composite function $\mathcal L = L \circ \Pers_p \circ F$ is stratifiably smooth. 
\end{proposition}

\subsection{Examples of stratified filtrations and their differentials}\label{subsec:example_filtrations}

As in Section~\ref{subsec:example_losses}, we now review the differentiability properties of the common filtrations in TDA presented in Section~\ref{subsec:background}, and we explicit their differentials.

\begin{example}[Vietoris--Rips filtration, \cref{def:vr}]
\label{example:vr_stratified}
    Let $\Delta_n$ is the complete simplicial complex on $n$ vertices, i.e., that comprises all the faces of the $(n-1)$-dimensional simplex. 
    The family of Vietoris--Rips filtrations on (ordered) point clouds of $n$ points $(x_1, \dots, x_n) \in (\bR^d)^n \eqqcolon M$ is the semi-algebraic parametrized family of filtrations:
    \begin{align}
        F \colon M  \to  \bR^{|\Delta_n|} = \bR^{2^n-1},
    \end{align}
    defined, for any ordered point cloud $X=(x_1, \dots, x_n) \in M$ and any simplex $\sigma \subseteq \{1, \dots, n\}$, by:
    \begin{align}
        [F(X)]_\sigma = \max_{i,j \in \sigma} \frac 12 \| x_i - x_j \|.
    \end{align}
    One can easily check that the permutation induced by $\Pers$ is constant on the strata whose boundaries are the subspaces $S_{i,j,k,l} = \{(x_1, \dots, x_n) \subset\bR^d \,|\, \| x_i - x_j \| = \| x_k - x_l \| \}$ over all the $4$-tuples $(i,j,k,l)$ such that at least three of the four indices $i,j,k,l$ are distinct. Then, 
    for a specific simplex $\sigma$ such that $[F(X)]_\sigma = \max_{i,j \in \sigma} \frac 12 \| x_i - x_j \|=\|x_{i^*} - x_{j^*}\|$, one can check that 
    the differential of $F$ is equal to $[\dd_X F(\cdot)]_\sigma=\langle \nabla_X [F(X)]_\sigma, \cdot\rangle$ with:
    \begin{equation}
        \nabla_X [F(X)]_\sigma=[{\bf 0},\dots,{\bf 0},\underbrace{\frac{x_{i^*}-x_{j^*}}{\|x_{i^*}-x_{j^*}\|}}_{{\rm index}\ i^*},{\bf 0},\dots,{\bf 0},\underbrace{\frac{x_{j^*}-x_{i^*}}{\|x_{i^*}-x_{j^*}\|}}_{{\rm index}\ j^*},{\bf 0},\dots,{\bf 0}]\in(\bR^d)^n.
    \end{equation}
    Finally, within a given stratum, computing the differential of $\Pers$ simply amounts to permuting the entries of $\dd_X F(\cdot)$ according to the paired and unpaired simplices of the stratum.

    This example naturally extends to general Vietoris--Rips filtrations for metric spaces in the following way. 
    Let $M \subset \cM_n(\bR)$ be the set of $n \times n$ symmetric matrices with non-negative entries and $0$ on the diagonal.
    This is a semi-algebraic subset of the space of $n$-by-$n$ matrices $\cM_n(\bR) \simeq \bR^{n^2}$, of dimension $m = (n-1)(n-2)/2$. 
    The map $F \colon M \to \bR^{|\Delta_n|} = \bR^{2^n-1}$ defined by
    $[F(A)]_\sigma = \max_{i,j \in \sigma} a_{i,j}$ for any $A = (a_{i,j})_{1 \leq i,j \leq n} \in M$, is a semi-algebraic family of filtrations. 
    Moreover, the set of strata can be chosen to be the set of matrices with at least two equal entries.
\end{example}

\begin{remark}
    The differentiability of the Vietoris--Rips filtrations was first considered in \cite{gameiro2016continuation}, which raised the first problem in the TDA literature that involved differentiating persistence diagrams.    
    The problem is as follows: 
    given an ordered point cloud $X=(x_1, \ldots, x_n) \in (\bR^d)^n$ and target persistence diagram $\alpha_0$, move $X$ continuously in order to make the persistence diagram $\Dgm_{\rm VR}(X)$ (computed with the Vietoris--Rips filtration, with $X$ regarded as a usual point cloud) closer to $\alpha^\ast$. 
    The authors proposed the so-called \emph{continuation method} to solve this problem, which is based on the Newton--Raphson method.
    
    We modify their approach and propose to (try to) build an ordered point cloud $X$ with a prescribed diagram $\alpha^*$, with the Jacobian $D_X \Pers$ of the map $\Pers \colon (\bR^d)^n \ni X \mapsto \DgmRips(X) \in \pdspace$, in the following way. 
    Denoting by $X_0$ the initial point cloud, we repeatedly update an ordered point cloud by $X_k + \gamma_k (D_X \Pers(X_k))^\dagger v_k$ for some vector $v_k$ in each step $k$, where $\gamma_k$ is a step size. 
    Here $v_k$ is defined with an optimal partial matching between $\Pers(X_k)$ and $\alpha^*$ with respect to $\FG_2$.
    Unlike the original approach in \cite{gameiro2016continuation}, we recompute the partial matching in each step, which enables to adaptively obtain the direction to update the point cloud.
    The detail of our implementation is presented in \cref{alg:continuation}.

    \begin{algorithm}
    \caption{Continuation of Point Clouds via Persistence Diagrams \cite{gameiro2016continuation}}
    \label{alg:continuation}
    \begin{algorithmic}[1]
        \STATE{\textbf{Input:} Initial ordered point cloud $X_0$, target diagram $\alpha^\star$, order of diagram distance $q$, number of steps $N$, a sequence of step sizes $\{\gamma_k\}_{k=0}^{N-1}$}
        \FOR{$k = 0$ to $N-1$}
        \STATE Set $\alpha_k \gets \Pers(X_k)$;  \quad \# Recall $\Pers(\cdot) = \DgmRips(\cdot)$ here
        \STATE Take $\pi_k$ be an optimal partial matching between $\alpha_k$ and $\alpha^\star$ with respect to $\FG_q$ (see \cref{def:dist_wasserstein});
        \STATE Take an arbitrary ordering $(y[1],\dots,y[m])$ of the points in $\alpha_k$ and regard it as a vector;
        \STATE Define a vector $v_k$ by setting the $i$-th element as 
        \begin{equation}
            v_k[i] \gets \pi_k(y[i]) -y[i] \in \bR^2;
        \end{equation}
        \STATE Compute Jacobian $J_k \gets D\Pers(X_k)$;    
        \STATE Update an ordered point cloud using the pseudo-inverse $J_k^\dagger$ by
        \begin{equation}
            X_{k+1} \gets X_k + \gamma_k J_k^\dagger v_k;
        \end{equation}
        \ENDFOR
        \RETURN $X_N$
    \end{algorithmic}
\end{algorithm}

\end{remark}

\begin{example}[Weighted Rips filtration]\label{ex:wrfilt}
Weighted Rips filtrations are a generalization of Vietoris--Rips filtrations where weights are assigned to the vertices of the complete simplicial complex $\Delta_n$. 
Given a function $f \colon \bR^d \to \bR$, the family of weighted Rips filtrations $F \colon M = (\bR^d)^n  \to  \bR^{|\Delta_n|} = \bR^{2^n-1}$ associated 
with $f$ is defined, for any $X=\{x_1, \dots, x_n\} \in M$ and any simplex $\sigma \subseteq \{1, \dots, n \}$, by:
\begin{equation}\label{eq:triangle}
[F(X)]_\sigma =
\begin{cases}
 2 f(x_j) & (\sigma = \{j\}); \\
 \max \{2 f(x_i), 2 f(x_j), \| x_i - x_j \| + f(x_i) + 
 f(x_j)\}, & (\sigma = \{i,j\}, i \not = j); \\
 \max \{[F(X)]_{\{i,j\}}\,|\, i,j \in \sigma\} & (|\sigma| \geq 3).
\end{cases}
\end{equation}
\end{example}

Since Euclidean distances and $\max$ function are semi-algebraic, this family of filtrations is semi-algebraic as soon as the weight function $f$ is semi-algebraic. Moreover, the differential of $F$ can be easily computed in a way that is similar to the Vietoris--Rips filtrations (one just needs to distinguish among the three cases in Equation~\eqref{eq:triangle}, and to add the differentials $\dd_{x_i}f(\cdot)$, $\dd_{x_j}f(\cdot)$ of $f$ when needed). The differential of $\Pers$ then follows by applying the stratum-specific permutation.

This example easily extends to the case where the weight function depends on the point cloud $X = \{ x_1, \dots, x_n \}$, i.e., when the weight at vertex $y$ is 
defined by $f(x,y)$ with $f \colon (\bR^d)^n \times \bR^d \to \bR$. 
A particular example of such a family is given by the so-called DTM filtration \cite{anai2020dtm}, where $f(x,y)$ is the average distance from $y$ to its $k$-nearest neighbors in $X$. 
In this case, $f$ is semi-algebraic, and the family of DTM filtrations is semi-algebraic.   

\begin{example}[Sublevel sets filtrations, \cref{example:height_filtration,ex:graph_filt,ex:img_filt}]\label{ex:sublevel}
    Let $K$ be a simplicial complex with $n$ vertices $v_1, \dots, v_n$. Any real-valued function $f$ defined on the vertices of $K$ can be represented as a vector $[f(v_1), \dots, f(v_n)]$ in $\bR^n$. 
    The family of sublevel sets filtrations $F \colon M=\bR^n \to \Filt_K$ of functions on the vertices of $K$ is defined by $[F(f)]_\sigma = \max_{i \in \sigma} f_i$ for any $f = [f_1, \dots, f_n] \in M$ and any simplex $\sigma \subseteq \{1, \dots, n\}$.
    This filtration is also known as the \emph{lower-star filtration} of $f$, and is a very general way of designing filtrations: it includes, for instance, the height filtration (\cref{example:height_filtration}), filtrations on graphs (using graph nodes as vertices, \cref{ex:graph_filt}) and filtrations on images (using pixel corners as vertices, \cref{ex:img_filt}). 
    The function $F$ is obviously semi-algebraic, and the set of strata $S$ can be chosen as $S = \bigcup_{1 \leq i < j \leq n} \{ [f_1, \dots, f_n] \in M \,|\, f_i = f_j \}$. Moreover, 
    for a specific simplex $\sigma$ such that $[F(f)]_\sigma = \max_{i \in \sigma} f_i = f_{i^*}$,
    the differential of $F$ is simply obtained as $[\dd_f F(\cdot)]_\sigma=[\cdot]_{i^*}$, i.e., picking the $i^*$-th entry of the vector. The differential of $\Pers$ then follows by applying the stratum-specific permutation.
    
    Note however that in many cases, the function $f$ itself depends on some parameters $f=f(\theta)$, and the differential of $f$ with respect to its own parameters has to be incorporated in the chain rule. This happens, e.g., when graph filtrations are learned by a graph neural network, as in~\cite{horn2021topological}, or when image filtrations are learned by a convolutional neural network, as in~\cite{Barbarani2024}.
\end{example}

\section{Optimizing persistence-based objective functions}\label{sec:gradient_descent}

In Section~\ref{sec:differentiability}, we presented the differentiability properties of the persistence maps $\Pers$ and filtration maps $F$. Now, in this section, our aim is to present the different 
optimization schemes that have been proposed in the literature, as well as their theoretical guarantees, for minimizing composite topological losses $\cL$ of the form 
$\cL \coloneqq L \circ \Pers \circ F \colon M \to \bR$, defined on a parameter manifold $M$.
As one typically performs topological optimization while targeting downstream machine learning tasks, and due to the prominence of gradient-based optimization through \emph{automatic differentiation} in deep learning, it is not surprising that most if not all techniques developed in the TDA literature consist on performing schemes akin to gradient descent. 

Below, we detail the most standard stochastic gradient descent scheme~\cite{Carriere2021a, subgradient} in Section~\ref{subsec:vanilla},
that we call {\em vanilla} gradient descent, as well as its convergence properties. 
Then, we explain two other schemes that both aim at improving it: 
the {\em stratified} gradient descent~\cite{LCLO2023gradient} with its improved convergence properties in Section~\ref{subsec:SGS},
and the {\em big-step} gradient descent~\cite{nigmetov2024topological} with its faster empirical convergence in Section~\ref{subsec:BSGS}.
Finally, we describe two extensions that can both be applied to any of the three previous schemes:
{\em downsampling}~\cite{solomon2020fast,Wagner2021}, and {\em diffeomorphic interpolations}~\cite{Carriere2024}, in Section~\ref{subsec:extensions}.
See \Cref{fig:exgrad} for a schematic overview of these methods.

\begin{figure}[ht]
	\centering
	\includegraphics[width=.95\textwidth]{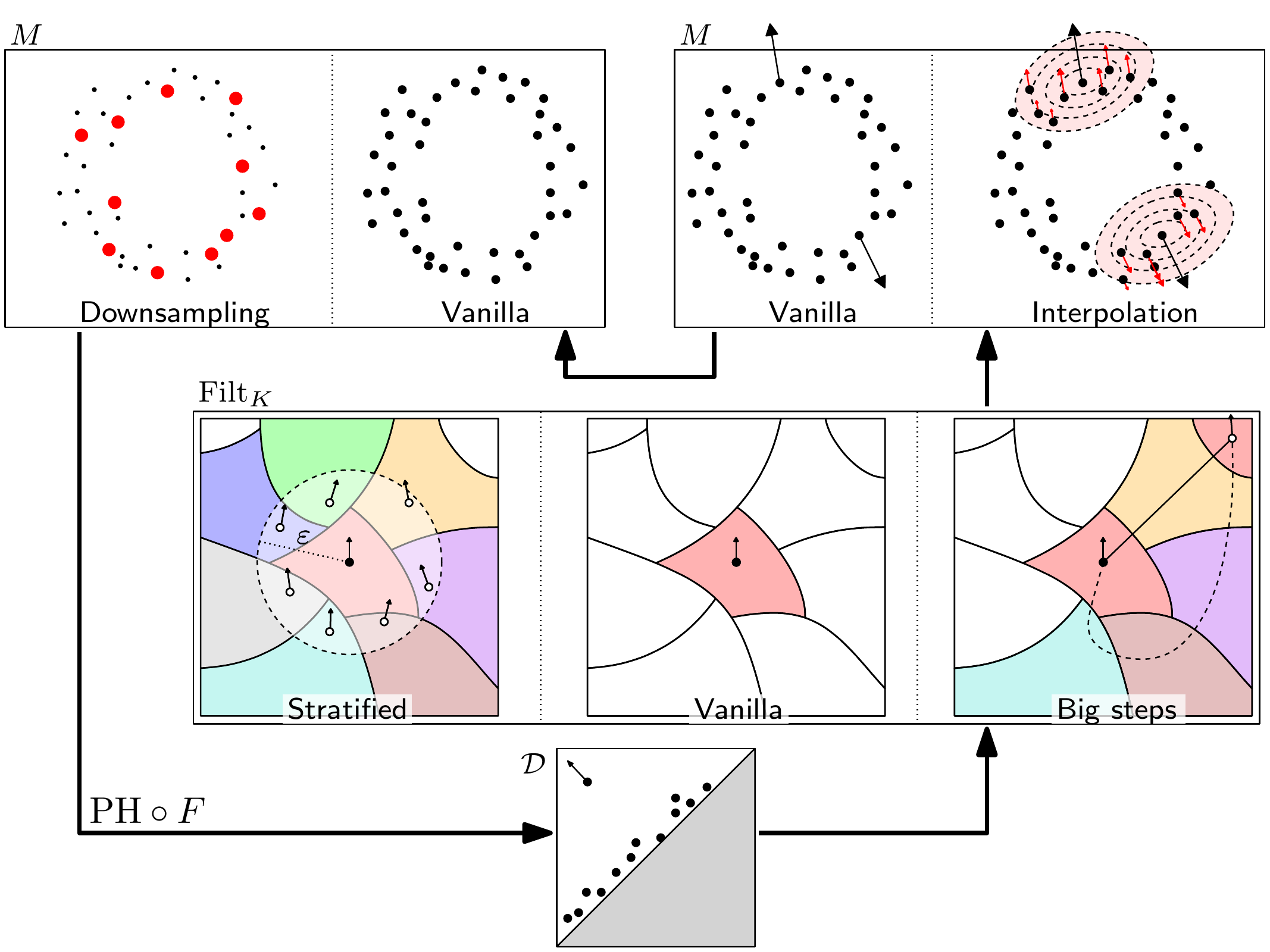}
	\caption{\label{fig:exgrad} Schematic illustration of the different gradient schemes  presented in this survey (for point cloud filtrations).
	The space of filtrations $\Filt_K$ is displayed in two dimensions with strata separated with solid lines. While stratified gradient descent aggregates the gradients associated to strata intersecting an open ball around the current estimate, big step gradient descent proposes a way to directly jump to the next iterate without going through other, intermediate strata (as vanilla gradient descent---the dashed line---would). Two possible extensions are also displayed: downsampling, which uses a smaller dataset in order to save computation time for computing persistence diagrams in the forward pass, and diffeomorphic interpolation, which extends the gradients computed in the backward pass to a vector field (defined everywhere in the parameter space) using kernels.}
\end{figure}

\subsection{Vanilla gradient descent}\label{subsec:vanilla}

In this section, we present the most natural way to minimize a topological loss based on PDs. 
Let $K$ be a simplicial complex, $M$ be an open subset of $\bR^d$, and $F \colon M \to \Filt_K$ be a parametrized family of filtrations of class $C^\infty$.
Let $L \colon \pdspace \to \bR$ be a 2-differentiable map, and define a loss function
$\cL$ with $\cL\coloneqq L \circ \Pers \circ F \colon M \to \bR$.
As explained in~\cref{sec:differentiability} (see, e.g.,~\cref{rem:rademacher} and~\cref{prop:2diff}), 
the loss $\cL$ is often smooth {\em almost everywhere} on $M$, so one can still define its {\em Clarke subgradient}: 
\begin{definition}
The {\em Clarke subgradient} of $\cL$ at $\theta \in M$ is:
\begin{equation}
    \partial \cL(\theta) \coloneqq \mathrm{Conv} \left\{ \lim_{\theta_i \to \theta} \nabla \cL(\theta_i) \relmid \text{$\cL$ is differentiable at $\theta_i$} \right\}.
\end{equation}
\end{definition}

In the following, we will refer to any element $g_\theta\in\partial \cL(\theta)$ as a {\em vanilla gradient} for $\cL$ at $\theta$ (as opposed to the 
other gradient alternatives presented in Sections~\ref{subsec:SGS} and~\ref{subsec:BSGS}). 
Stochastic (sub)gradient descent is then performed by progressively updating an estimation $\theta_k$ of a minimizer $\theta^*$ of $\cL$ on $U$ with: 
\begin{equation} \label{eq:stochastic_gradient}
    \theta_{k+1} = \theta_k - \gamma_k ( g_{\theta_k} + \zeta_k),\quad g_{\theta_k} \in \partial \cL(\theta_k),
\end{equation}
where the sequence $\{\gamma_k\}_k$ is the {\em learning rate} of the process, and $\{\zeta_k\}_k$ is a sequence of random variables
used for robustifying the estimations; e.g.~$\zeta_k \overset{\text{i.i.d.}}{\sim} \cN(0,I)$.

\paragraph*{Implementation.} In practice, it can be computed simply by composing the differential
of $L$ with the differential of the persistence map---which amounts to finding the permutation associated to the current persistence
pairs and unpaired simplices computed with~\cref{alg:persistent-pairs},\footnote{Most TDA libraries include such functions, see for instance the \texttt{persistence\_pairs()} method in the \Gudhi\ library (\url{https://gudhi.inria.fr/python/latest/}).} as per~\cref{prop:gradpd}---and the differential of the filtration map (see, e.g., Section~\ref{subsec:example_filtrations} for examples of such differentials).
See~\cref{alg:VanillaGradient}.

\begin{algorithm}[ht]
	\caption{$\mathtt{VanillaGradient}(\theta)$}
    \label{alg:VanillaGradient}
	\begin{algorithmic}
		\STATE{\textbf{Input}: Current iterate $\theta\in M$} 
        \STATE \text{\# First step is to identify the critical simplices associated to the persistence map in the current}
        \text{filtration}
		\STATE $((\sigma_1, \sigma_2),\dots,(\sigma_{2m+1},\sigma_{2m})), (\tau_1,\dots,\tau_n)\gets \texttt{PersistencePairs}(F(\theta))$;
        \STATE \text{\# Use these critical simplices as a local lift and compute the differential as per \cref{prop:gradpd}}
        \STATE \text{\# Recall that $[v]_{\sigma}$ denotes the entry of $v$ at the position of $\sigma$ in an arbitrary simplex ordering}
        \STATE $\alpha \gets [[F(\theta)]_{\sigma_1},\dots,[F(\theta)]_{\sigma_{2m}},[F(\theta)]_{\tau_{1}},\dots,[F(\theta)]_{\tau_{n}}]\in\bR^{2m+n}$;
        \STATE $\nabla_\theta [\alpha]_i \gets \nabla_\theta [F(\theta)]_{\sigma_i}, \forall 1\leq i\leq 2m$; 
        \STATE $\nabla_\theta [\alpha]_{2m+j} \gets \nabla_\theta [F(\theta)]_{\tau_j}, \forall 1\leq j\leq n$; 
        \STATE \text{\# Finally, compute the vanilla gradient with the chain rule as per~\cref{prop:chain}}
        \STATE $g_\theta \gets \sum_{i=1}^{2m+n} \frac{\partial L}{\partial [\alpha]_i}([\alpha]_i)\cdot \nabla_\theta [\alpha]_i$;
        \RETURN $g_\theta$
	\end{algorithmic}
\end{algorithm}

Usually, the gradients $\nabla_\theta [F(\theta)]_\sigma$ and the partial derivatives $\frac{\partial L}{\partial [\alpha]_i}$ can be either computed explicitly 
as they are easy to derive from the usual losses and filtrations in the TDA literature 
(see Sections~\ref{subsec:example_losses} and~\ref{subsec:example_filtrations}),
or they can be automatically computed with backpropagation using standard deep learning libraries such as \texttt{PyTorch}, \texttt{TensorFlow} or \texttt{jax},
as long as the filtration map $F$ and the loss $L$ are implemented with operations from these libraries.  

\paragraph*{Convergence guarantees.} It has been shown in~\cite{subgradient} that under mild technical conditions on the sequences $\{\alpha_k\}_k$ and $\{\zeta_k\}_k$, 
stochastic subgradient descent converges almost surely to a critical point of $\cL$ as soon as $\cL$ is locally Lipschitz. 
More precisely, consider the following standard assumptions (see \cite[Assumption~C]{subgradient}):
\begin{enumerate}
    \item[(1)] for any $k$, $\alpha_k \geq 0$, $\sum_{k=1}^\infty \alpha_k = +\infty$, and $\sum_{k=1}^\infty \alpha_k^2 < +\infty$;
    \item[(2)] $\sup_k \| \theta_k \| < +\infty$, almost surely;
    \item[(3)] denoting by $\cF_k$ the increasing sequence of $\sigma$-algebras $\cF_k = \sigma(\{\theta_j, g_{\theta_j},\zeta_j \,|\, j < k\})$, there exists a function $p \colon \bR^d \to \bR$ which is bounded on bounded sets such that almost surely, for any $k$, 
    \begin{align}
        \bE[\zeta_k|\cF_k] = 0 \quad \text{and} \quad 
    \bE[\| \zeta_k \|^2 |\cF_k] < p(\theta_k).
    \end{align}
\end{enumerate}

Assumption~(1) is easily satisfied, for example, by taking $\alpha_k = 1/k$. 
Assumption~(2) is usually easy to check for most of the functions $\cL$ encountered in practice (see Section~\ref{subsec:example_losses}). 
Assumption~(3) is also a standard condition, which states that, conditioned upon the past, the variables $\zeta_k$ have zero mean and controlled moments; e.g., this can be achieved by taking a sequence of independent and centered variables with bounded variance that are also independent of the $\theta_k$'s and $g_{\theta_k}$'s.   

Under these assumptions, the following result is an immediate consequence of Corollary~5.9 in \cite{subgradient}. 

\begin{theorem}\label{thm:convergence_guarantees}
    Let $K$ be a simplicial complex, $M \subseteq \bR^d$ be an open subset, and $F \colon M \to \Filt_K$ be a Whitney stratifiable parametrized family of filtrations of $K$. 
    Let $L \colon \pdspace \to \bR$ be a function such that $L \circ Q_{m_p,n_p}$ is Whitney stratifiable (e.g.,~$L$ is $2$-differentiable in the sense of \cref{def:differentiable_fromD}) and $\cL \coloneqq  L \circ \Pers \circ F$ is locally Lipschitz.  
    Then, under Assumptions~(1), (2) and (3) above, almost surely the limit points of the sequence $\{\theta_k\}_k$ obtained from the iterations of \cref{eq:stochastic_gradient} 
    are critical points of $\cL$ and the sequence $\{\cL(\theta_k)\}_k$ converges. 
\end{theorem}

While easy to define and implement, the vanilla gradient also suffers from a few weaknesses. First, on the theoretical side,
while~\cref{thm:convergence_guarantees} ensures that stochastic subgradient descent converges, nothing is said about the quality
of the critical point it converges to, nor about the convergence rate, i.e., the descent speed. Second, on the practical
side, it is easy to see that updating points in a PD (by minimizing some loss) only influences the persistence pairs and unpaired simplices associated to them;
this in turn means that, for a given PD point, only at most two filtration values are updated accordingly, leading to sparse gradients and slow, erratic convergence.
In the following sections, we will explore two gradient alternatives and two gradient extensions that were designed to deal with these issues.

\subsection{Stratified gradient descent}\label{subsec:SGS}

A recent alternative to vanilla gradients was proposed in~\cite{LCLO2023gradient}. The main idea is to make use of the stratifications of $\Filt_K$ induced by the persistence and
filtration maps in order to design a more efficient gradient using a procedure similar to the well-known {\em gradient sampling} method.
Recall that gradient sampling at a point $x$ involves sampling points in an $\varepsilon$-neighborhood of $x$, and taking the vector with smallest norm in the convex hull
of the corresponding gradients as a new, smoother gradient (after renormalization). The key idea of this section is to incorporate the information encoded in the strata
of the persistence and filtration maps within this procedure. This in turn allows to guarantee that the limit points of the sequence of iterates $\{\theta_k\}_k$ is close 
to an approximate critical point, as characterized with {\em Goldstein subgradients}. Let $\cL \colon M \to \bR$ be a stratifiably smooth function (see~\cref{def:stratifiably_smooth}), with $M \subset \bR^d$. In this section, we also
 consider the following assumptions on $\cL$:
\begin{enumerate}
    \item[(A1)] $\cL \colon M \to \bR$ has bounded sublevel sets. \label{assump:bounded_sublevel}
    \item[(A2)] For any $\theta \in M$, we have an oracle checking whether $\cL$ is differentiable at $\theta$. :\label{assump:diff}
    \item[(A3)] For each $\theta \in M$ and $\varepsilon$-close top-dimensional stratum $M_i$, we have an oracle that 
    returns an element $\theta_i\in M_i$ with $\|\theta-\theta_i\| \le \varepsilon$.
\end{enumerate}

\begin{definition}\label{def:goldstein}
Let $\varepsilon >0$. The {\em Goldstein subgradient} of $\cL$ is
\begin{equation}
    \partial_\varepsilon \cL(\theta) \coloneqq \mathrm{Conv} \left\{ \lim_{\theta_i \to \theta'} \nabla \cL(\theta_i) \relmid \|\theta-\theta'\|\leq\varepsilon,\text{$\cL$ is differentiable at $\theta_i$} \right\}.
\end{equation}
Here, $\mathrm{Conv}$ denote the convex hull of a set.
A point $\theta$ is said to be $\varepsilon$-stationary if $0\in\partial_\varepsilon \cL(\theta)$ and it is $(\varepsilon,\eta)$-stationary, for some $\eta \geq 0$, if $d(0,\partial_\varepsilon \cL(\theta))\leq\eta$.
\end{definition}

The set of Goldstein subgradients at $\theta$ can be thought of as ``average'' (convex combinations) of nearby gradients, and is well defined even if $\cL$ is not differentiable at $\theta$. 
Note that local minimum (as long as local maximum and usual saddle points) satisfy $0 \in \partial_{\epsilon = 0} \cL(\theta)$. 

The whole point of stratified gradient descent is to use an approximation $\tilde\partial_\varepsilon\cL(\theta)$ (computed with assumption (A3) above) of $\partial_\varepsilon \cL(\theta)$, 
provided that $\cL$ is stratifiably smooth, to update the iterates.

\begin{definition}
Let $\varepsilon >0$. The $\varepsilon$-{\em stratified subgradient} of $\cL$ is
\begin{equation}
    \tilde\partial_\varepsilon \cL(\theta) \coloneqq \mathrm{Conv} \left\{ g_{\theta_i}\,|\, M_i\cap B(\theta,\varepsilon)\neq\varnothing \right\},
\end{equation}
where $\theta_i$ is an arbitrary point of $M_i$ such that $\|\theta-\theta_i\|\leq\varepsilon$ (obtained using assumption (A3) above), 
and $g_{\theta_i}$ is the vanilla gradient of $\cL$ at $\theta_i$. Moreover, the vector:
\begin{equation}
	\stratgrad_\theta \coloneqq {\rm argmin}_{g\in \tilde\partial_\varepsilon \cL(\theta)}\, \|g\|
\end{equation}
is called the \emph{descent direction} associated to $\tilde\partial_\varepsilon \cL(\theta)$.
\end{definition}

A key result of~\cite{LCLO2023gradient} is that using stratified subgradients allows to guarantee a loss decrease.

\begin{proposition}[{\cite[Proposition 4]{LCLO2023gradient}}]\label{prop:loss_decrease}
Assume $\cL$ is stratifiably smooth, with Lipschitz constant $C>0$. Let $\theta$ be non-stationary, and $\stratgrad_\theta$ be the descent direction associated to $\tilde\partial_\varepsilon\cL(\theta)$. Finally, let $\beta > 0$.
Then:
\begin{enumerate}
\item[$(i)$] for small enough $\varepsilon$, one has $\varepsilon\le \frac{1-\beta}{2C}\|\stratgrad_\theta\|$, and
\item[$(ii)$] for such $\varepsilon$, and any $\alpha\le \frac{\varepsilon}{\|\stratgrad_\theta\|}$, one has:
\begin{equation}
\cL(\theta-\alpha\stratgrad_\theta) \leq \cL(\theta) - \beta\cdot\alpha\cdot\|\stratgrad_\theta\|^2.
\end{equation}
\end{enumerate}
\end{proposition}

\paragraph{Implementation.} We now detail through a series of implementation how to compute faithful approximations of $\varepsilon$-stratified subgradients.
The most basic step is~\cref{alg:StratifiedGradient},
that is basically an instance of gradient sampling, except that the neighborhood sample points
are not chosen randomly within $B(\theta,\epsilon)$ but will instead be picked precisely in $\varepsilon$-close top-dimensional strata (if any), leveraging the knowledge we have about the structure of stratification of the persistence map $\Pers$. 
See \cref{fig:exgrad}.

\begin{algorithm}[ht]
	\caption{$\mathtt{StratifiedGradient}(\theta,S)$}
    \label{alg:StratifiedGradient}
	\begin{algorithmic}
		\STATE{\textbf{Input}: Current iterate $\theta\in M$, set of samples $S=\{\theta_1,\dots,\theta_m\}$}
		\STATE $g_\theta \gets \texttt{VanillaGradient}(\theta)$;
        \STATE $G \gets \{g_\theta\}$; \quad \text{\# Initialize with the vanilla gradient at current iterate}
		\FOR{$1\leq i\leq m$}
            \STATE $g_{\theta_i}\gets \texttt{VanillaGradient}(\theta_i)$;
            \STATE $G\gets G\cup \{g_{\theta_i}\}$; \quad \text{\# Add the vanilla gradients of samples in neighboring strata}
        \ENDFOR
        \STATE Solve quadratic minimization problem $\tstratgrad_\theta  \coloneqq  \text{argmin}\{\|g\|^2\,|\,g\in \mathrm{Conv}(G)\}$; 
        \RETURN $\tstratgrad_\theta$
	\end{algorithmic}
\end{algorithm}

Then, the final stratified gradient $\stratgrad_\theta$ is obtained by progressively refining the gradient $\tstratgrad_\theta$ of~\cref{alg:StratifiedGradient} over smaller and smaller 
neighborhoods with~\cref{alg:ControlledStratifiedGradient}, until the neighborhood size satisfies item $(i)$ in~\cref{prop:loss_decrease}.

\begin{algorithm}[ht]
	\caption{$\mathtt{ControlledStratifiedGradient}(\theta,\varepsilon,m,\gamma,\beta,C,\eta)$}
    \label{alg:ControlledStratifiedGradient} 
	\begin{algorithmic}
		\STATE{\textbf{Input}: Current iterate $\theta$, neighborhood size $\varepsilon$, sampling size $m$,
                               neighborhood decrease rate $0<\gamma<1$, loss decrease constant $\beta$, (upper bound on the) Lipschitz constant $C$ of $\cL$,
                               norm threshold $\eta$}
        \STATE $\tilde\varepsilon\gets\varepsilon$;
		\STATE $S \gets \texttt{Sample}(\theta,\tilde\varepsilon,m)$;
        \STATE $\stratgrad_\theta\gets\texttt{StratifiedGradient}(\theta,S)$;
		\WHILE{$\tilde\varepsilon > \frac{1-\beta}{2C}\|\stratgrad_\theta\|$} 
            \STATE \text{\# Ensures the neighborhood size is small enough so that the loss decrease can be controlled}
            \text{as per~\cref{prop:loss_decrease}}
            \IF{$\|\stratgrad_\theta\|\leq\eta$}
                \RETURN $\alpha=0$ and $\stratgrad_\theta$
            \ELSE
                \STATE $\tilde\varepsilon \gets \gamma \tilde\varepsilon$;
         		\STATE $S \gets \texttt{Sample}(\theta,\tilde\varepsilon,m)$;
                \STATE $\stratgrad_\theta\gets\texttt{StratifiedGradient}(\theta,S)$;
            \ENDIF
        \ENDWHILE
		\RETURN $\stratgrad_\theta$ and $\alpha \coloneqq \tilde\varepsilon/\|\stratgrad_\theta\|$
	\end{algorithmic}
\end{algorithm}

Finally, the corresponding stratified gradient descent is given by~\cref{alg:StratifiedGradientDescent}.

\begin{algorithm}[ht]
	\caption{$\mathtt{StratifiedGradientDescent}(\theta_0,\varepsilon,m,\gamma,\beta,C,\eta)$}
    \label{alg:StratifiedGradientDescent} 
	\begin{algorithmic}
		\STATE{\textbf{Input}: Initial iterate~$\theta_0$, neighborhood size $\varepsilon$, sampling size $m$,
                               neighborhood decrease rate $0<\gamma<1$, loss decrease constant $\beta$, (upper bound on the) Lipschitz constant $C$ of $\cL$,
                               norm threshold $\eta$}
		\STATE $k \gets 0$;
        \STATE $\stratgrad_{\theta_{-1}}\gets +\infty\cdot {\bf 1}$;
		\WHILE{$\|\stratgrad_{\theta_{k-1}}\| > \eta$}
		\STATE $(\stratgrad_{\theta_k}, \alpha_k) \gets \mathtt{ControlledStratifiedGradient}(\theta_k,\varepsilon,m,\gamma,\beta,C,\eta)$;
        \STATE Apply~\cref{eq:stochastic_gradient} with $\alpha_k,\stratgrad_{\theta_k}$ and $\zeta_k=0$;
        \STATE $k\gets k+1$;
		\ENDWHILE
        \RETURN $\theta_k$
	\end{algorithmic}
\end{algorithm}

\begin{remark}\label{rem:general_algo}
The stratified gradient descent scheme presented in~\cref{alg:StratifiedGradientDescent} is a simplified version 
of the original one~\cite[Algorithm 2]{LCLO2023gradient},
which does not require to know the Lipschitz constant $C$ of $\cL$ in advance (see also Remark 5 and Algorithm 6 in~\cite{LCLO2023gradient}). 
However, such (upper bounds on) Lipschitz constants can be derived for most standard TDA losses (see Section~\ref{subsec:example_losses}).
\end{remark}

\begin{remark}\label{rem:make_diff}
Strictly speaking, nothing guarantees that the sequence of iterates $\{\theta_k\}_k$ belongs to the interiors of
some strata (on which $\cL$ is differentiable), which is required to ensure that the convergence properties of 
stratified gradient descent hold in~\cite{LCLO2023gradient}. As such, the authors of~\cite{LCLO2023gradient} use small perturbations of the iterates
to make sure that $\cL$ is differentiable on them (using assumption (A2) above). However, since obtaining iterates that end up exactly at the boundary between several
strata happens with probability zero in most applications, practical implementations are usually oblivious to this fact, so we leave this 
issue aside in this article and refer to~\cite{LCLO2023gradient} for more details.  
\end{remark}

\begin{remark}
Instead of progressively decreasing $\tilde\varepsilon$ as in~\cref{alg:ControlledStratifiedGradient}, one can reduce it in constant time with~\cref{alg:ConstantTimeControlledStratifiedGradient},
which is a slight variation of~\cref{alg:ControlledStratifiedGradient}. In words, the idea is simply to fix a set $S_0$ of samples, and to use subsets of this set instead of sampling new points
when reducing $\tilde\varepsilon$. The reason that the condition $\tilde\varepsilon \leq \frac{1-\beta}{2C}\|\stratgrad_\theta\|$ can be satisfied after only one iteration is that it suffices 
to decrease $\tilde\varepsilon$ to $\frac{1-\beta}{2C}\|\stratgrad_\theta\|$ directly (if the condition was not satisfied initially); indeed, as the new set of samples $S$ satisfies $S\subseteq S_0$, it follows 
that the corresponding gradient sets (in~\cref{alg:StratifiedGradient}) also satisfy $G\subseteq G_0$, and thus 
$\mathrm{Conv}(G)\subseteq\mathrm{Conv}(G_0)$ and $\min\{\|g\|^2\,|\,g\in \mathrm{Conv}(G)\}\ge \min\{\|g\|^2\,|\,g\in \mathrm{Conv}(G_0)\}$.
The norm of the stratified gradient can thus only increase when using subsets of the initial sample set, ensuring that the condition becomes satisfied.
However, as $\stratgrad_\theta$ is computed with a smaller number of samples (as $S\subseteq S_0$) by~\cref{alg:StratifiedGradient}, it becomes a rougher estimation of 
$\mathrm{argmin}\{\|g\|^2\,|\,g\in \tilde\partial_{\tilde\varepsilon}\cL(\theta)\}$ and it thus becomes less likely for~\cref{prop:loss_decrease}
to apply. 

\begin{algorithm}[ht]
	\caption{$\mathtt{ConstantTimeControlledStratifiedGradient}(\theta,\varepsilon,m,\gamma,\beta,C,\eta)$}
    \label{alg:ConstantTimeControlledStratifiedGradient} 
	\begin{algorithmic}
		\STATE{\textbf{Input}: Current iterate $\theta$, neighborhood size $\varepsilon$, sampling size $m$,
                               loss decrease constant $\beta$, (upper bound on the) Lipschitz constant $C$ of $\cL$,
                               norm threshold $\eta$}
        \STATE $\tilde\varepsilon\gets\varepsilon$;
		\STATE $S_0 \gets \texttt{Sample}(\theta,\tilde\varepsilon,m)$;
        \STATE $\stratgrad_\theta\gets\texttt{StratifiedGradient}(\theta,S_0)$;
		\IF{$\tilde\varepsilon > \frac{1-\beta}{2C}\|\stratgrad_\theta\|$} 
            \IF{$\|\stratgrad_\theta\|\leq\eta$}
                \RETURN $\alpha=0$ and $\stratgrad_\theta$
            \ELSE
                \STATE $\tilde\varepsilon \gets \frac{1-\beta}{2C}\|\stratgrad_\theta\|$;
         		\STATE $S \gets \{\tilde\theta\in S_0\,|\,\|\tilde\theta-\theta\|\leq\tilde\varepsilon\}$;              
                \STATE $\stratgrad_\theta\gets\texttt{StratifiedGradient}(\theta,S)$;
            \ENDIF
        \ENDIF
		\RETURN $\stratgrad_\theta$ and $\alpha \coloneqq \tilde\varepsilon/\|\stratgrad_\theta\|$
	\end{algorithmic}
\end{algorithm}
\end{remark}

\paragraph{Sampling strata.}
There are several ways to implement the sampling algorithm \texttt{Sample} required by~\cref{alg:ControlledStratifiedGradient} and~\cref{alg:ConstantTimeControlledStratifiedGradient}.
Recall that the main objective for \texttt{Sample} is the ability to sample points in $\varepsilon$-close strata. In the case of the persistence map $\Pers$, the strata are 
completely characterized by permutations of the filtration values as per~\cref{prop:permpd}. Hence, a practical way of exploring neighboring strata is by looking at shortest paths on the permutahedron:
first identify the permutation associated to the current iterate $\theta$, retrieve the corresponding node $V_\theta$ on the permutahedron, and then explore other permutations / strata by looking 
for the permutahedron nodes that are the closest to $V_\theta$ with, e.g., Dijkstra's algorithm or diffusion. 
This will allow to explore the strata that can be obtained from the current one by a minimal number of transpositions first. 
Moreover, it is also possible to store the persistence pairs and unpaired simplices corresponding to the permutations that have already been
visited, so that the vanilla gradients need not necessarily be recomputed from scratch at every iteration of~\cref{alg:StratifiedGradientDescent}. See also~\cite[Section 4.2]{LCLO2023gradient}. 
Recall however that when the filtration map $F$ is also stratified, the strata of $\Pers$ and those of $F$ have to be combined in the sampling algorithm \texttt{Sample}.

\paragraph*{Convergence guarantees.} In short, it can be shown that stratified gradient descent 
converges to an $(\varepsilon,\eta)$-stationary point in a finite number of iterations, which is a substantial
improvement over the guarantees of vanilla gradient descent from a theoretical viewpoint.

\begin{theorem}[{\cite[Thm.~7]{LCLO2023gradient}}]
    If $\eta>0$, then~\cref{alg:StratifiedGradientDescent} produces an $(\varepsilon,\eta)$-stationary point using at most $O\lb \frac{1}{\eta \min(\varepsilon, \eta)} \rb$ iterations.
\end{theorem}

The above theorem provides explicit conditions ensuring the convergence of stochastic subgradient descent for functions of persistence. 
The main criterion to be checked is the local Lipschitz condition for $\cL$, which is guaranteed as soon as $F$ and $L$ are Lipschitz.

\subsection{Big-step gradient descent}\label{subsec:BSGS}

In this section, we detail another gradient, that we call the {\em big-step} gradient $\bigstepgrad_\theta$.
It was introduced in \cite{nigmetov2024topological} for minimizing singleton losses (see \cref{ex:singleton}),
in a much faster way than vanilla gradient descent. The key idea of the big-step gradient is to move, for a single
update on PD points, a much bigger set of simplices (and their filtration values) than the associated persistence pairs and unpaired simplices only.
In terms of stratified spaces, this provides a way to skip a lot of strata with one ``jump'' in a single gradient descent iteration---see \cref{fig:exgrad}. This also allows to make the singleton loss---associated to a persistence pair $q_0=(\sigma,\tau)$---more ``global'': indeed, since big-step gradient descent always keeps $\sigma$ and $\tau$ paired together by design, it makes sense to use a loss that is defined with respect to the PD point $q_0$, which always exists---on the other hand, if gradient iterates were allowed to go through different strata 
that possibly do not pair $\sigma$ and $\tau$ together,
the singleton loss would 
be ill-defined.

\paragraph*{Minimizing singleton loss.}
Let $K$ be a simplicial complex, and $f\colon K\to \bR$ be a filtration of $K$.
In this paragraph, we focus on minimizing a singleton loss, that is, pushing one PD point $p_0$ towards a target point $q_0=(q_x,q_y)$. 
Let $(\sigma, \tau)$ be the persistence pair corresponding to $p_0$. 
The goal is to update the value of $f(\sigma)$ towards $q_x$, and similarly for $f(\tau)$ and $q_y$, while ensuring that they stay matched together so that the singleton loss stays well-defined.

In other words, given a $p$-simplex $\tau$ and a target value $t$, we want to change the value of $f(\tau)$ towards $t$ without modifying the simplex $\sigma$ that is paired with $\tau$.
To this end, we want to find the $p$-simplices that belong to the set $X_\sigma$, defined as:
\begin{equation}\label{eq:movingset}
    X_\sigma \coloneqq \{\tau'\in \Sk_p(K)\,|\,\sigma\text{ becomes paired with }\tau'\text{ after swapping }\tau\text{ and }\tau'\text{ in the filtration}\},
\end{equation}
where $\Sk_p(K)$ is the $p$-skeleton of $K$ (see \cref{def:simpcomp}).

Indeed, if one sets $f(\tau')$ to $t$ directly for all $\tau' \in X_\sigma$, one can then safely move $f(\tau)$ to $t$ directly without modifying the simplex $\sigma$ that is paired with $\tau$.
This allows to modify much more simplices at the same time than the vanilla gradient, leading to much faster convergence
in the iterations of~\cref{eq:stochastic_gradient}.
Hence, one defines a new gradient as follows.

\begin{definition}
    The {\em big-step} gradient $\bigstepgrad_\theta$ associated to the singleton loss $L$ is obtained as the usual vanilla gradient computed with a modification of the differential of the map $L$, obtained by replacing the partial derivatives associated to every simplex $\tau'\in X_\sigma$ and every simplex $\sigma'\in X_\tau$ with the partial derivatives of the simplices $\tau$ and $\sigma$ respectively, where $(\sigma,\tau)$ is the persistence pair associated to $L$. 
\end{definition}

\paragraph*{Implementation.} 
The big-step gradient can be computed with~\cref{alg:BigStepGradient},\footnote{Note that if the parameter space $M$ is $\Filt_K$ itself, the updates of the partial derivatives prescribed by \cref{alg:BigStepGradient} need also be applied on the faces (if the simplex value is decreased) or cofaces (if the simplex value is increased) of the simplices in $X_\sigma$ and $X_\tau$ in order to ensure that $F(\theta)$ remains a filtration.} which is a variation of~\cref{alg:VanillaGradient}.
Gradient descent can then be performed with~\cref{eq:stochastic_gradient}.

\begin{algorithm}[ht]
	\caption{$\mathtt{BigStepGradient}(\theta)$}
    \label{alg:BigStepGradient}
	\begin{algorithmic}
		\STATE{\textbf{Input}: Current iterate $\theta\in M$, critical simplices $(\sigma,\tau)$ associated to singleton loss $L$}
		\STATE $((\sigma_1, \sigma_2),\dots,(\sigma_{2m+1},\sigma_{2m})), (\tau_1,\dots,\tau_n)\gets \texttt{PersistencePairs}(F(\theta))$;
		\STATE $\alpha \gets [[F(\theta)]_{\sigma_1},\dots,[F(\theta)]_{\sigma_{2m}},[F(\theta)]_{\tau_{1}},\dots,[F(\theta)]_{\tau_{n}}]\in\bR^{2m+n}$;
		\STATE $\nabla_\theta [\alpha]_i \gets \nabla_\theta [F(\theta)]_{\sigma_i}, \forall 1\leq i\leq 2m$; 
		\STATE $\nabla_\theta [\alpha]_{2m+j} \gets \nabla_\theta [F(\theta)]_{\tau_j}, \forall 1\leq j\leq n$; 
		\STATE $t_1\gets \frac{\partial L}{\partial [\alpha]_{\sigma}}, t_2\gets\frac{\partial L}{\partial [\alpha]_{\tau}}$;
        \STATE \text{\# Compute simplices that one needs to move jointly with $\sigma$ and $\tau$ to preserve persistence pairs}
        \text{(see Equation~\eqref{eq:movingset} and \cref{alg:critical-set-naive})}
        \STATE $X_{\sigma}\gets \texttt{MovingSet}(\sigma,\tau,t_1,F(\theta))$;
        \STATE $X_{\tau}\gets \texttt{MovingSet}(\tau, \sigma,t_2,F(\theta))$;
        \STATE \text{\# Update partial derivatives of simplices to match match them with partial derivatives of $\sigma, \tau$}
        \FOR{$\tau'\in X_{\sigma}$}
			\STATE $\frac{\partial L}{\partial [\alpha]_{\tau'}}\gets t_1$;
        \ENDFOR
        \FOR{$\sigma'\in X_{\tau}$}
			\STATE $\frac{\partial L}{\partial [\alpha]_{\sigma'}}\gets t_2$;
        \ENDFOR
        \STATE $\bigstepgrad_\theta \gets \sum_{i=1}^{2m+n} \frac{\partial L}{\partial [\alpha]_i}([\alpha]_i)\cdot \nabla_\theta [\alpha]_i$;
        \RETURN $\bigstepgrad_\theta$
	\end{algorithmic}
\end{algorithm}

We now present how to compute the set $X_\sigma$ with the function $\texttt{MovingSet}$.
To ease notation, let $f \coloneqq F(\theta)$ be the current filtration.
Suppose there are $|\Sk_p(K)|=m_p$ simplices $\tau_1 (= \tau), \ldots, \tau_{m_p}$
~(sorted in ascending order according to the proximity of their values under $f$ to $f(\tau)$) between $f(\tau)$ and $t$.
We define the set $X_\sigma^k~(k=1, \ldots, m_p)$ as the set of $p$-simplices as follows:
\begin{itemize}
    \item For $k=1$, $X_\sigma^1 \coloneqq \{\tau_1\}$.
    \item For $k=2, \ldots, m_p$:
    \begin{equation}
        X_\sigma^k \coloneqq
        \begin{cases}
            X_\sigma^{k-1} & \text{if $\tau$ is paired with $\sigma$ for the filtration $f_k$}\\
            X_\sigma^{k-1} \cup \{\tau_k\} & \text{otherwise}
        \end{cases}
    \end{equation}
    where the order induced by $f_k$ is obtained from that of $f$ with the following modifications: 
    \begin{enumerate}
        \item For any $\tau' \in X_\sigma^{k-1}$, the order of $\tau_k$ and $\tau'$ is swapped. 
        \item For any $\sigma_1, \sigma_2\in K$ with $\sigma_1, \sigma_2\notin X_\sigma^{k-1}\cup\{\tau_k\}$, the order of $\sigma_1$ and $\sigma_2$ is preserved. 
    \end{enumerate}
\end{itemize}
Now, the set $X_\sigma \coloneqq X_\sigma^{m_p}$ is the set of $p$-simplices defined in~\cref{eq:movingset}. 

One way for obtaining $X_\sigma^{m_p}$ is to iteratively compute $X_\sigma^k$ based on the definition for $k = 1, \ldots, m_p$, as described in \cref{alg:critical-set-naive}. 
To examine whether the persistence pair $(\sigma, \tau)$ changes due to swapping simplex $\tau_k$ with the simplices in the set $X_\sigma^{k-1}$, one can use the algorithm proposed in \cite{Cohen-Steiner2006}.
This algorithm computes the persistence pair changes in $O(|K|)$ time, for a single swap of adjacent simplices (in the order induced by the filtration). 
In the $k$-th for-loop, we need to swap simplices $2(k-1)$ times, and thus the overall number of swaps is $O(m_p^2)$ in the worst case.
Therefore, the overall computational complexity is $O(|K| \cdot m_p^2)$.

\begin{algorithm}[th]
    \caption{$\mathtt{MovingSet}(\tau,\sigma,t,f)$}
    \label{alg:critical-set-naive}
    \begin{algorithmic}
  	\STATE{\bf Input}: Target simplex $\tau\in \Sk_p(K)$, the simplex $\sigma$ paired with $\tau$, target value $t$, filtration $f$
  	\STATE $X_\sigma^1 \gets \{\tau\}$;
        \FOR{each $\tau_k\in \Sk_p(K)$ with $f(\tau_k)$ between $f(\tau)$ and $t$}
            \STATE transpose $\tau_k$ with each simplex in $X_\sigma^{k-1}$, updating the pairing using the algorithm in \cite{Cohen-Steiner2006}; --- ({\dag})
            \IF{$\tau_k$ becomes paired with $\sigma$}
                \STATE $X_\sigma^k \gets X_\sigma^{k-1}\cup\{\tau_k\}$;
                \STATE undo the transpositions in ({\dag}) so that $\tau_k$ is at the opposite end of $X_{\sigma}^k$; 
            \ELSE
                \STATE $X_\sigma^k \gets X_\sigma^{k-1}$;
            \ENDIF
        \ENDFOR
        \RETURN $X_\sigma^{m_p}$
    \end{algorithmic}
\end{algorithm}

To develop a faster algorithm for finding $X_\sigma^{m_p}$, the authors in \cite{nigmetov2024topological} provide a method to compute $X_\sigma^{m_p}$ in $O(m_p)$ time.
The key observation is that persistence pairing is computed by reducing the boundary matrix, which can be interpreted as finding decompositions $R_p = D_p V_p$, where matrices $R_p$ are reduced, meaning the lowest non-zeros in their columns appear in unique rows, and $V_p$ are upper-triangular invertible matrices (see~\cite[Section VII.1]{edelsbrunner2010computational}). 
Sticking with the notation of \cite{nigmetov2024topological}, we let $U_p = V_p^{-1}$.
Building on this idea, they prove the following theorem, which shows that we can compute $X_\sigma^{m_p}$ explicitly using such matrices $U_p, V_p$. 

\begin{theorem}
    Let $X_\sigma$ be the set of $p$-simplices computed by \cref{alg:critical-set-naive}.
    It holds
    \begin{equation}
        X_\sigma = \begin{cases}
            \{\tau'\mid t < f(\tau') < f(\tau), V_p[\tau', \tau] \neq 0\} & \text{if $\tau$ is a death simplex and $t < f(\tau)$},\\
            \{\tau'\mid f(\tau) < f(\tau') < t, U_p[\tau, \tau'] \neq 0\} & \text{if $\tau$ is a death simplex and $f(\tau) < t$}, \\
            \{\tau'\mid t < f(\tau') < f(\tau), U^\perp_p[\tau, \tau'] \neq 0\} & \text{if $\tau$ is a birth simplex and $t < f(\tau)$}, \\
            \{\tau'\mid f(\tau) < f(\tau') < t, V^\perp_p[\tau, \tau'] \neq 0\} & \text{if $\tau$ is a birth simplex and $f(\tau) < t$}. \\
        \end{cases}
    \end{equation}
\end{theorem}

\paragraph{Minimizing combined loss.}
Given a more general loss $L=\sum_{(p, q)\in \gamma}(p-q)^2$, where $\gamma\subseteq \alpha\times \alpha_0$ is a (fixed)
partial matching between the PD of the current estimate $\theta$ and a fixed, target PD $\alpha_0$, one difficulty
is that the different terms in the sum can induce several different target values for every simplex.
Hence, one has to reduce these multiple target values to a single one for each simplex.
In order to achieve this, the authors in \cite{nigmetov2024topological} propose the following heuristic.
Fix a simplex $\sigma\in K$.
Let $a\coloneqq f(\sigma)$ be the initial value, and $a_1, a_2, \ldots$ be the different target values induced by the singleton losses appearing in the combined loss for $\sigma$.
Then, they propose to choose the target value $a'$ for $\sigma$ (and the corresponding singleton loss) with:
\begin{equation}
    a' = a_{i^*}, \quad \text{where $i^*={\rm argmax}_i |a-a_i|$}.
\end{equation}
Intuitively, this advocates for using the largest push on every simplex---see~\cite[Section 3.7]{nigmetov2024topological} for a motivation of this choice.

\subsection{Gradient extensions}~\label{subsec:extensions}

In this section, we assume that some procedure for computing a gradient $(\theta,P,K)\mapsto G_{\theta,P,K}$ for $\cL$ at $\theta$ on the simplicial complex $K$ has been fixed (with $P$ being 
additional parameters depending on the procedure), whether it is the vanilla ($P=\varnothing$), stratified ($P=\{\varepsilon,m,\gamma,\beta,C,\eta\}$) or
big-step method ($P=\{\sigma, \tau\}$), and we present two extensions of these gradients, that is, two additional procedures that produce new gradients-like objects with additional benefits, such as being less sparse, more robust or more computationally efficient.

\subsubsection{Smoothing gradient with downsampling}\label{subsec:sample}

A simple idea to reduce the computational cost and sparsity of topological gradients is to compute them on several
smaller simplicial complexes of fixed and controlled sizes, and average the results. This is the approach
called {\em downsampling simplicial complexes}, 
advocated in~\cite{Wagner2021} and~\cite{solomon2020fast}. See~\cref{fig:exgrad}. In these works, two different ways are proposed for downsampling simplicial complexes.

\paragraph{Downsampling with subcomplexes.}
The first approach involves averaging gradients computed from subcomplexes of the initial simplicial complex.

\begin{definition}
    Let $K$ be a simplicial complex and $\cP(K)$ be a family of subcomplexes of $K$.
    Let $G_{\theta,P,K}$ be a gradient computation procedure.
    The {\em downsampled, or distributed gradient} is defined as:
    \begin{equation}
        \tilde G_{\theta,P,K} \coloneqq  \frac{1}{|\cP(K)|}\sum_{K'\in\cP(K)} G_{\theta,P,K'},
    \end{equation}
    where the gradients $G_{\theta,P,K'}$ were computed using the original filtration restricted to $K'$.
\end{definition}

A common example of family of subcomplexes is $\cP_n(K) \coloneqq \{K'\subseteq K\,|\, |\Sk_0(K')|=n \text{ and }\forall\sigma\in K, \Sk_0(\sigma)\subseteq\Sk_0(K')\Longrightarrow\sigma\in K'\}$, comprised of those maximal subcomplexes 
with exactly $n$ vertices. Such a family was used in~\cite{Wagner2021} for instance, in which the authors focused on the family of Vietoris--Rips
complexes computed from subsamples (of fixed size) of an initial point cloud. All such complexes can be seen as subcomplexes of the Vietoris--Rips complex of the full point cloud. 

\paragraph{Downsampling with nerves.} 
The second approach involves averaging gradients computed from {\em nerve complexes}.

\begin{definition}
    Let $K$ be a simplicial complex and $\cU$ be an open cover of $K$, that is, a family of open sets such that $K=\bigcup_{U\in\cU}U$.
    The {\em nerve} of $\cU$ is the simplicial complex $\cN(\cU)$ with vertices $\Sk_0(\cN(\cU))=\cU$ and simplices: 
    \begin{equation}
        \sigma=\{U_{i_0},\dots,U_{i_p}\}\in\cN(\cU)\Longleftrightarrow \bigcap_{j=0}^p U_{i_j}\neq\varnothing.
    \end{equation}
\end{definition}

\begin{definition}
Let $K$ be a simplicial complex and $\cC$ be a family of open covers of $K$.
Let $G_{\theta,P,K}$ be a gradient computation procedure.
Then, the {\em nerve gradient} is defined as:
\begin{equation}
    \tilde G_{\theta,P,K} \coloneqq  \frac{1}{|\cC|}\sum_{\cU\in\cC} G_{\theta,P,\cN(\cU)},
\end{equation}
where the gradients $G_{\theta,P,\cN(\cU)}$ were computed using a filtration of $\cN(\cU)$ derived from the original one.
\end{definition}

In~\cite{solomon2020fast}, the authors suggest several way for designing a filtration on $\cN(\cU)$. One possibility is, e.g., to
assign to every vertex $U_i$ of $\cN(\cU)$ the average of the original filtration values of the vertices of $K$ that belong to $U_i$,
or, more generally, any weighted average of such values such that the weights sum to one, and to extend to the whole nerve with lower-star (see~\cref{ex:sublevel}). Gradient smoothing can be made even stronger
by computing $G_{\theta,P,\cN(\cU)}$ as the integral over all such weighted averages.

\paragraph{Application to topological gradient descent.} There are several benefits for using either downsampled or nerve gradients in~\cref{eq:stochastic_gradient}. Indeed, 
they both tend to produce denser gradients (as they incorporate several gradients
computed on different complexes) in a much faster way than the original gradient procedures (as the size of complexes they are applied to are controlled).

\subsubsection{Extending gradient with diffeomorphic interpolation}
\label{subsubsec:diffeo}

This approach proposed in this section, introduced in \cite{Carriere2024}, proposes an alternative way to extend a gradient associated to the persistence diagram of a
\emph{geometric} simplicial complex $K$, that is, a simplicial complex embedded in $\bR^d$ with a filtration that is entirely parametrized by the location of the vertices of $K$, 
i.e., any parameter $\theta$ is of the form $\theta=\{x_1,\dots,x_n\}\in (\bR^d)^n$. 
Examples of such filtrations include the Vietoris--Rips, the \v Cech, or the height filtrations (see Section~\ref{subsec:background}). 
As explained before, 
any given topological loss $\cL$ yields 
a vanilla gradient $G_{\theta,\varnothing,K}\in\partial\cL(\theta)\in \bR^{n \times d}$ that is typically sparse, as its rows represent the~$x_i$'s, and are thus $0$ for any non critical vertex in the filtration.

The main idea conveyed by this approach is to build a vector field $\tilde G_{\theta,P,K} \colon \bR^d \to \bR^d$ that \emph{interpolates}
any topological gradient $G_{\theta,P,K}$ on its \emph{non-zero} entries. 
As they are many possible candidates, it is natural to seek the \emph{smoothest} way to interpolate the gradient, which is done by finding a vector field of minimal norm in a given Reproducing Kernel Hilbert Space (RKHS), hence the name of \emph{diffeomorphic interpolation}. 

\paragraph{Diffeomorphic interpolation.} Let $\theta = \{x_1,\dots,x_n\} \in \bR^{n \times d}$, let $I \subseteq \{1,\dots,n\}$ and consider a set of vectors $a_i \in \bR^d$ for $i \in I$. 
Let $k \colon \bR^d \times \bR^d \to \bR^{d \times d}$ be a matrix-valued kernel operator whose outputs are symmetric and positive definite, i.e.,~$\forall x,y, \alpha \in \bR^d$, $\alpha^T k(x,y) \alpha \geq 0$. 
This kernel induces an RKHS\footnote{In many applications, RKHS are restricted to spaces of functions valued in $\bR$, but the theory adapts faithfully to the more general setting of vector-valued maps.} $\cH \subset (\bR^d)^{\bR^d}$ whose elements are vector fields $V \colon \bR^d \to \bR^d$. 
Observe also that, as for any $\alpha,x \in \bR^d$, $V \in \cH \mapsto \alpha^T V(x) \in \bR$ is a (continuous) linear form, Riesz's representation theorem gives the existence of $k_x^\alpha \in \cH$ such that $\braket{k_x^\alpha,V}_{\cH} = \braket{\alpha,V(x)}$. 
The goal is thus to find $\tilde G_{\theta,P,K} \in \cH$ such that, for all $i \in I$, $\tilde G_{\theta,P,K}(x_i) = a_i$ (with $a_i$ being the non-zero entries of $G_{\theta,P,K}$) that would be as smooth as possible, 
i.e., of minimal norm, yielding the minimization problem:
\begin{equation}\text{minimize } \|V\|_\cH,\ \text{s.t. } V(x_i) = a_i,\ \forall i \in I.\end{equation}
The solution $\tilde G_{\theta,P,K}$ of this problem is the projection of $0$ onto the affine set $\{V \in H\,|\, V(x_i) = a_i, \forall i \in I\}$ and thus belongs to $\{V \in H\,|\, V(x_i) = 0, \forall i \in I\}^\perp$, 
and thus to $\{V \in H\,|\, \braket{k_{x_i}^{\alpha_i}, V}_\cH = 0,\ \forall i \in I,\alpha_i \in \bR^d\}^\perp$. 
Eventually $\tilde G_{\theta,P,K} \in \vspan(\{k_{x_i}^{\alpha_i}\,|\, i \in I\})$. 
This justifies to search for $\tilde G_{\theta,P,K}$ in the form of $\tilde G_{\theta,P,K}(x) = \sum_{i\in I} k(x,x_i) \alpha_i$, and the interpolation that it must satisfy yields
the following definition (see also~\cite[Theorem 8.8]{Younes2010}).

\begin{definition}
    Let $K$ be a geometric simplicial complex associated to $\theta=\{x_1,\dots,x_n\}\in(\bR^d)^n$, $k$ be a kernel on $\bR^d$, and $G_{\theta,P,K}$ be a gradient computation procedure.
    Let $\{a_i\}_{i\in I}$, $I\subseteq \{1,\dots,n\}$ be the collection of non-zero entries of $G_{\theta,P,K}$.
    Then, the {\em diffeomorphic gradient} is defined as:
    \begin{equation}
        \tilde G_{\theta,P,K}(x) = \sum_{i \in I} k(x,x_i) (\bK^{-1} a)_i,
    \end{equation}
    where $\bK$ is the block matrix $(k(x_i,x_j))_{i,j \in I}$ and $a = (a_i)_{i \in I}$.
\end{definition} 

In particular, $\tilde G_{\theta,P,K}$ inherits from the regularity of $k$ and will typically be a diffeomorphism. 
Note that $\tilde G_{\theta,P,K}$ can be understood as the convolution of $G_{\theta,P,K}$ with the kernel $k$, but involving a correction $\bK^{-1}$ guaranteeing that after the convolution, the interpolation constraint is satisfied. See~\cref{fig:exgrad}.

\begin{example}
    A natural choice of a kernel operator is to let $k$ be the Gaussian kernel defined by $k(x,y) \coloneqq \exp\left(-\frac{\|x-y\|^2}{2 \sigma^2}\right) I_d$ for some bandwidth $\sigma > 0$.
    In that setting, the expression of $\tilde{V}$ reduces to 
\begin{equation}\label{eq:expression_diffeo_simplified}
    \tilde G_{\theta,P,K}(x) = \sum_{i \in I} \rho_\sigma(\|x - x_i\|)  \alpha_i,
\end{equation}
where $\rho_\sigma(u) \coloneqq e^{-\frac{u^2}{2\sigma^2}}$, and $\alpha_i \coloneqq  (\bK^{-1} a)_i$ with $\bK = (\rho_\sigma(\|x_i - x_j\|)I_d)_{i,j \in I}$. 
\end{example}

\paragraph{Application to topological gradient descent.} 
Applying~\cref{eq:stochastic_gradient} with $\tilde G_{\theta,P,K}$ (instead on $G_{\theta,P,K}$) on the point cloud $\theta=\{x_1,\dots,x_n\}$ leads to several benefits:
\begin{itemize}
    \item From a computational perspective, the most expensive operation is to invert a $|I| \times |I|$ matrix, where $I$ denotes the non-zero entries of $G_{\theta,P,K}$ and is typically fairly small due to its sparsity (especially for the vanilla gradient $G_{\theta,\varnothing,K}$). 
    \item The diffeomorphism $\tilde G_{\theta,P,K}$ is defined on the \emph{whole space} $\bR^d$, not only on the set of vertices $\{x_1,\dots,x_n\}$. 
    It enables \emph{extrapolation} of the gradient to unseen points. This observation induces several useful consequences:
    \begin{enumerate}
    	\item Diffeomorphic gradients combine very well with distributed gradients (see Section~\ref{subsec:sample}): if $\{x_1,\dots,x_n\}$ is actually a subsample of a much larger point cloud $X$ 
    	(for which computing persistence diagrams is computationally prohibitive), then 
    	one can compute $\tilde G_{\theta,P,K}$ on $\{x_1,\dots,x_n\}$ only, and then use it in order to update the whole point cloud $X$.
    	\item Diffeomorphic gradients can be \emph{reused}: assume one has performed topological optimization on some dataset, and a new data point is given. In order to optimize it as well with vanilla gradient descent, one would have to redo topological optimization from scratch using the original dataset complemented with the new point: indeed, the vanilla gradients recorded in the previous optimization are not well-defined on the new data point. On the other hand, the sequence of diffeomorphic gradients from previous optimization can be reapplied at no additional cost as they are vector fields.
    \end{enumerate}
    \item In the regime $\alpha_k \to 0$, the update in \cref{eq:stochastic_gradient} yields the same decrease in the loss $\cL$ than an update with the vanilla gradient $G_{\theta,\varnothing,K}$. 
\end{itemize}

\begin{remark}
	When performing several iterations of topological optimization, one can collect the sequence of diffeomorphic gradients, and produce a \emph{discretized gradient flow} by composing these vector fields. 
    Note however that producing a flow from a sequence of vanilla gradients can also be done in different ways (i.e., without necessarily using diffeomorphic interpolations with kernels), the only required ingredient is to build a parametrized time varying vector field $v_t(\cdot ; \theta) \colon \bR^d \to \bR^d$ that interpolates, or at least approximate, the vanilla gradient whenever that one is non-zero. 
    An appealing alternative is thus to rely on \emph{Neural Ordinary Differential Equations} (NODEs), which are ODEs that are parametrized by neural network architectures, i.e.,~$\dot x(t) = v_t(x, \theta)$ where $v_t$ is a neural network taking $x \in \bR^d$ as input and $\theta$ denotes its parameters. 
    These neural nets are trained so that they induce ODEs which produce time-varying datasets that are as close as possible to some ground-truth datasets computed explicitly on a finite number of time points. If one measures these proximities between datasets with PDs (using, e.g., Vietoris--Rips filtrations), learning a NODE will produce a flow which tries to match the ground-truth PDs (on the provided time points). 
    Namely, the typical loss to be optimized would be $\theta \mapsto \FG( \Dgm(x_\theta(t = 1)), \alpha_{\mathrm{ref}})$, where $x_\theta(t=1)$ denotes the solution of the ODE at time $t=1$ (with parameters $\theta$), $\Dgm(x(t=1))$ is its Vietoris-Rips persistence diagram, and $\alpha_{\mathrm{ref}}$ is some target reference persistence diagram. 
\end{remark}

\subsection{Summary table}

\cref{tab:summary_table} synthesizes the guarantees and complexities associated to the gradient schemes and extensions presented in this section.
Recall that the baseline is the vanilla gradient, whose worst-case complexity is cubic with respect to the number of simplices considered $N$. 
In \cref{tab:summary_table}, we used the following notations: \begin{enumerate}
	\item $\Omega^{\texttt{Sample}}_\varepsilon$ is the complexity of sampling gradients on strata at distance $\varepsilon$ using the oracle associated to Assumption (A3) in Section~\ref{subsec:SGS},
	\item $\Omega^{\texttt{QP}}_{S_\varepsilon}$ is the complexity (that depends on the numerical solver) of solving the quadratic programming problem ${\rm argmin} \|g\|^2$ subject to $S_\varepsilon$ constraints (where $S_\varepsilon$ is the number of strata returned by the oracle),
	\item $\Omega^{\texttt{Inv}}_{d\cdot |I|}$ is the complexity of inverting a matrix of size $d\cdot |I|$ (where $d$ is the data dimension and $|I|$ is the number of non-zero entries of the gradient), and
	\item $n$ is the number of downsampled complexes with $N'$ simplices, and $m_p$ is the number of $p$-simplices, where $p$ is the dimension of the simplex associated to the singleton loss.
\end{enumerate}   

Note that both gradient schemes and both gradient extensions make the vanilla gradient denser, ensuring a smoother descent.
In terms of limitations, recall that interpolation requires geometric simplicial complexes, and that big-step gradients only work for singleton (or combined) losses.

\begin{table}[tbh]
    \centering
    \begin{tabular}{|c||c|c|}
        \hline 
		& Complexity & Guarantees (in addition to smoother descent)\\
        \hline
        \hline
         Vanilla & $O(N^3)$ & -\\ 
         \hline
         Stratified & $O(N^3) + O(\Omega^{\texttt{Sample}}_\varepsilon) + O(\Omega^{\texttt{QP}}_{S_\varepsilon})$ &  Non-asymptotic convergence \\
         Big-step & $O(N^3) + O( N \cdot m_p^2)$ & Empirically faster descent \\
         \hline
         Distributed & $O(n\cdot (N')^3)$ & Scalable + Faster iterations \\
         Diffeomorphic & $O(N^3) + O(\Omega^{\texttt{Inv}}_{d\cdot |I|})$ & Scalable  + Reusable gradients \\
         \hline
    \end{tabular}
    \caption{Comparisons of the gradient schemes and extensions presented in Section~\ref{sec:gradient_descent} for a fixed simplicial complex with $N$ simplices. }
    \label{tab:summary_table}
\end{table}

\section{An overview of some applications}\label{sec:applications}

We now provide a tour of practical applications involving topological optimization based on persistent homology.
The aim is not to be exhaustive, but rather to give an overview of typical problems for which topological optimization has proved to be useful.
We separate these problems into two families: \emph{filtration learning for producing better descriptors}, in which filtrations used for PD computation are automatically learned instead of being chosen a priori (see Section~\ref{subsec:app_filtlearn}), and
\emph{topological regularization for constraining models}, in which PDs are used as features for penalties in statistical inference and machine learning models in order to reduce complexity and/or use topological priors (see Section~\ref{sec:app_reg}).

\subsection{Filtration learning}\label{subsec:app_filtlearn}

The aim of filtration learning is to optimize filtrations so that the resulting PDs (and their subsequent vectorizations if applicable) are tailored for solving a given data science task. This provides a way for designing topological descriptors that does not rely on a user-specific filtration choice.
It is often used in a supervised context, where one is given some simplicial complexes $\{K_i\}_{1\leq i\leq n}$, together with some labels $\{y_i\}_{1\leq i\leq n}$ from a label space $\cY$, that 
a model $\phi \colon \pdspace\to\cY$ aims at predicting using the information captured by PDs. Typically, each data point induces a filtration map $F_i:\theta\mapsto f_\theta\in\Filt_{K_i}$, usually
parametrized with, e.g., a neural network, and a loss $L_i$ of the form $L_i:\alpha\mapsto \ell(y_i,\phi(\alpha))$, where $\ell$ depends on the label type (e.g., mean squared error
when $y_i\in\bR^d$, or cross-entropy when $y_i$ is a categorical variable).
The final loss is then:
\begin{equation}\label{eq:topolearn}
\cL \colon \theta\mapsto\sum_{i=1}^n L_i(\Pers\circ F_i(\theta)).
\end{equation}

\begin{remark}
    It is common that the model $\phi$ depends itself on some parameters $\theta'$, i.e., it is of the form $\phi_{\theta'}(\alpha)$. 
    These parameters are also optimized at training time,
    leading to a final loss of the form $\cL(\theta,\theta')$.
    This typically happens when, e.g., one wants to complement the topological descriptors with some standard models and/or descriptors, which come with their own parameters.
    For the sake of simplicity, we leave this point aside in this section. 
\end{remark}

\subsubsection{Filtration learning for images}

{\bf Reference article:} Barbarani et al.~\cite{Barbarani2024} \\

In this work, authors propose to rely on topological optimization to propose a methodology to detect \emph{salient points} on images, or image \emph{keypoints}. 
The main idea is to learn a filtration $F_\theta$ of which the local maxima (recorded in a PD in the sub) are interpreted as keypoints. 
In this case, the filtration map $F_\theta \colon \bR^{d_1 \times d_2} \to \bR^{d_1' \times d_2'}$ is encoded as a convolutional neural network (CNN), $d_1,d_2$ (resp.~$d_1',d_2'$) being the shape of the input (resp.~output) image. 
Given an input image $x \in \bR^{d_1 \times d_2}$, $F_\theta(x)$ can be understood as a height map amenable to persistence computation (see \cref{ex:img_filt}). 
See \cref{fig:keypoint} for a schematic illustration of the pipeline described in \cite{Barbarani2024}. 
The topological loss $L_i$ is defined on each image $I_i$ as
\begin{equation}
    L_i=\sum_{j\neq i} L_{(I_i,I_j)}(\Pers\circ F_\theta(I_i), I_j),
\end{equation}
where $L_{(I_i,I_j)}(\alpha_i, I_j)$ is the sum over the PD points in $\alpha_i$ of the products between: 1. the opposite of the lifetime of the PD point (hence strongly penalizing points close to the diagonal, thus favoring peaked height maps), and 2. a similarity term measuring the differences between $(a)$ the height map values of the critical pixels in $I_i$ associated to the persistence pair of the PD point, with $(b)$ the height map values of their corresponding pixels in $I_j$ using a ground-truth correspondence $U$ between the pixels of $I_i$ and $I_j$ (hence penalizing keypoints that are not reproducible across images). 
In a nutshell, $L_i$ is small if $F_\theta(I_i)$ has its prominent local maxima on pixels that also correspond to local maxima in the other images $\{I_j\}$ under $U$. 
The benefits on using persistent homology in this work is that it removes the need of using any prior and/or hyperparameters on the keypoint densities, scales and frequencies across images, in contrast to related works in the literature. 

\begin{figure}[ht]
	\centering
	\includegraphics[width=.5\textwidth]{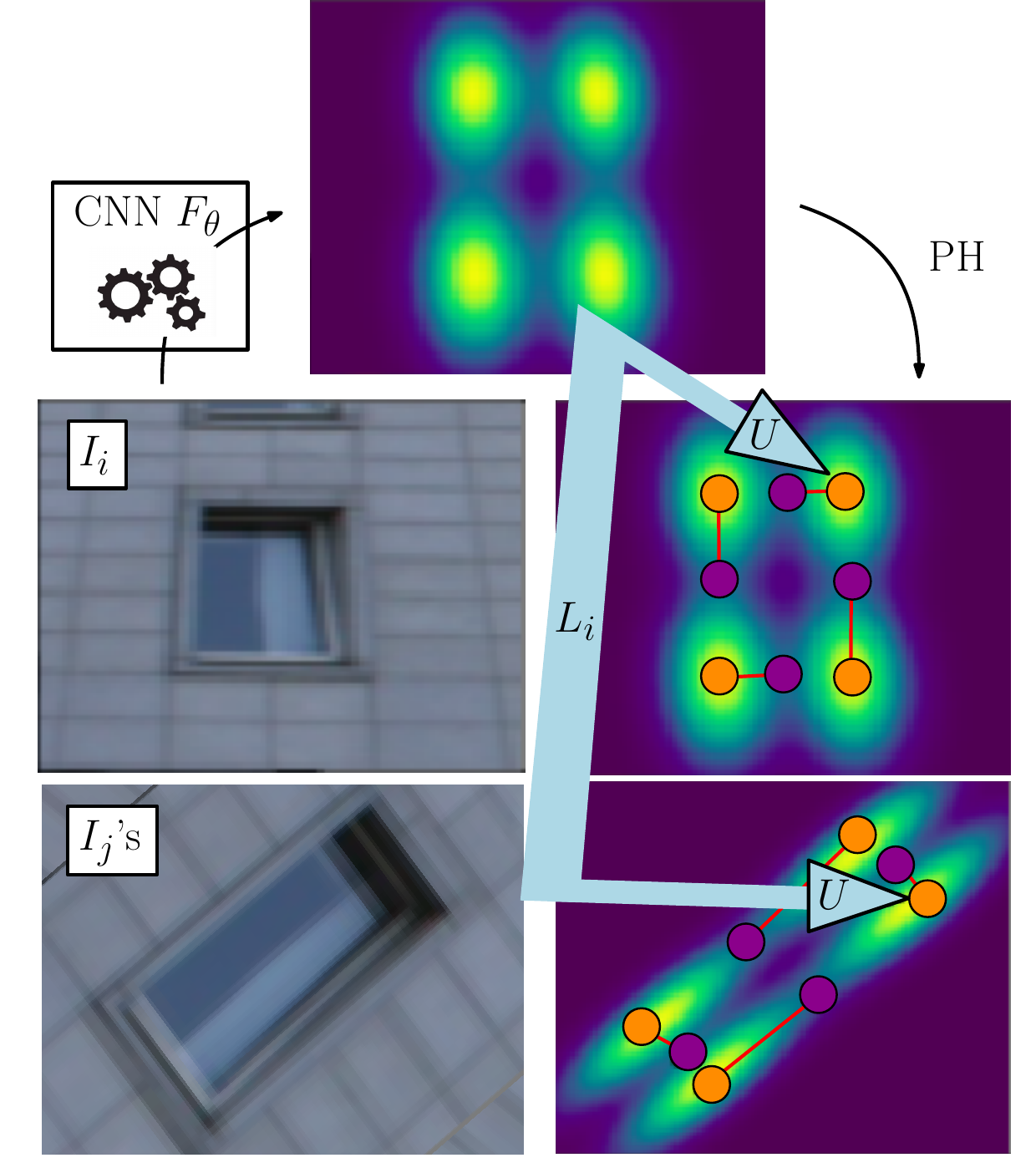}
	\caption{Filtration learning for images. Once orange/magenta pixels corresponding to prominent maxima of the height map (i.e.,~yielding points with large persistence in the PD) have been obtained for two given images, they are compared using a ground truth correspondence $U$. The goal is to optimize the CNN weights $\theta$ so that $F_\theta$ produces useful height maps, in the sense that corresponding PDs enable the identification of keypoints as local maxima of the height map. Inspired from \cite[Figure 2]{Barbarani2024}.}
    \label{fig:keypoint}
\end{figure}

\subsubsection{Filtration learning for graphs}

{\bf Reference article:} 
Horn et al.~\cite{horn2021topological} \\

Given that graphs can be interpreted as 1-dimensional simplicial complexes, it is natural to incorporate topological features in graph classification and regression. 
In the context of the work \cite{horn2021topological}, this is done in the following way. 
Graph nodes (a.k.a.~vertices) $v \in V$ come with attributes $x^{(v)} \in \bR^d$. 
A first network $\Phi \colon \bR^d \to \bR^k$ (e.g.,~a single hidden layer network in \cite{horn2021topological}) assigns to each node $x^{(v)}$ a set of $k$ values $a_1^{(v)},\dots,a_k^{(v)}$, yielding $k$ different \emph{views} of the graph $G$: each view represents the same graph $G$ but with different attributes $a_i^{(v)} \in \bR,\ i \in \{1,\dots,k\},\ v \in V$ on the nodes. 
Given these $k$ views, one can consider the filtrations $(G_{i,t})_{t \geq 0}$ for $i = 1,\dots,k$ where the set of vertices of $G_{i,t}$ is $\{v\in V\,|\, \Phi(x^{(v)})_i < t \}$, and the set of edges is $\{(v,v')\,|\, \max( \Phi(x^{(v)})_i, \Phi(x^{(v')})_i) < t \}$. 
These $k$ filtrations yield $k$ persistence diagrams (accounting for homology dimension $0$ and $1$). 
These $k$ diagrams are themselves turned into vectors (see \cref{ex:linear_reps}) and aggregated by another network $\Psi \colon \pdspace^k \to \bR^d$, eventually processed to obtained a new representation $\tilde{x}^{(v)}$ of the node $v$ which accounts for the topological structure of the initial  graph $G$ with nodes attributes $x^{(v)}$. 
Here, the maps $\Phi$ and $\Psi$ are optimized during training, involving an intermediate topological optimization step, yielding a loss of the form of Equation~\eqref{eq:topolearn}.
See~\cref{fig:togl} for a schematic illustration of the pipeline developed in  \cite{horn2021topological}. 
We refer the interested reader to~\cite{phamTopologicalDataAnalysis2025} for another example of graph filtration learning, and graph classification problems with topological descriptors.

\begin{figure}[ht]
\centering
\includegraphics[width=.99\textwidth]{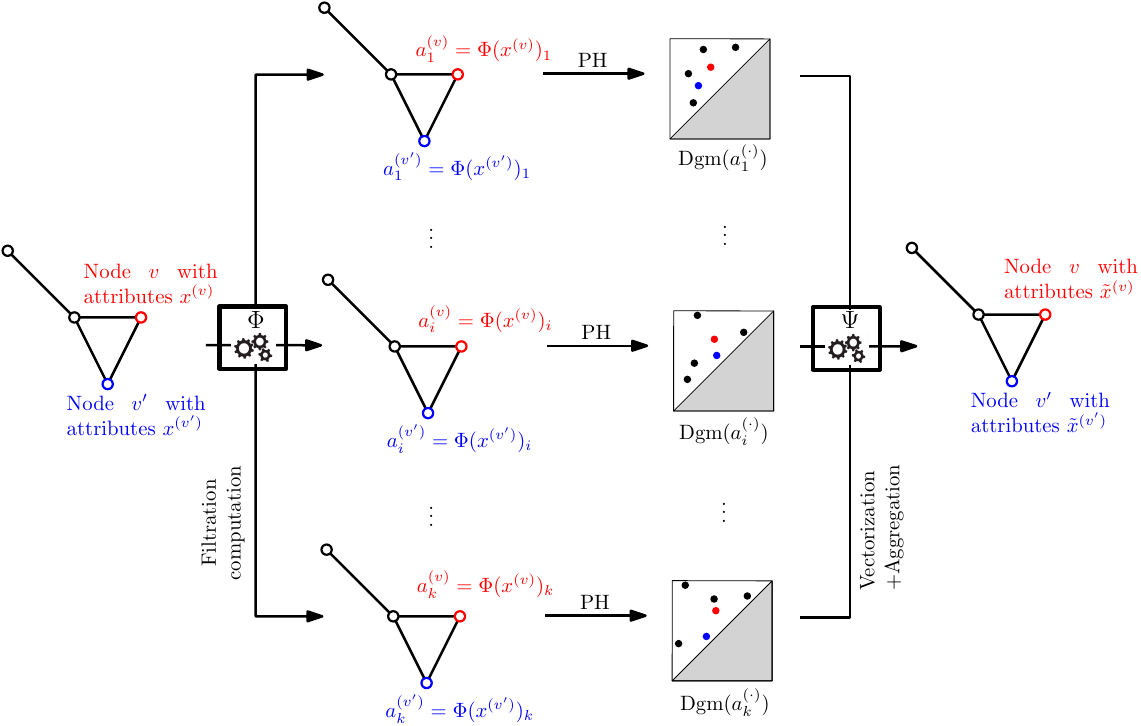}
\caption{Filtration learning for graphs. Inspired from \cite[Figure 2]{horn2021topological}.} \label{fig:togl}
\end{figure}

\subsubsection{Filtration learning for geometric complexes}

{\bf Reference article:} Nishikawa et al.~\cite{nishikawa2024adaptive} \\

Learning filtrations can also be achieved for geometric complexes induced by point clouds. Such a method has recently been proposed 
for weighted Rips filtrations (see~\cref{ex:wrfilt}).
More precisely, for a given point cloud $X$, the weight function $f=F(\theta)=f_\theta:X\to\bR$ can be parametrized
with a combination of DeepSet~\cite{zaheer2017deep} and fully-connected neural network architectures trained on the pairwise distance matrices of the point clouds. 
See~\cref{fig:topcl}.

\begin{figure}[ht]
\centering
\includegraphics[width=.8\textwidth]{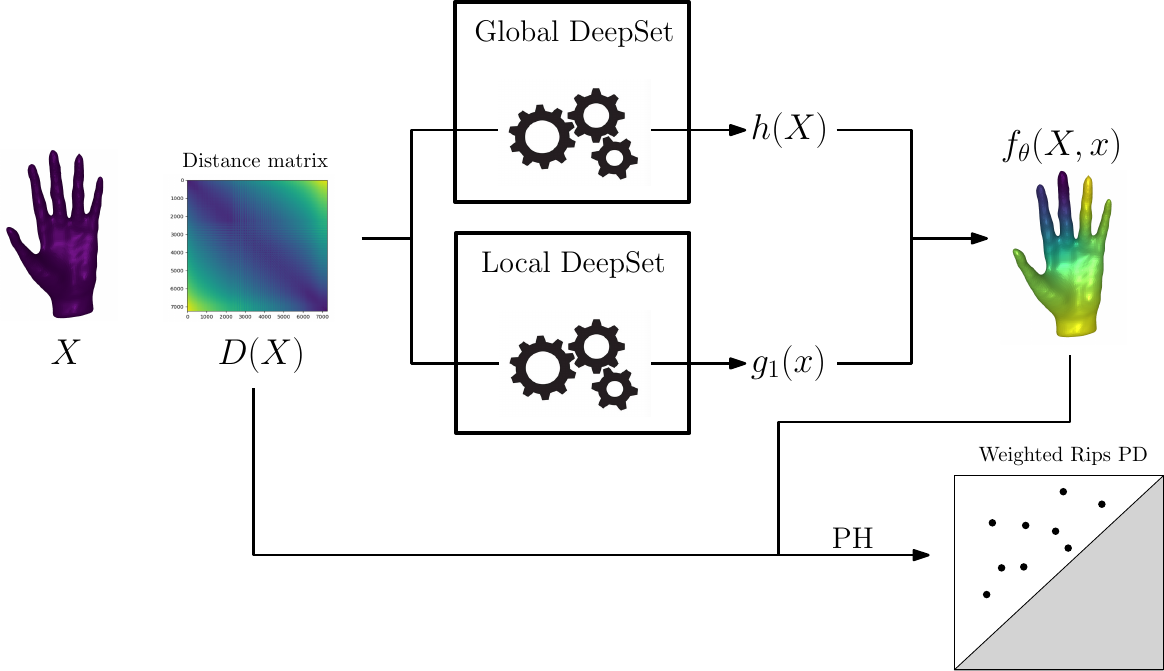}
\caption{\label{fig:topcl} Filtration learning for point clouds. }
\end{figure}

More formally, given a point cloud $X=\{x_1, \ldots, x_n\}\subset\bR^d$ and a point $x\in\bR^d$, the proposed network outputs a weight value $f_\theta(X, x)$ for the point $x$.
To make the resulting PD isometry invariant, 
it is desirable that the weight function also enjoys such invariance.
To achieve this, the network architecture computes features based on distance matrices, that can be decomposed in the following three parts:

\paragraph{Pointwise feature.}
The feature of point $x\in\bR^d$ can be obtained through aggregation of the distances $\|x-x_1\|, \ldots, \|x-x_n\|$ to the points in $X$ as follows:
\begin{equation}
    g_1(x)
    \coloneqq
    \phi^{(2)}(
        \mathbf{op} ( \{\phi^{(1)}(\|x-x_j\|)\}_{j=1}^n )
    ),
\end{equation}
where $\phi^{(1)}$ and $\phi^{(2)}$ are fully-connected neural networks.

\paragraph{Global feature.}
The global feature of the entire point cloud $X$ is inspired by the DeepSet architecture~\cite{zaheer2017deep} and can be obtained as follows. 
First, pointwise feature vectors are computed with the same architecture as $g_1$, i.e., 
\begin{equation}
    g_2(x_i)
    \coloneqq
    \phi^{(4)}(
        \mathbf{op} ( \{\phi^{(3)}(d(x_i, x_j))\}_{j=1}^n )
    ),
\end{equation}
where $\phi^{(3)}$ and $\phi^{(4)}$ are fully-connected neural networks.
Then, a feature vector $h(X)$ for the entire point cloud $X$ is obtained by applying a DeepSet architecture:
\begin{equation}
    h(X) \coloneqq \phi^{(5)}(
        \mathbf{op}(\{g_2(x_i)\}_{i=1}^n)
    ), 
    \label{eq:distmatnet}
\end{equation}
where $\phi^{(5)}$is a fully-connected neural network.

\paragraph{Combining features.}
The local and global features can be gathered in the concatenated vector $[h(X), g_1(x)]$.
The weight value of point $x$ can finally be obtained through a fully-connected neural network $\phi^{(6)}$:
\begin{equation}
    f_\theta(X,x) \coloneqq \phi^{(6)}([h(X),g_1(x)]^\top),
\end{equation}
where $\theta \coloneqq(\theta_k)_{k=1}^6$ is the set of parameters that appear in network $\phi^{(k)}$~$(k=1, \dots, 6)$.
Once the filtration $f_\theta$ is computed, the corresponding PDs can be obtained and vectorized in a differentiable way with PersLay~\cite{carriere2020perslay}, and fed to any standard deep learning model $\phi$ in order to classify point clouds,
by minimizing the cross-entropy loss $\ell$ in Equation~\eqref{eq:topolearn}.
See the corresponding article~\cite{nishikawa2024adaptive} for applications on protein and 3D mesh datasets.

\subsection{Topological regularization}\label{sec:app_reg}

Topology has also been proved to be a good method for \emph{constraining} models, either for $(i)$ limiting their \emph{complexities}, or for $(ii)$ imposing \emph{topological priors}. 

\begin{enumerate}
	
	\item[$(i)$] The more complicated a model is,
	the more likely it is that PDs computed out of it contain complex information. Hence, there are 
	various contexts in statistical inference and machine learning in which models are improved by penalizing those 
	whose PDs are too rich, which is called {\em topological regularization}.
	
	\item[$(ii)$] In the context of geometric complexes and point clouds, topological regularization can often be interpreted as
	\emph{topological priors}, i.e.,
	imposing that datasets contain specific geometric features, which is particularly relevant in \emph{dimensionality reduction} (see Section~\ref{subsec:app_dr}), or in generative models (e.g.,~generate images that contains topological priors).
	
\end{enumerate}

Given a parametrized model $\phi(\theta)$, this is usually achieved by optimizing losses of the form:
\begin{equation}\label{eq:toporeg_loss}
	\cL:\theta\mapsto L_{\rm data}(\phi(\theta)) + \lambda_{\rm topo} L_{\rm topo} (\Pers\circ F(\phi(\theta))),
\end{equation}
where $L_{\rm data}$ is a standard machine learning loss. 
The coefficient $\lambda_{\rm topo} > 0$ controls the strength of the topological penalty or prior, and is often picked with cross-validation.

\subsubsection{Topological penalization of model complexity}

{\bf Reference article:} Chen et al.~\cite{chen2019topological} \\

In this work, the authors consider a binary classification problem, i.e.,~find a model $\phi_\theta \colon \bR^d \to [-1,1]$, where $\mathrm{sign}(\phi_\theta(x))$ eventually states the class assigned to a point $x \in \bR^d$. 
Given a finite (training) sample $x_1,\dots,x_n$ and corresponding labels $\{y_1,\dots,y_n\} \subset \{-1, 1\}$, one seek to optimize the model parameters $\theta$ so that $\mathrm{sign}(\phi_\theta(x_i)) \simeq y_i,\ i=1,\dots,n$. 
For expressive classes of models (e.g.,~large neural networks), many values of $\theta$ may yield near-to-perfect accuracy on the training set; but in order to mitigate overfitting, one may balance between the model training accuracy and the model regularity. 
Of particular interest is the decision boundary $\{x\in \bR^d\,|\,\phi_\theta(x) = 0 \}$: intuitively, models with simpler decision boundaries (that still achieve a decent training accuracy) should be preferred. 

In \cite{chen2019topological}, the authors address the regularity of the decision boundary in terms of \emph{topological complexity}: ideally, the decision boundary should have as few topological features (connected components, number of holes...) as possible. 
Precisely, given $\phi_\theta$, one can use the model predictions as a filtration
for the whole ($k$-NN graph of the) data space $X$, i.e., one can use $F \colon \phi_\theta \mapsto \{\phi_\theta (x)\}_{x\in X}$ in~\cref{eq:toporeg_loss}.
As the classification boundary of the model is exactly $\phi_\theta^{-1}(\{0\})$, one can thus set $L_{\rm topo}={\rm PersTot}$ (see \cref{ex:perstot}), or the following loss:
\begin{equation}\label{eq:loss}
	L_{\rm topo} \colon \alpha\mapsto \sum_{(b,d)\in \alpha} \min\{|b|,|d|\}, 
\end{equation}
and then apply these losses to $\Dgm(F_\theta)$, where $F_\theta \coloneqq F(\phi_\theta)$, restricted to the PD points whose birth times are negative and death times are positive. 
As for the term $L_{\rm data}$, it is usually defined as any of the standard classification scores from the data science literature.
See~\cref{fig:clreg}.
This work showcases in particular an important application of topological regularization on graphs. 
It relies on the capacity of graphs (such as, e.g., $k$-nearest neighbor ($k$-NN) graphs) to model the data space. 
Namely, when given a function $f$ defined on a finite set of data points, one can extend the domain of the function to a $k$-NN graph built on top of the data points in order to be able to compute $\Dgm(f)$ (which would not be well-defined otherwise). 
This in turn allows to use $\Dgm(f)$ to penalize and reduce the complexity of the function $f$. 
See also~\cite{Hacquard2022} for a similar technique applied to the regression of functions defined on manifolds (and approximated by $k$-NN graphs computed on finite samples).

\begin{figure}[ht]
	\centering
	\includegraphics[width=.9\textwidth]{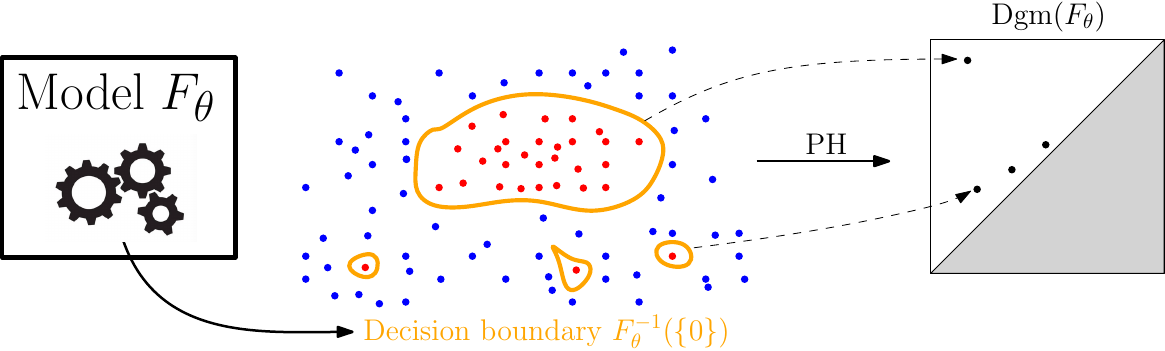}
	\caption{\label{fig:clreg} A schematic illustration of the use of topological regularization for binary classifiers defined. On the left, a dataset made of two classes (red and blue points, better in color) and the classification boundary of the model (orange). On the right, the PD of that boundary. Intuitively, the small circles used to catch isolated red points are likely to yield overfitting. They are accounted as points close to the diagonal in the persistent diagram. Penalizing the apparition of topological features in the classification boundary that are only used to catch few points is likely to mitigate overfitting. Inspired from \cite[Figure 2]{chen2019topological}.}
\end{figure}

\subsubsection{Penalization to favor topological priors}\label{subsec:app_dr}

\paragraph{Topology-constrained image generation.} {\bf Reference article:} Wang et al.~\cite{wangtopogan} \\

Constraining models with topological penalties has found some success on generative models for images.
Indeed, considering, e.g., generative adversarial models (GANs), one can complement the standard GAN losses 
(the discriminator and generator losses) with a topological penalty that forbids the generator to produce images with
incorrect topologies. 
Incorrect topologies are criteria that are difficult to measure and quantify with standard image descriptors, but that can be easily obtained
on binary images
using, e.g., the distance to the closest black pixel as filtration values (referred to as the distance transform (DT) in the article). 
See \cref{fig:topogan}.
More specifically, the topological GAN loss can be written as:
\begin{equation}
L_{\rm topo}(\Pers\circ F(\phi(\theta))) = W(\{\alpha_i\}_i, \{\alpha_j(\phi(\theta))\}_j),
\end{equation}
where $W$ is an optimal transport distance between two sets of PDs (using the $1$st diagram distance $\FG_1$ as ground metric, see \cref{def:dist_wasserstein}), and $i$ (resp. $j$) ranges over the index set of the real (resp. fake, or synthetic) images. 
As the synthetic PDs $\{\alpha_j\}_j$ are computed
on the synthetic images, which are themselves produced by the generator $\phi$, they depend on the generator parameters $\theta$.
Overall, this helps to produce images whose geometric and topological features match better those of the images in the training set. 
Works sharing a similar goal have been proposed in mechanical sciences \cite{behzadi2022gantl,sugai2024data,hu2024if,kii2025data}, where persistence-based losses are included to favor the apparition of certain topological features in material design, or topological diversity in the population of material designed. 

\begin{figure}[ht]
	\centering
	\includegraphics[width=.7\textwidth]{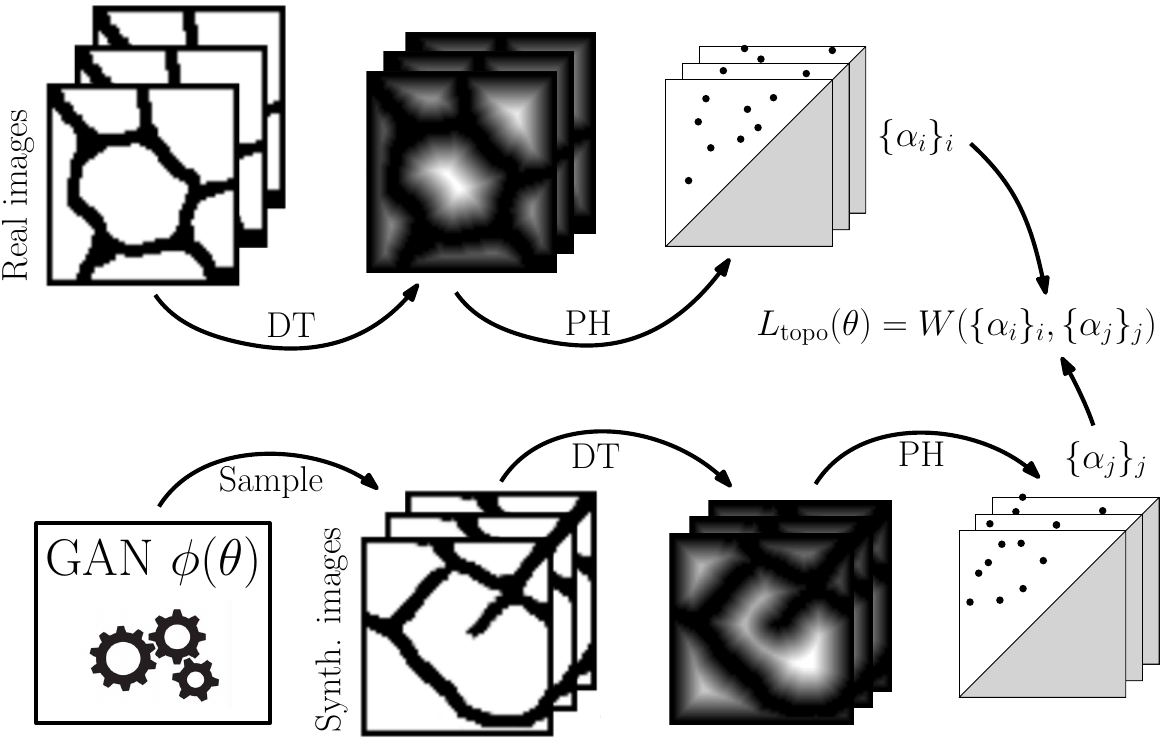}
	\caption{Topological regularization for images. Inspired from Wang et al.~\cite{wangtopogan}.}\label{fig:topogan}
\end{figure}

\paragraph{Dimensionality reduction.} {\bf Reference article:} Moor et al.~\cite{moor2019topological} \\

The main goal of dimensionality reduction is to visualize
datasets in high dimensions. Given a dataset $X\in\bR^{n\times d}$, representing $n$ points
with $d \gg 1$ dimensions, the goal is thus to find a \emph{reduced dataset} $\tilde X\in\bR^{n\times d'}$, with $d' \ll d$, typically $d'=2$ or $3$.
Additionally, this reduced dataset should be faithful, that is, close to $X$ with respect to a given notion of proximity between $X$ and $\tilde X$.
For instance, multidimensional scaling (MDS) aims at preserving the pairwise distances of $X$, i.e., the pairwise distances between 
embedding points in $\tilde X$ should be as close as possible to the corresponding distances in the original dataset $X$.
In addition to such losses, it is also natural to ask that the topology, understood through the lens of persistent homology, of $\tilde X$ and $X$ should 
coincide to some extent. 
This can be achieved with the loss provided in \cref{eq:toporeg_loss}, applied to point clouds:
\begin{equation}\label{eq:topodimred}
	\cL(\theta) \coloneqq L_{\mathrm{data}}(\tilde X(\theta)) + \lambda_{\mathrm{topo}} L_{\rm topo}(\Pers\circ F(\tilde X(\theta)),
\end{equation}
where $F$ is a filtration map on the $\tilde X(\theta)$, $L_{\rm topo} \colon \cD \to \bR$ is a loss that controls the PD of the reduced dataset, and
$L_{\rm data}$ is a standard geometric loss, such as, e.g., the reconstruction loss of autoencoder neural networks.
For instance, in \cite{moor2019topological}, authors consider an encoder $E_\theta \colon \bR^{d} \to \bR^{d'}$ producing low-dimensional embeddings and set $\tilde X(\theta) \coloneqq E_\theta(X)$, a decoder $D_{\theta'} \colon \bR^{d'} \to \bR^d$, and thus $L_{\rm data} (\tilde{X}(\theta)) \coloneqq |D_{\theta'} \circ E_\theta (X) - X|^2$. 
On the other hand, in order to favor encoders that preserves the topology (namely, the Vietoris--Rips PD) of the input point cloud $X$, one may take~$L_{\rm topo} = \FG_q(\cdot,\DgmRips(X))$ (see \cref{ex:dist_to_dgm})\footnote{Strictly speaking, the topological loss used in \cite{moor2019topological} is slightly different and involves minimizing the differences between the distances associated to the critical pairs in both PDs---we leave this subtlety aside in this survey.}.
See~\cref{fig:toae}, as well as~\cite{doraiswamyTopoMap0DimensionalHomology2021, vandaele2021topologically, Wagner2021, Carriere2024} for other instances of topological losses for dimensionality reduction. 
Note also that in~\cite{Carriere2024}, the use of the diffeomorphic gradient (see Section~\ref{subsubsec:diffeo}) allows to decouple \cref{eq:topodimred}: one can compute the topological regularization on the output of a black-box, pre-trained model $\phi_\theta$ \emph{after it has been trained}, and use the corresponding sequence $G$ of vector fields to produce a \emph{new} model $\tilde \phi_\theta \coloneqq G \circ \phi_\theta$.

\begin{figure}[ht]
	\centering
	\includegraphics[width=\textwidth]{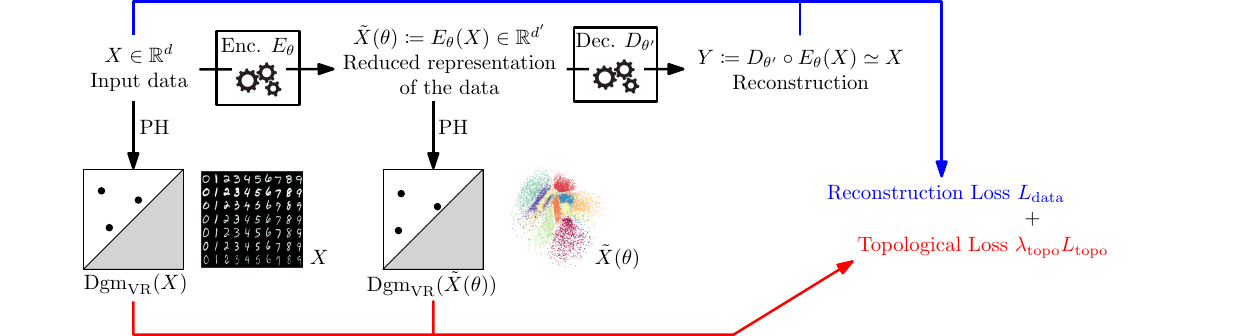}
	\caption{\label{fig:toae} Topological regularization for point clouds. Inspired from \cite[Figure 2]{moor2019topological}.}
\end{figure}

\paragraph{Topology-aware segmentation.} {\bf Reference article:} Liu et al.~\cite{liu2022toposeg}

Consider a dataset consisting of a collection of point clouds in $\mathbb{R}^d$, each containing $N$ labeled points. 
Let $\scC\coloneqq\{1, \ldots, C\}$ be the set of labels and $\scX$ be the space of point clouds. 
The goal of segmentation tasks is to learn a model $\phi_\theta\colon\scX\times\bR^d\to\scC$ to estimate the label $\phi_\theta(X, x_k)$ of the point $x_k$ in the point cloud $X$.
As usual, $\theta$ represents the trainable parameters. 
Typically, one can use neural networks such as PointNet++~\cite{qi2017pointnet++} and DGCNN~\cite{phan2018dgcnn} for $\phi_\theta$.

Let $\{\{(x^i_{k}, y^i_{k})\}_{k=1}^N\}_{i=1}^M$ be a training dataset, where $x^i_{k}$ is the $k$-th point in the $i$-th point cloud, and $y^i_{k}\in\scC$ denotes its label. 
Treating this problem as a classification task, the model $\phi_\theta$ can be optimized using a standard cross-entropy loss, denoted by $L_{\mathrm{CE}}$ hereafter. 
However, while neural networks showcase remarkable ability to learn local feature of point clouds, 
practitioners observe that they often fail to capture global features, especially topological features.
To deal with this issue, the authors of \cite{liu2022toposeg} propose to incorporate a topological term in their loss function to make the segmented point cloud have close topological structures to the ground truth.
Namely, given $t>0$, one constructs two filtrations $F^{(c)}_\theta(X_i), G^{(c)}(X_i)\colon R[X]_t\to\bR$ on the Vietoris-Rips complex $R[X]_t$ using the output of model $\phi_\theta$ and the ground truth in the following way:
\begin{equation}
	F^{(c)}_\theta(X_i)(\sigma) = 1 - \min_{x_{ij}\in\sigma}\phi^{(c)}_\theta(X_i, x_{ij}),
	\quad
	G^{(c)}(X_i)(\sigma) = 1 - \min_{x_{ij}\in\sigma}g^{(c)}_{ij}.
\end{equation}
They define their topological loss $L_{\mathrm{topo}}$ as (a variant of) the FG distance between the persistence diagrams of $F_\theta$ and $G$:
\begin{equation}
	L_{\mathrm{topo}}(f_\theta, g)\coloneqq\sum_{k=1}^M\sum_{i=1}^N\sum_{c=1}^C \FG_2(D(F^{(c)}_\theta(X_i)), D(G^{(c)}(X_i))),
\end{equation}

Minimizing the following combined loss of the form of Equation~\eqref{eq:toporeg_loss}
\begin{equation}
	L_{\mathrm{CE}}(\phi_\theta, g)+\lambda_{\rm topo} L_{\mathrm{topo}}(\phi_\theta, g)
\end{equation}
favors a segmentation with topological structures close to the ground truth. 
Mixing these two terms eventually produces a segmentation generally closer to that given by the ground truth than using $L_\mathrm{CE}$ alone \cite[Table 1]{liu2022toposeg}. 
A similar approach has been explored in different works, see for instance \cite{demir2023topology,clough2020topological,brito2025persistent}.

\begin{figure}
    \centering
    \includegraphics[width=0.9\linewidth]{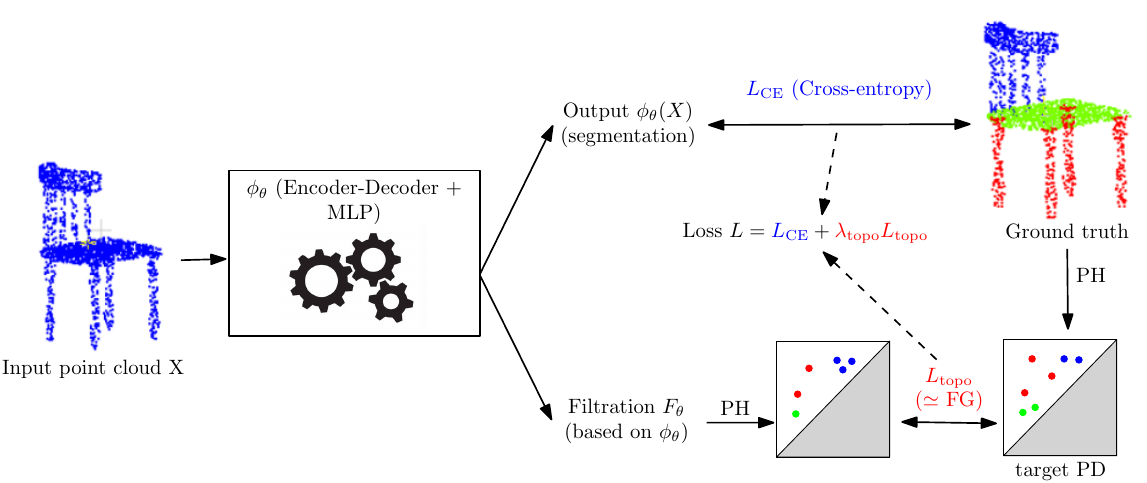}
    \caption{Schematic overview of topological segmentation as proposed by \cite{liu2022toposeg}. The model $\phi_\theta$ is trained in order to both produce a segmentation close to that given by the ground truth, both in terms of cross-entropy and in terms of similarity between the computed persistence diagrams. Inspired from \cite[Figure 2]{liu2022toposeg}}
    \label{fig:placeholder}
\end{figure}

\section{Numerical illustrations}\label{sec:expes}

In this section, we showcase the empirical behaviors of the different topological gradients and their extensions presented in~\cref{sec:gradient_descent}. 
We start with a simple proof-of-concept experiment, and then consider a \emph{topology-preserving} dimensionality reduction problem,  
a classical application of topology-based optimization with gradient descent (see also~\cref{subsec:app_dr}). 
We stress that the goal of this section is mostly pedagogical: it does not aim at establishing state-of-the-art results but to provide an all-in-one-place illustration of the different persistence-based topological optimization variants discussed in this survey. 
Our code\footnote{Note that, in its current implementation, our code can only be applied to point clouds.} is available at
\url{https://github.com/git-westriver/benchmark_ph_optimization}. 

\subsection{Illustration of the different gradient descent methods}

\paragraph{A first proof-of-concept.} In this experiment, we optimize the coordinates of a point cloud $X$ with respect to a loss that is based on the Vietoris--Rips persistence diagram $\DgmRips(X)$. 
The point cloud $X$ is initialized in $\bR^2$ as $100$ points close to a unit circle, along with one outlier near the origin.
We write $X=(x_0,\dots,x_{100}) \in \bR^{2 \times 100}$ and optimize $X$ with the loss function 
\begin{equation}
    \scL(X) \coloneqq - \FG_2(\DgmRips^{(1)}(X), \varnothing) + \mathrm{Reg}(X),
\end{equation}
where $\DgmRips^{(1)}(X)$ is the subset of $\DgmRips(X)$ corresponding to homology dimension 1, and $\mathrm{Reg}(X) = \sum_{(x,y) \in X} (|x| - 2)_+^2 + (|y|-2)_+^2$ is a regularization term that penalizes points escaping the compact set $[-2,2]^2$. 
Intuitively, minimizing this loss amounts to maximizing the total $1$-persistence in $X$, i.e., amounts to making the initial loop (that is notably perturbed by the outlier) as salient as possible. 
We use the different optimization methods described in \cref{sec:gradient_descent}. 
\cref{fig:poc} shows the snapshots with the different gradient displayed, along with the loss evolution of gradient steps for each optimization method; \cref{tab:execution_time_poc} displays the execution times over 20 steps.

For each method, we have two hyper-parameters: 
$(a)$ the learning rate $\eta$ used to minimize the loss, 
and chosen from $\eta\in\{0.064, \allowbreak 0.128, \allowbreak 0.256\}$ , and $(b)$
the decay rate $\gamma$ used to decrease $\eta$ with $\eta\cdot \gamma^{t-1}$ at the $t$-th step, and chosen from $\gamma\in\{1.0, 0.9, 0.8, 0.7\}$. 
For each method, we selected the values of $\eta,\gamma$ that eventually yield the smallest loss value after $20$ gradient steps.

\begin{figure}[htbp]
    \centering
    \includegraphics[width=0.95\linewidth]{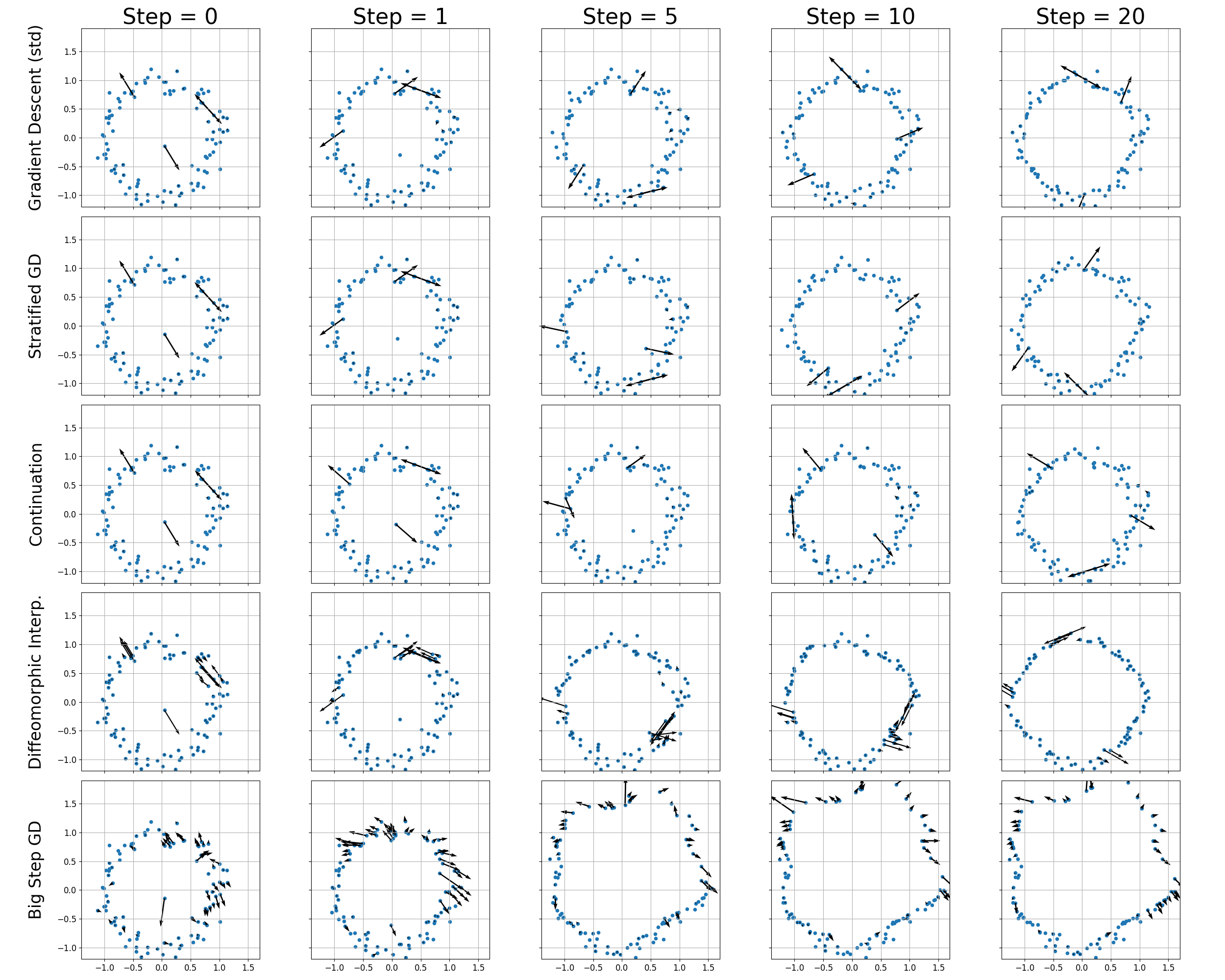}
    \captionof{figure}{Snapshot of $X$ over timesteps $0,1,5,10,20$ for the different methods.}
    \label{fig:poc}
    
    \vspace{.5cm}
    \begin{minipage}[t]{0.49\textwidth}
        \centering
        \includegraphics[width=\linewidth]{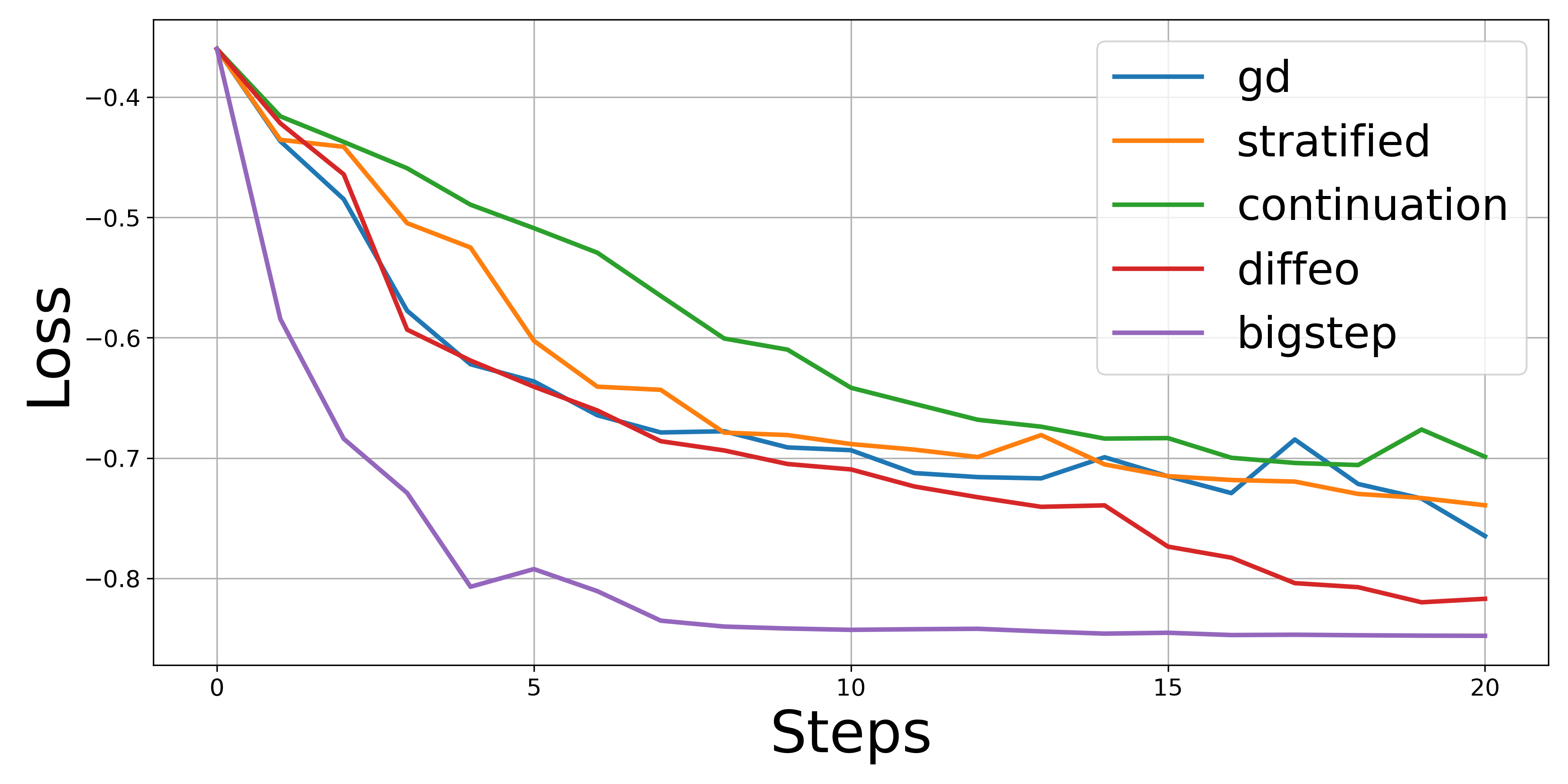}
        \captionof{figure}{Evolutions of the losses over steps.}
        \label{fig:poc_evolution}
    \end{minipage}
    \hfill
    \begin{minipage}[t]{0.49\textwidth}
    \vspace{-3cm}
    \centering
    
    \captionof{table}{Running time over 20 steps.}
    \label{tab:execution_time_poc}
    \smallskip
    
    \begin{tabular}{|c|c|}
    \hline
    Method & Time (s)\\
    \hline
    Vanilla            
        & $0.55 \pm 0.025$ \\
    Stratified
        & $4.55 \pm 0.326$ \\
    Big-step
		& $62.1 \pm 0.064$ \\
    Continuation                
        & $0.60 \pm 0.006$ \\
    Diffeomorphic                     
        & $0.64 \pm 0.004$ \\
    \hline
    \end{tabular}
    \end{minipage}
\end{figure}

The main observation that one can make from \cref{fig:poc} is that, as expected, standard gradients (first row) of $\scL$ are typically sparse, moving only few points at each step. 
A similar comment holds for the stratified gradient descent approach and the continuation one. 
In contrast, the diffeomorphic (fourth row) and big-step (fifth row) gradients provide a gradient-like update that is non-zero on much more points of $X$, yielding overall more efficient updates.

Overall, the big-step gradient yields, by a significant margin, the most efficient updates in terms of loss decrease, reaching a near-global configuration in less than 10 iterations. 
This comes at a computational price: big-step gradients are significantly longer to compute than vanilla ones and all other alternatives showcased here\footnote{Note that stratified gradients are also longer to compute due to the need of exploring nearby strata, while barely presenting other benefits in this simple setting. Recall that the purpose of this approach is mostly to provide convergence guarantees, which are not necessary in this experiment.}.

\paragraph{Scaling topological optimization via subsampling.} All the methods considered in this proof-of-concept experiment require to repeatedly compute the Vietoris--Rips persistence diagram of $X$ (in homology dimension $1$) in order to get the corresponding gradient. 
This is known to be prohibitive for large point clouds (with around a few thousands of points) and invite the practitioner to rely on subsampling methods. 
Point clouds that are close in the Gromov--Hausdorff distance yield close VR persistence diagrams thanks to \cref{th:stability}. However, this is not true when it comes to computing gradients: a critical pair in a subsample has little probability to also be critical in the original point cloud. 
Furthermore, sparsity of standard gradients is an even bigger issue when the input object is large, as only a tiny fraction of points will be updated at each step. 

We investigate qualitatively two approaches that help mitigating this phenomenon: the diffeomorphic and distributed gradients (see \cref{subsec:sample}). 
We recall that diffeomorphic interpolation takes a (standard) gradient, possibly computed on a subsample $X'$ of a large point cloud $X$, provides a vector field defined on the whole space (hence in particular on the support of $X$, not only $X'$). 
The distributed gradient consists, substantially, in averaging the gradients of several subsamples of $X$, also yielding denser gradients.

For illustrative purpose, we first consider a point cloud $X \subset \bR^2$ with $n = 2,000$ points, close to a unit circle, and consider the loss $\scL(X) = \FG_2^2(\DgmRips^{(1)}(X),\varnothing)$, i.e.,~the goal is to minimize the persistence of the underlying loop, by both collapsing it (i.e., by reducing the death time) and tearing it (i.e., by increasing the birth time). 
We fix the subsampling size at $s = 50$, and showcase on \cref{fig:poc_subsamplig} the vanilla gradient (computed from a given subsample of size $s$), the diffeomorphic gradient and the distributed gradient averaged over 100 repetitions. Both diffeomorphic and distributed gradients achieve their goal of providing an alternative to the vanilla gradient that is non-zero on much more points in $X$. The former acts at a local scale, while the latter is more global, a possibly appealing property which comes at the price of requiring to compute more (small) persistence diagrams to get a single gradient-like object. 

\begin{figure}
    \centering
    \includegraphics[width=\linewidth]{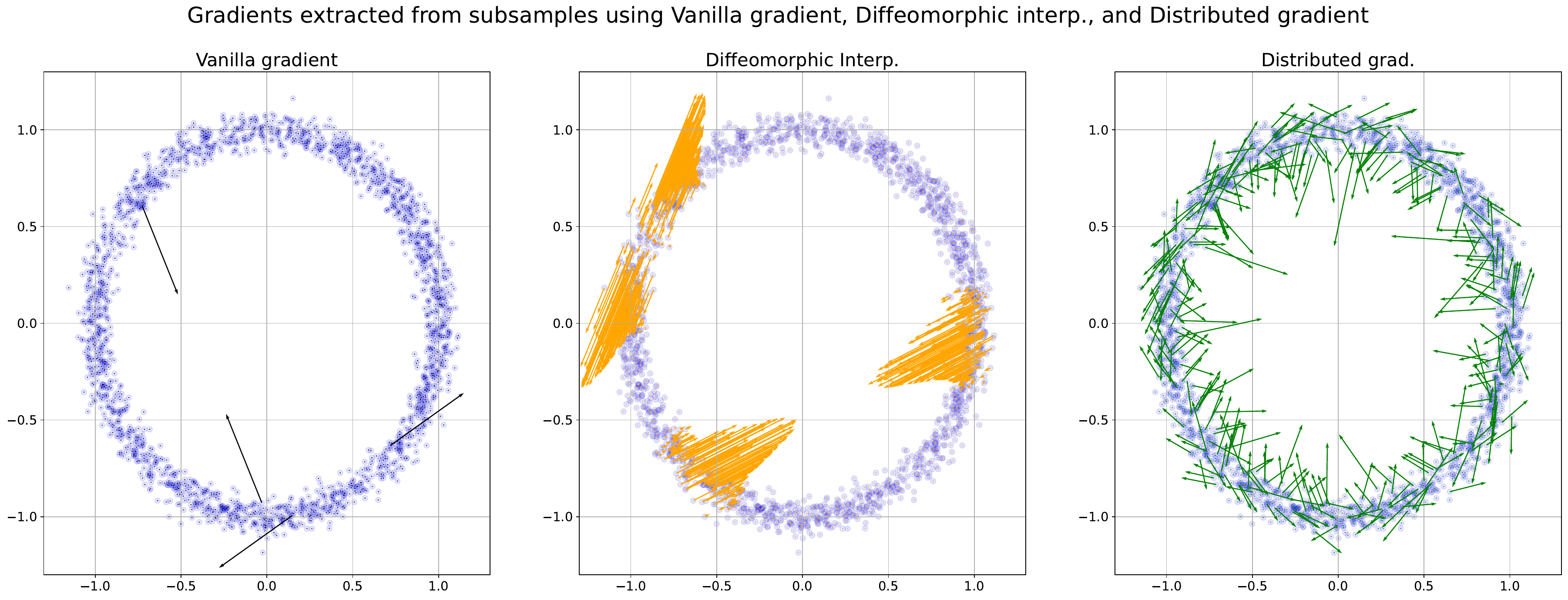}
    \caption{(Left) Vanilla gradient computed on a subsample of the original point cloud. Two pairs of points were critical in that subsample (one corresponding to a reduction of the death time, one to an increase of the birth time), yielding only four non-zero components in the gradient. (Middle) The diffeomorphic gradient extends the vanilla one in a smooth way on the input point cloud $X$, and thus follows a similar pattern while moving substantially more points than the vanilla gradient. (Right) The distributed gradient also exhibits much more non-zero components than the vanilla one.}
    \label{fig:poc_subsamplig}
\end{figure}

In order to compare these different methods, we reproduce an experiment proposed in \cite{Carriere2024}. We let $X$ be the \texttt{Stanford bunny} \cite{turk1994zippered}, a 3D point cloud made of $n = 35,947$ points, and consider the loss $\scL(X) = - \FG_2^2(\DgmRips^{(2)}(X), \varnothing) + \mathrm{Reg}(X)$, where (as in the previous experiment) $\mathrm{Reg}$ is a confinement term that penalizes points that would go out of the compact set $[-1, 1]^3$. 
This loss aims at increasing the persistence of the bunny's cavity (in homology dimension 2). 
The size of $X$ prevents from a direct computation of $\DgmRips^{(2)}(X)$ and gradients of the loss $\scL$ must be replaced by estimates obtained from subsamples. 
We therefore fix a subsampling size of $s=100$, and compare the vanilla gradient descent with the diffeomorphic gradient descent (with $\sigma = 0.05$, the value used in \cite{Carriere2024}) and the distributed gradient descent (summing $10$ vanilla gradients)\footnote{These values make the running time per step of diffeomorphic and distributed gradients similar.}. 
Eventually, we also consider the descent scheme using the diffeomorphic interpolation of a distributed gradient (also with $10$ repetitions) instead of a vanilla gradient. 
Results are displayed in \cref{fig:poc_subsampling2}. 
As in \cite{Carriere2024}, the vanilla gradient descent---which updates the position of only very few points (around 4 among the $\sim 36$k points of $X$) at each iteration---does not produce any perceptible modifications after 200 epochs. A similar comment holds for the distributed gradient approach: as we average only 10 vanilla gradients at each iteration, one may expect only $\sim 40$ points to be moved at each iteration, which remains too low to produce noticeable changes. As in \cite{Carriere2024}, the diffeomorphic interpolation approach produces satisfactory results on this type of task. 
Combining the diffeomorphic and distributed gradients yields the best lost decrease \emph{per step}, suggesting that taking combinations of these two methods can be efficient in some situations. 
Note that in terms of raw running times, distributed gradients (with 10 repetitions) require to compute 10 persistence diagrams and 10 corresponding vanilla gradients at each update step (hence a running time per step about 10 times longer than using diffeomorphic gradients alone), which overall makes diffeomorphic gradients faster to reach a near global optimum on this task. 
Nonetheless, it suggest that distributed and diffeomorphic gradients can interplay positively.

\begin{remark}
    Similarly, the interaction between big-step and diffeomorphic gradients was initiated in \cite{Carriere2024}, but no clear benefit was identified. More generally, investigating the possible interactions (suggested by \cref{fig:exgrad}) between topological gradients further is an interesting future research track. 
    Moreover, combining variants of the standard gradient descent algorithm (e.g.,~using momentum\footnote{Indeed, momentum (which accumulates gradients over iterations) provides a natural way to tackle vanilla gradient sparsity.} as explored in \cite[\S 4.1]{nigmetov2024topological}) with the topological gradients presented in this article is also another important research avenue.
\end{remark}

\begin{figure}
    \centering
    \includegraphics[height=3.5cm]{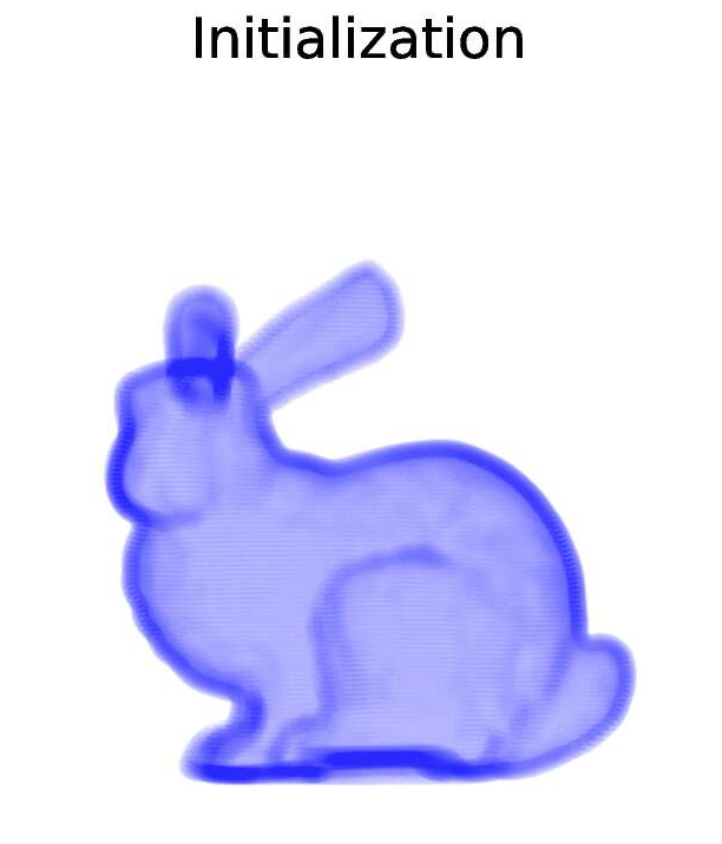}
    \includegraphics[height=3.5cm]{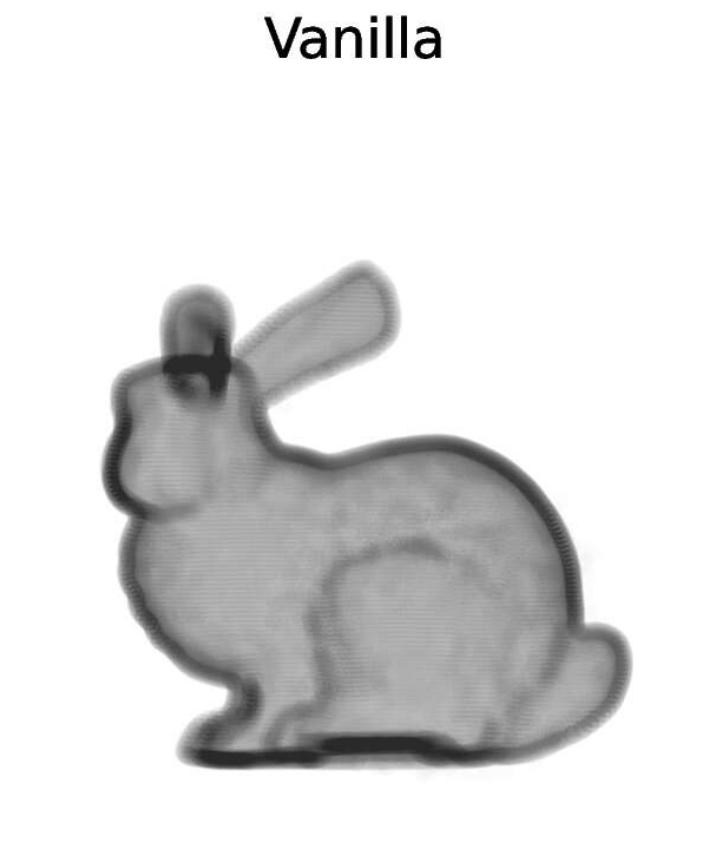}
    \includegraphics[height=3.5cm]{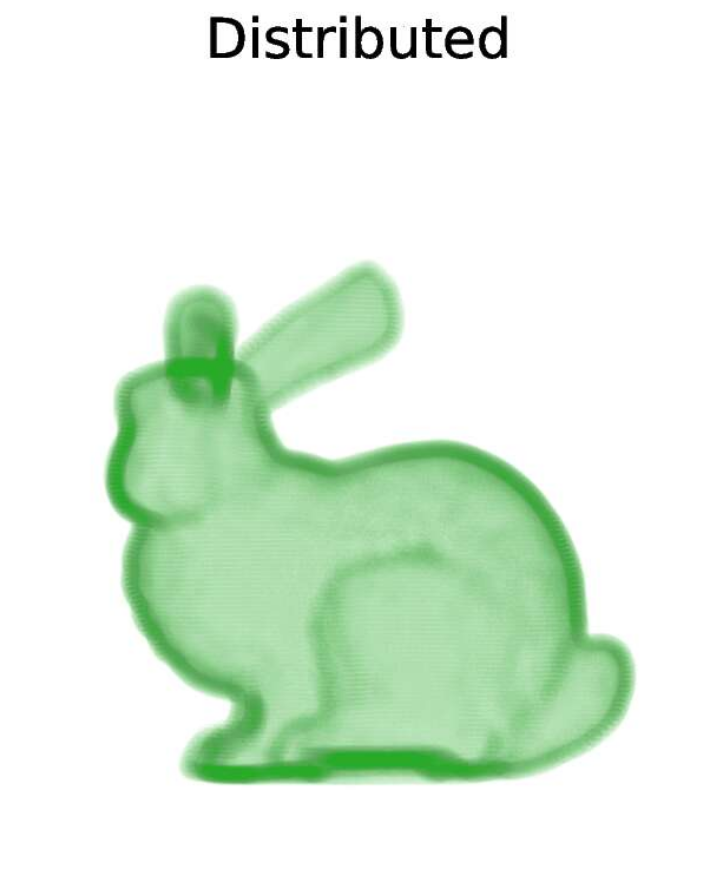}
    \includegraphics[height=3.5cm]{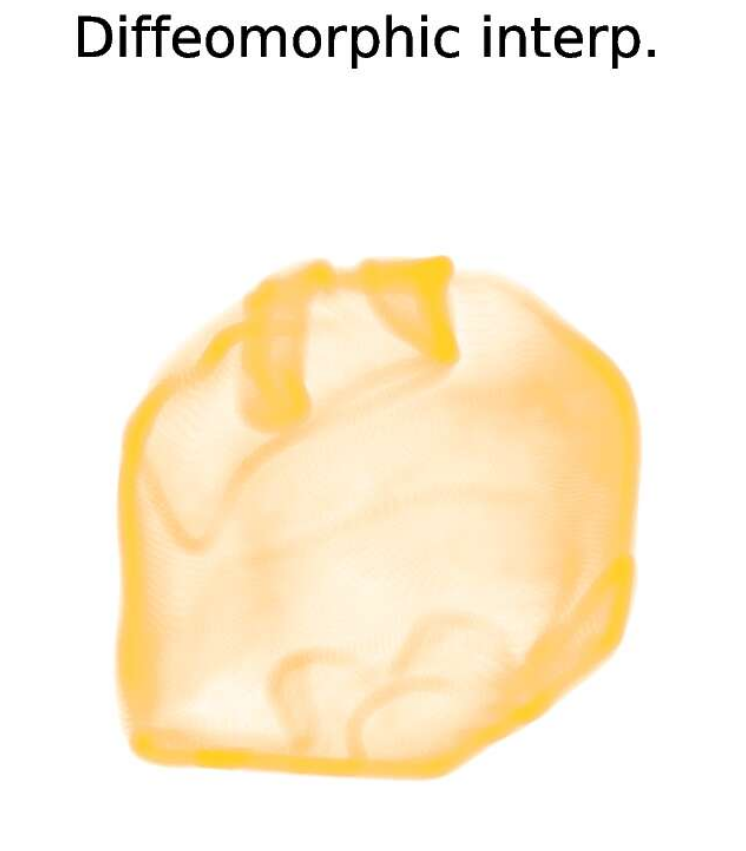}
    \includegraphics[height=3.5cm]{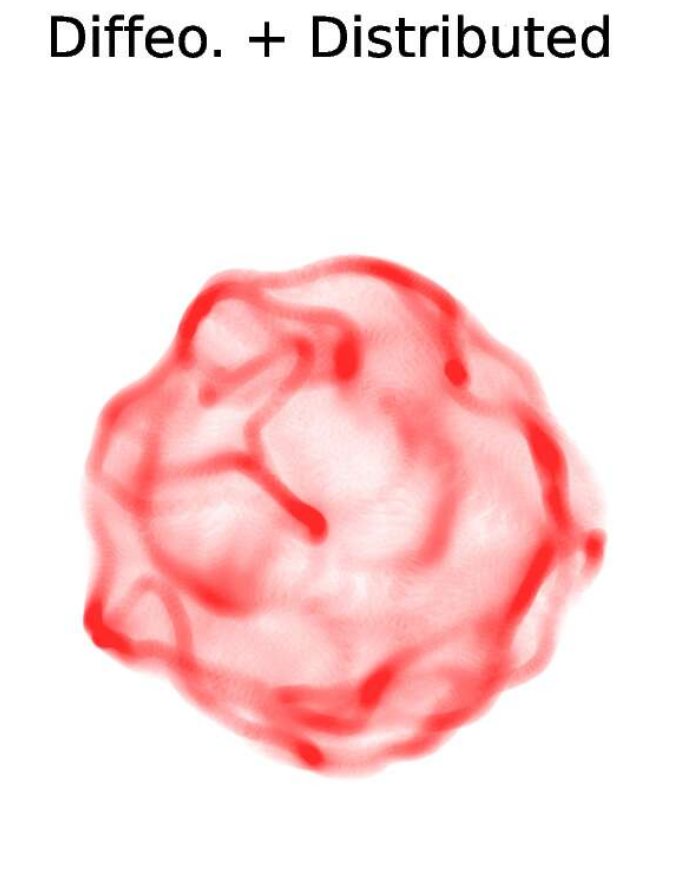}
    \includegraphics[height=5cm]{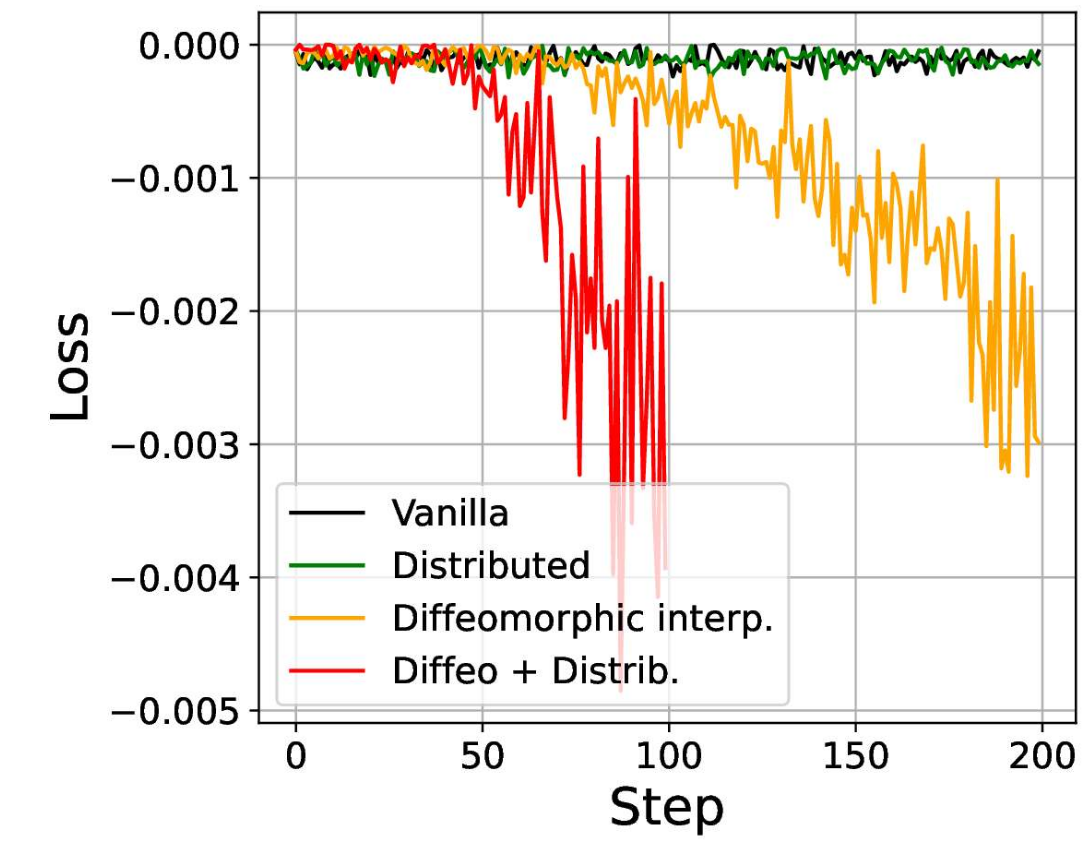}
    \caption{(Top row, from left to right) The initial point cloud $X$, the output of vanilla gradient descent with subsampling after $200$ epochs (barely any changes), the output of distributed gradient descent (averaged over 10 samples) after 200 epochs (barely any changes), the output of diffeomorphic gradient descent (using vanilla gradient) after 200 epochs, and the output of diffeomorphic gradient descent (using distributed gradients averaged over 10 samples) after 200 epochs. (Bottom row) Evolution of the losses across iterations.}
    \label{fig:poc_subsampling2}
\end{figure}

\subsection{Topological Autoencoder}

In this section, we reproduce the experimental setting of topological autoencoders described in \cite{Carriere2021a} (see also \cref{subsec:app_dr} and \cite{moor2019topological}). 
More precisely, we generated a point cloud in $\bR^3$ comprised of two nested circles (see \cref{fig:3D_visualize_topoae}) and embedded it non-linearly in $\bR^9$ by converting each point $p=(x,y,z)$ into the exponential of the $3 \times 3$ anti-symmetric matrix whose coefficients are $x,y$, and $z$.
The converted point cloud is denoted by $X \subset \bR^9$. 
The Vietoris--Rips persistence diagram of $X$ in homology dimension $1$ exhibits two loops (as a smooth transformation of the nested circles depicted in \cref{fig:3D_visualize_topoae}). 
The goal of this experiment is to propose an embedding of $X$ in $\bR^2$ using an autoencoder that would preserve topology, i.e.,~that would contain (exactly) two underlying loops. 

We use an autoencoder whose encoder $f_\theta \colon \bR^9 \to \bR^2$ and decoder $g_\theta \colon \bR^2 \to \bR^9$ are both made of three fully-connected layers with width $32$ and ReLU activations as well as batch normalization. 
We initialized the autoencoder by pre-training with the usual reconstruction loss (900 epochs). 
See the leftmost part of \cref{fig:topoae_embedding}. 
Then, we trained the autoencoder using the sum of the reconstruction loss and the topological loss computed as 
\begin{align}
    \scL(\theta) = \|X - (g_\theta \circ f_\theta)(X)\|_2 + \lambda \cdot \FG_2\left(\DgmRips^{(1)}(X), \DgmRips^{(1)}(f_\theta(X))\right),
\end{align}
i.e., the second diagram distance between the Vietoris-Rips persistence diagrams (in homology dimension 1) of the original space ($\bR^9$) and the latent space ($\bR^2$).
For the optimization of the loss function with the topological loss, we used the different topological gradients described in \cref{sec:gradient_descent}. 
See \cref{fig:topoae_embedding} for the resulting point clouds in $\bR^2$ and \cref{tab:execution_time} for the execution times over 100 epochs for each gradient. 

\begin{figure}[htbp]
	\centering
	\begin{minipage}{0.59\textwidth}
		\centering
	    \includegraphics[width=0.4\linewidth]{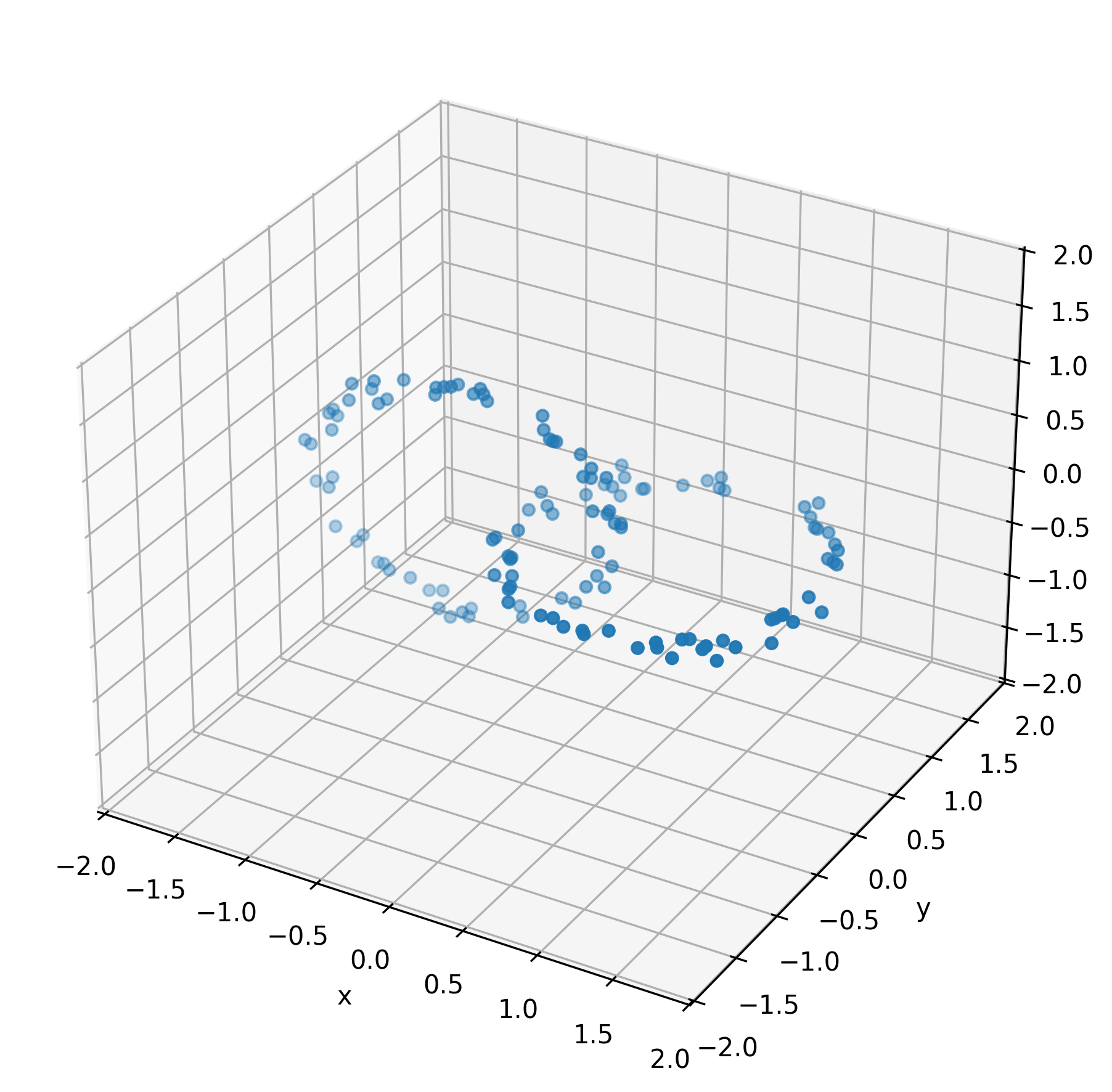}
        \caption{The original point cloud in $\bR^3$}
        \label{fig:3D_visualize_topoae}
	\end{minipage}
	\hfill
	\begin{minipage}{0.4\textwidth}
	\centering
	\captionof{table}{Running time over 100 steps.}
	\label{tab:execution_time}
	\begin{tabular}{|c|c|}
	\hline
	Method & Time (s)\\
	\hline
	Vanilla            
	& $3.39 \pm 0.06$ \\
	Stratified 
	& $4.54 \pm 0.6$ \\
	Big-step                    
	& $776.5 \pm 2.99$\\
	Continuation                
	& $1.25 \pm 0.06$ \\
	Diffeomorphic                      
	& $1.89 \pm 0.07$ \\
	\hline
	\end{tabular}
	\end{minipage}
\end{figure}

\begin{figure}[htbp]
  \centering

  \begin{minipage}[t]{0.25\textwidth}
    \centering
    \includegraphics[width=\linewidth]{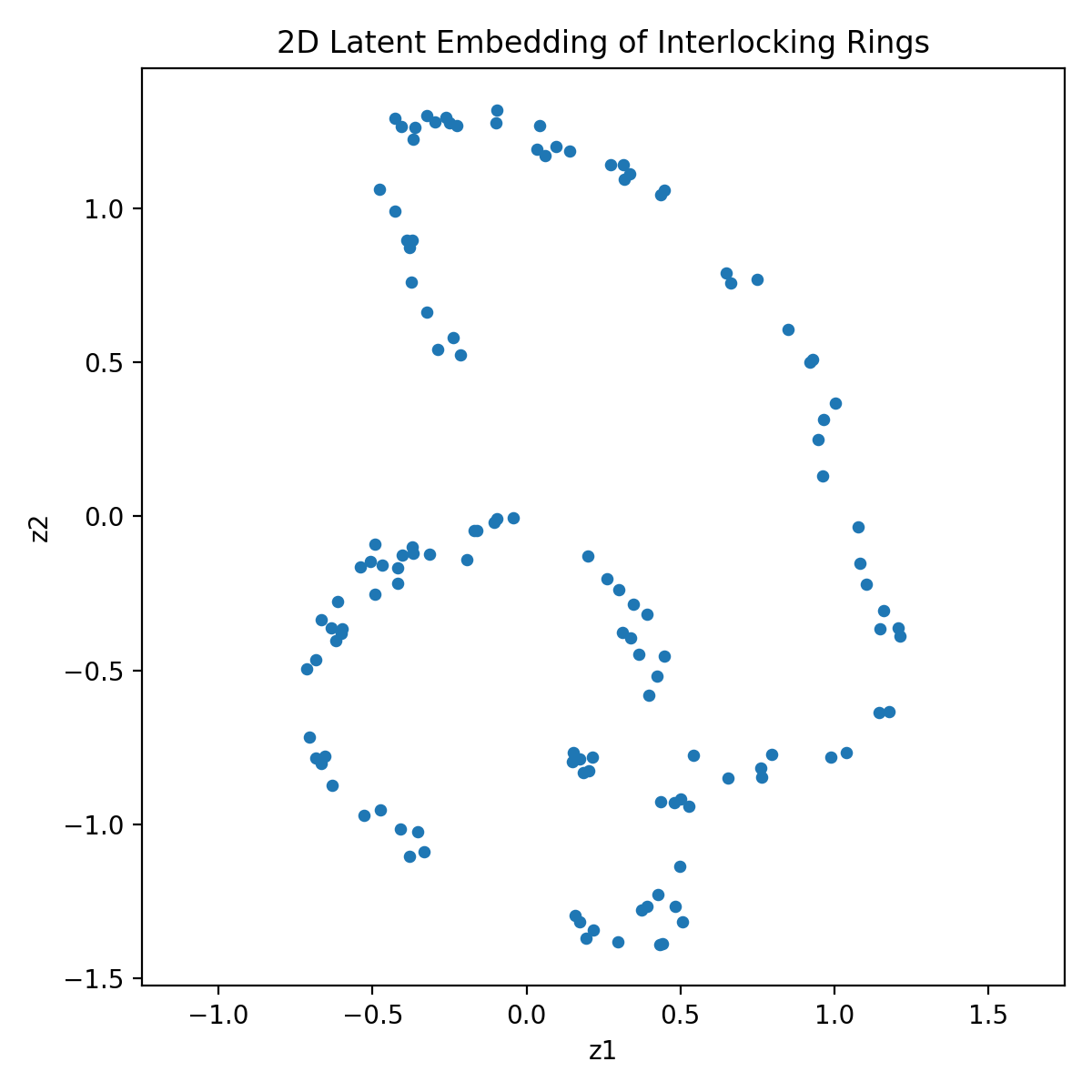}
  \end{minipage}\hfill
  \begin{minipage}[t]{0.25\textwidth}
    \centering
    \includegraphics[width=\linewidth]{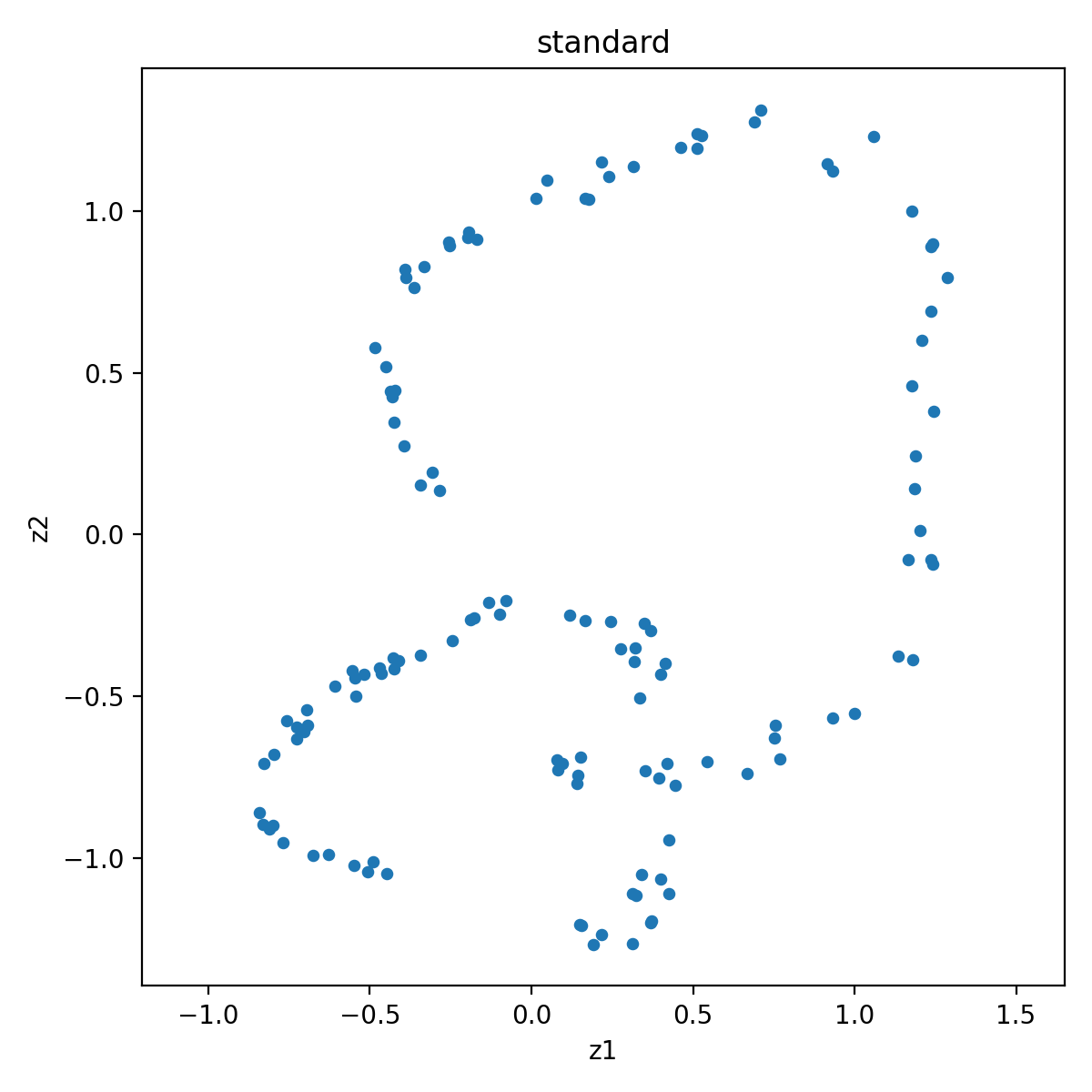}
  \end{minipage}\hfill
  \begin{minipage}[t]{0.25\textwidth}
    \centering
    \includegraphics[width=\linewidth]{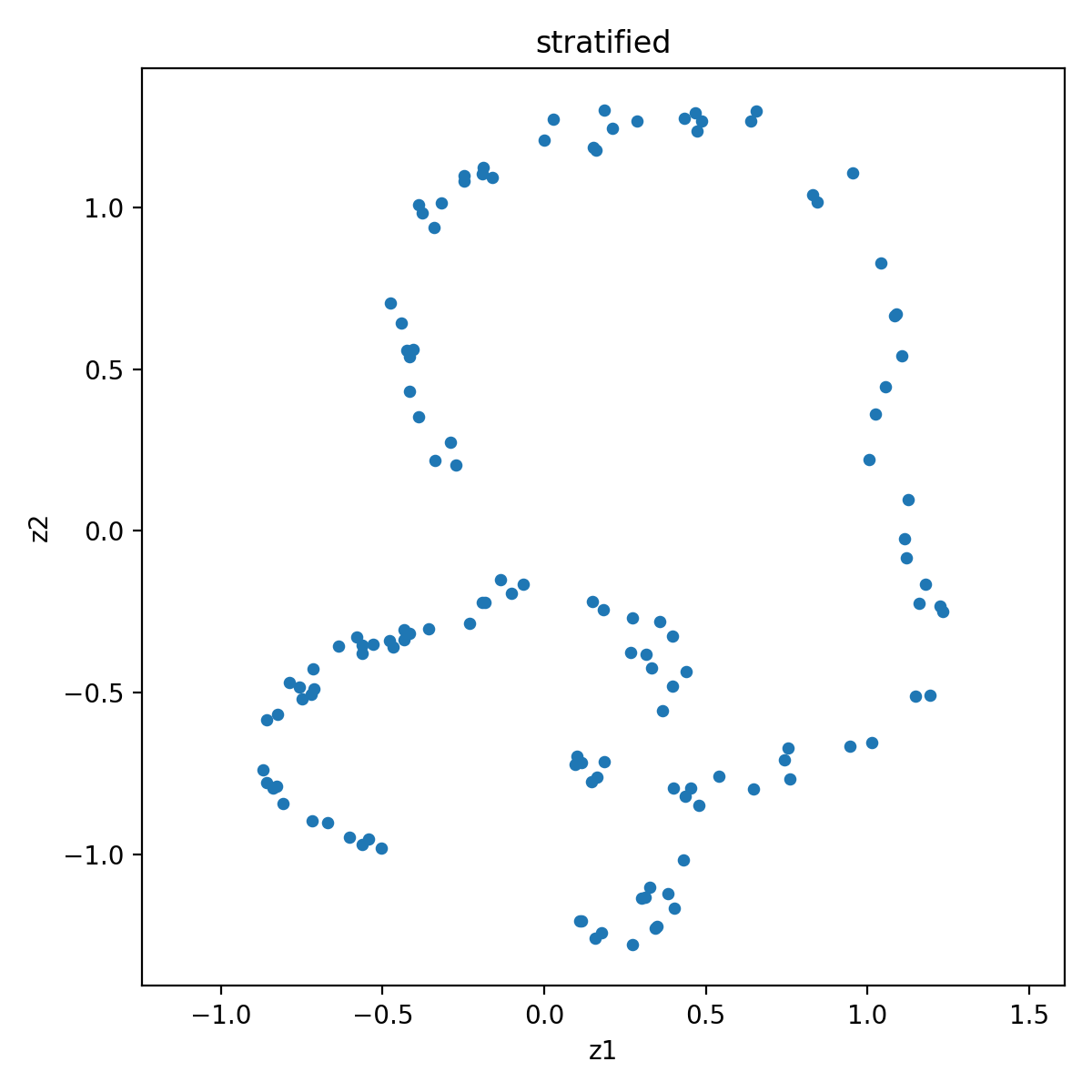}
  \end{minipage} \\

  \vspace{2mm} 

  \begin{minipage}[t]{0.25\textwidth}
    \centering
    \includegraphics[width=\linewidth]{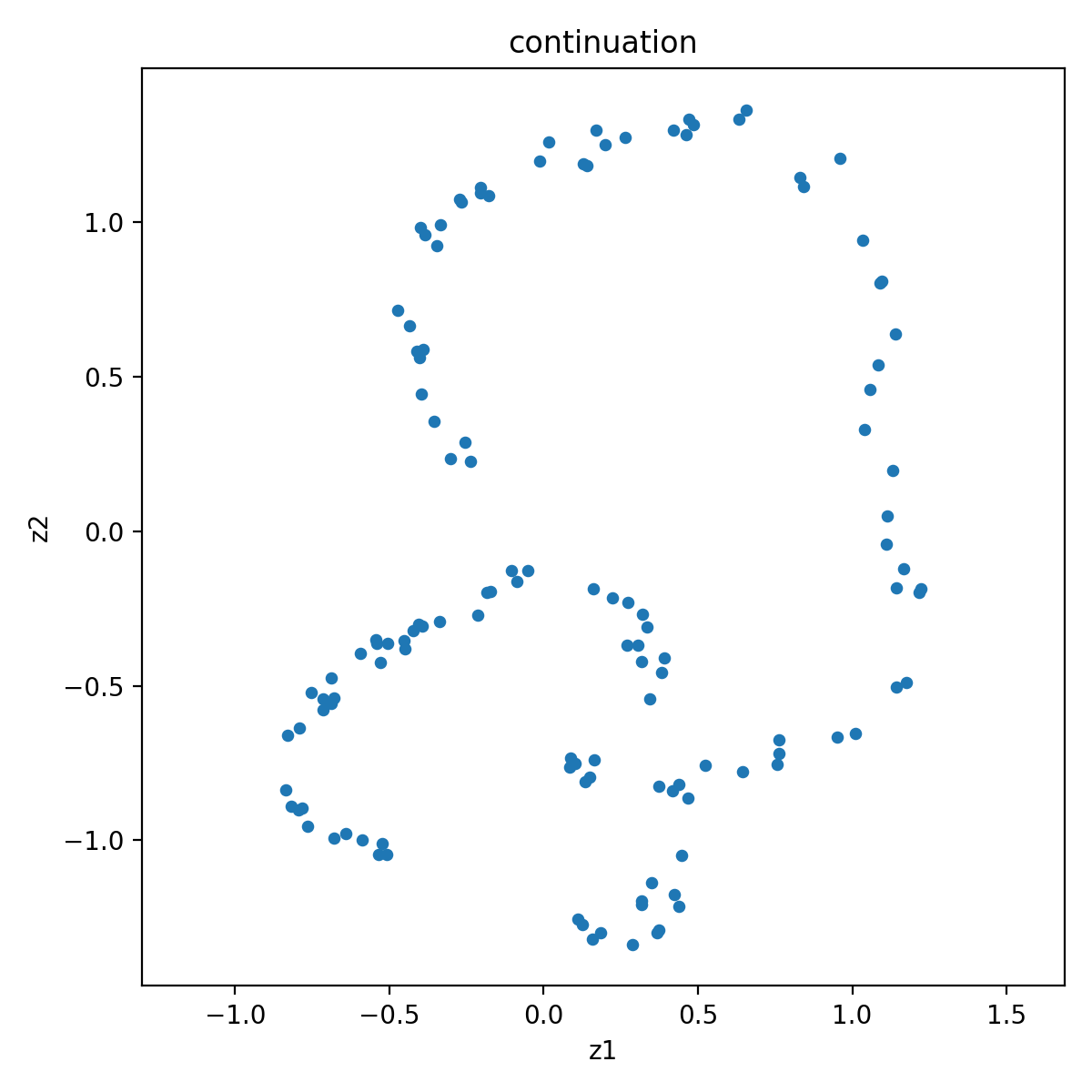}
  \end{minipage}\hfill
  \begin{minipage}[t]{0.25\textwidth}
    \centering
    \includegraphics[width=\linewidth]{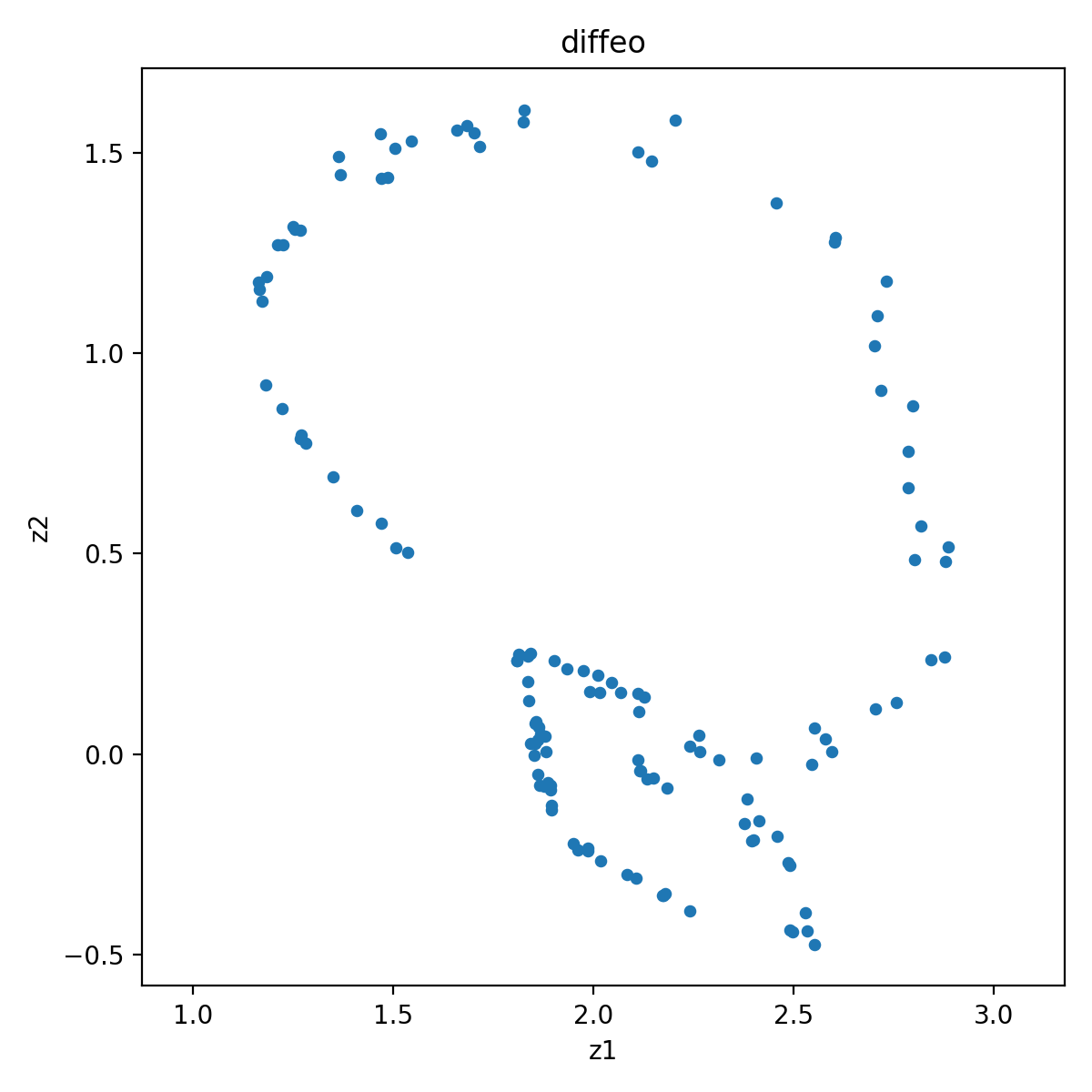}
  \end{minipage}\hfill
  \begin{minipage}[t]{0.25\textwidth}
    \centering
    \includegraphics[width=\linewidth]{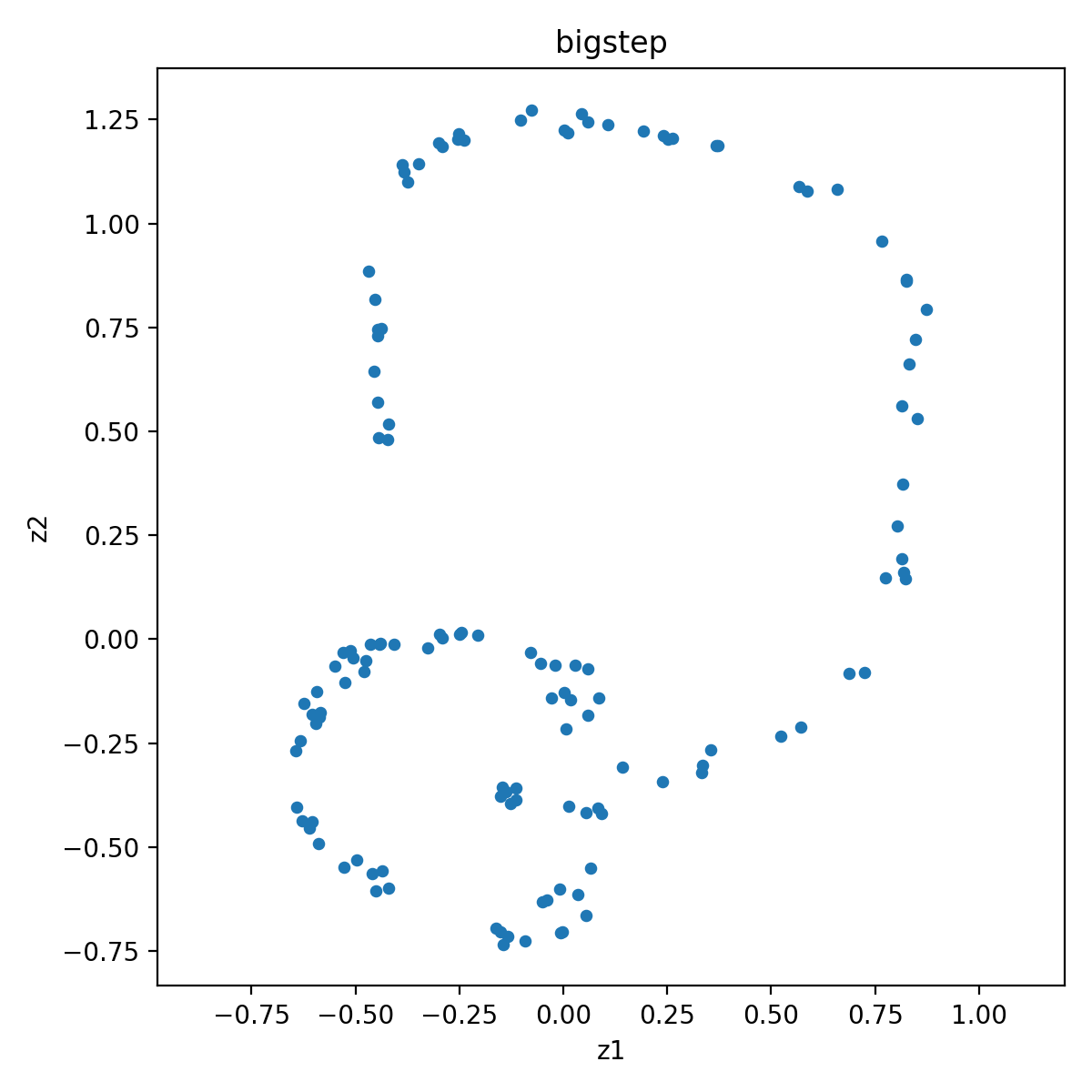}
  \end{minipage}

  \caption{Embedding of the point cloud $X$ in $\bR^2$ for the different topological gradient descents considered. (Top row, from left to right) No topological regularization, standard and stratified gradient descents. (Bottom row, from left to right) Continuation, diffeomorphic, and big-step gradient descents.}
  \label{fig:topoae_embedding}
\end{figure}

First, one may observe that even without adding an explicit topological penalty (top left of \cref{fig:topoae_embedding}), the latent space embedding already exhibits a fairly satisfying topology, which is a consequence of the implicit regularity of the encoder. The vanilla gradient descent, the stratified gradient descent and the continuation approach barely improve on this
initial try. 
As in the previous experiment, the diffeomorphic gradient descent (computed with vanilla gradients) and big-step gradient descent yield the best output (i.e., the lowest topological loss), the latter at the price of a higher computational time. 
This experiment once again suggests that mitigating sparsity of gradients in topological optimization tends to produce better results.

\section{Conclusion}

Persistence-based topological optimization is one of the most promising research directions in single-parameter persistent homology. 
The ability of learning topological descriptors, enforcing topological priors in scientific fields (computer vision, material science, computational biology, etc.) and to use topological quantities to regularize the training of machine learning models are exciting tracks worth of development.

It remains important to stress that this tool comes with important limitations if used naively: the poor computational efficiency inherent to the computation of persistence diagrams on large-scale problems (that typically requires resorting to subsampling), but also the sparsity of vanilla gradients described at the end of Section~\ref{subsec:vanilla} yielding to slow optimization schemes. 
This second aspect, specific to topological optimization, has attracted interest in the recent years and independent workarounds have been proposed, each with their own merits, though a more uniform and canonical formalism has yet to come. 

We hope that this survey along with the library we provide will give the necessary tools for mathematicians and computer scientists interested in working on this topic to get started easily. 

\paragraph{Open questions.} In addition to the refinement of the tools presented in this survey, we identified some research directions that would contribute to develop the field:
\begin{enumerate}
    \item \textbf{Creating topology.} Destroying topology in a given object (e.g.,~a point cloud) is a fairly easy task: one just has to take a loss function that penalizes points away from the diagonal, such as the total persistence (see \cref{ex:perstot}), and update the filtration values accordingly (typically reducing the death time, and possibly increasing the birth time as well). In contrast, creating topology if none is initially present in the filtration (i.e.,~the corresponding PD is empty) requires, at the level of PDs, to create a brand new point on the diagonal and then push it away from it. Unfortunately, the number of topological features (or the absence of) is generically a locally stable property of the filtration. That is, at the level of the filtration, no topology would be created when performing a gradient descent. This may explain why topological optimization has seen much more success when used as a regularization tool (i.e.,~limiting/destroying topological features) rather than a way to create topology (e.g.,~in generative models). 
    \item \textbf{Exploring non-gradient-based optimization.} While gradients are obviously the most important tool in optimization in general, the difficulties they pose in persistence-based optimization (non-smoothness, non-convexity, sparsity, etc) suggest that it may be worth seeking for other optimization routines in TDA. One appealing idea would be to rely on \emph{genetic algorithms}, where the filtration values would evolve randomly and be iteratively selected based on their quality (measured with respect to the considered loss function). This type of algorithm may in particular be a first attempt to tackle the question of creating topology presented above. 
    \item \textbf{Extensions to multiparameter persistence.} This survey focuses on single-parameter persistent homology, that is we restricted to filtrations valued in $\bR$. In contrast, multiparameter persistence consider functions valued in $\bR^d$. While being more general and appealing, this setting comes with huge theoretical and computational challenges, as no canonical counterparts of PDs can be defined \cite{carlsson2007theory,bauer2022generic}. Several descriptors have been defined in the literature (see for instance \cite{botnan2022signed,carriere2020multiparameter,chen2024emp,corbet2019kernel}). The question of involving these descriptors in the context of topological optimization has recently been considered in \cite{scoccola2024differentiability,carriere2024sparsification,loiseaux2024multipers} and is a natural track to follow. 
\end{enumerate}

\clearpage
\bibliographystyle{plain}
\bibliography{references}

\end{document}